\documentclass[12pt]{article}
\usepackage{amsthm,amsmath,amsfonts,amssymb,cases}
\usepackage{calc} 
\usepackage{comment}
\usepackage{enumerate}
\usepackage{tikz}
\usepackage{pgfplots}
\usepackage[shortlabels]{enumitem}
\usepackage{hyperref}
\usepackage{cleveref}
\usepackage{float}
\usepackage{xcolor}

\usepackage[style=alphabetic,giveninits=true,maxnames=4, backend=bibtex,doi=false,url=false,isbn=false,eprint=false]{biblatex}
\newcommand{\R}{{\mathord{\mathbb R}}}
\newcommand{\Z}{{\mathord{\mathbb Z}}}
\newcommand{\N}{{\mathord{\mathbb N}}}
\newcommand{\C}{{\mathord{\mathbb C}}}

\newcommand{\HS}{{\mathord{\mathbb H}}}
\newcommand{\E}{{\mathord{\bf E}}}

\newcommand{\muu}[1]{\begin{align*} #1 \end{align*}}
\newcommand{\muun}[1]{\begin{align} #1 \end{align}}

\newcommand{\SP}[1]{\left\langle #1  \right\rangle }
\renewcommand{\epsilon}{\varepsilon}

\newcommand{\vertiii}[1]{{\left\vert\kern-0.25ex\left\vert\kern-0.25ex\left\vert #1 
    \right\vert\kern-0.25ex\right\vert\kern-0.25ex\right\vert}}

\newcommand{\ben}{\begin{displaymath}}
\newcommand{\een}{\end{displaymath}}
\newcommand{\beqn}{\begin{equation}}
\newcommand{\eeqn}{\end{equation}}
\newcommand{\beqna}{\begin{eqnarray*}}
\newcommand{\eeqna}{\end{eqnarray*}}

\newcommand{\nn}{\nonumber} 

\def\tr{\operatorname{tr}}

\def\supp{\operatorname{supp}}

\newcommand{\sfrac}[2]{\textrm{\footnotesize $\frac{#1}{#2}$}}

\newcommand{\inn}[1]{\langle {#1} \rangle }

\newtheorem{lemma}{Lemma}
\newtheorem{theorem}[lemma]{Theorem}
\newtheorem{remark}[lemma]{Remark}
\newtheorem{proposition}[lemma]{Proposition}
\newtheorem{corollary}[lemma]{Corollary}
\newtheorem{definition}[lemma]{Definition}
\newtheorem{hyp}{Hypothesis}

\numberwithin{equation}{section}
\numberwithin{lemma}{section}

\makeatletter
\newsavebox{\@brx}
\newcommand{\llangle}[1][]{\savebox{\@brx}{\(\m@th{#1\langle}\)}%
  \mathopen{\copy\@brx\kern-0.5\wd\@brx\usebox{\@brx}}}
\newcommand{\rrangle}[1][]{\savebox{\@brx}{\(\m@th{#1\rangle}\)}%
  \mathclose{\copy\@brx\kern-0.5\wd\@brx\usebox{\@brx}}}
\makeatother

\newlength{\hght}
\newlength{\dpth}

\begin{document}
\title{On asymptotic  expansions   of the density of states   for  Poisson distributed random Schr\"odinger operators}
\author{\vspace{5pt} Jan Fischer$^1$\footnote{ E-mail: jan.fischer@uni-jena.de} \quad  David Hasler$^1$\footnote{
E-mail: david.hasler@uni-jena.de } \quad  Jannis Koberstein$^1$\footnote{E-mail: jannis.koberstein@uni-jena.de}
 \\
\vspace{-4pt} \small{$1.$ Department of Mathematics,
Friedrich Schiller  University  Jena} \\ \small{Jena, Germany } }
\date{}
\maketitle

\begin{abstract}
We study a random   Schr\"odinger operator with a  potential distributed according 
to a Poisson process. Asymptotic expansions for traces of resolvents  in the limit of small disorder are derived.  Explicit  
estimates for the expansion coefficients are given and  we show that their infinite volume limits  
are  finite as   the spectral parameter approaches the spectrum 
of the free Laplacian.  As an application we derive bounds on the integrated density of states. 
\end{abstract}

\section{Introduction}

Self-adjoint operators which are given  as the sum of a Laplacian and a  random multiplication operator  belong 
to a class of so-called random Schr\"odinger operators. 
These  operators  are used to model  properties of disordered media.
In case of a finite difference   Laplacian on a lattice the random operator
is often referred to as the   Anderson model.
The pioneering work of Anderson  \cite{Anderson.1958} led to an intensive investigation of random Schr\"odinger  operators. 
 In the  large disorder regime 
almost sure existence of pure point spectrum could be established with exponentially localized eigenfunctions \cite{FrohlichSpencer.1983,AizenmanMolchanov.1993}. The  small disorder regime is less well understood.

The  integrated density of states (IDS) for a random Schr\"odinger operator  is defined as the expected  number of eigenvalues 
per unit volume  in an energy interval.
This quantity  is of fundamental importance in condensed matter physics and determines 
physical properties of the system. 
The IDS  is  a priori defined  for a  finite box, and 
for infinitely extended models,  e.g. on $\R^d$ or $\Z^d$,    it  is   obtained by a   limiting process as the  side length of the finite box tends to infinity, cf. \cite{Pastur.1973}.
For random Schr\"odinger operators the existence of this limit has been shown for a large class of models 
\cite{KirschMartinelli.1982}.
Proving upper bounds on the IDS in terms of 
 the energy interval has attracted much attention, since these bounds, known as  Wegner estimates \cite{Wegner.1981}, 
 constitute an essential ingredient in the so-called  multiscale analysis proof for  localization of eigenfunctions \cite{FrohlichSpencer.1983}. 
In particular, 
upper bounds have been obtained which are uniform in the volume and  proportional to the length of the energy interval
as well for  lattices \cite{Wegner.1981} as  for continuous models, cf.  \cite{CombesHislopKlopp.2003,CombesHislopKlopp.2007} and  references in \cite{KirschMetzger.2007}.
These bounds  imply Lipschitz continuity of the IDS and    give  upper 
bounds which are inverse proportional to the  disorder strength. Whereas  this provides  a very good upper 
bound  at  large disorder,     for small disorder   the behaviour of the IDS  is not well understood.
In the small disorder regime  it is conjectured \cite{Schenker.2004} that  the  proportionality constant should be uniform 
in the disorder strength.  In fact, in the limit of small disorder it is natural to expect  the   IDS   to  converge  to that of the free Laplacian. In  \cite{Schenker.2004} a result in this direction was  obtained, 
where    H\"older continuity with   exponents less or equal to  $\frac{1}{2}$ was shown 
 for the IDS on a lattice. 
At least  in the special case of a Cauchy distribution on the  lattice the  IDS  can be explicitly 
calculated and is thus shown to satisfy   aforementioned conjecture 
\cite{Lloyd.1969,CarmonaLacroix.1990}.

 In this paper, we study the IDS for  a random Schr\"odinger operator  on $\R^d$  where the randomness, parameterized by $\lambda \geq 0$, is given by a 
Poisson distributed random potential \cite{KirschMartinelli.1982,KirschMartinelli.1982c,KirschMartinelli.1983,CombesHislop.1994}, which are used e.g.   to  model amorphous materials like glass or rubber.   
For these random Schr\"odinger operators the existence of the infinite volume limit of the IDS has been established, cf. 
 \cite{KirschMartinelli.1982}  and references therein.
Moreover,    there are Wegner estimates for Poisson distributed random potentials  \cite{CombesHislop.1994},
cf. references in \cite{KirschMetzger.2007}, which however  are not optimal in volume.

We investigate the  IDS   using   expectations of normalized traces of  resolvents,
which are in relation by means of the Borel transform  \cite{Teschl.2014}.
We are interested in the  small disorder regime, i.e. small $\lambda \geq 0$, and expand 
the  
   resolvent in a Neumann series. This yields  an expansion of  the expectation of the normalized trace 
of the resolvent.  The first part of the paper  is concerned with  properties  of these expansion coefficients.
We show that their  infinite volume limit 
exists
and that their  boundary values exist as the spectral parameter approaches 
any point of the positive real axis.
For the latter result we use  the method of analytic dilation,  cf.  \cite{CombesThomas.1973,SR4} and references therein. 
In the second part of the paper we derive an error estimate for the difference between the expectation of the 
 normalized traces of resolvents and the sum of the first $N$  expansion coefficients.
In \cite{MagnenPoirotRivasseau.1996,MagnenPoirotRivasseau.1998,Poirot.1999} related results about expansions of resolvents up to  finite order have been studied.
In these works  a random Schr\"odinger operator in two  dimensions  with a random potential given by a Gaussian random field was studied.  In  \cite{elgart}  the density of states for the Anderson Model was estimated by similar method in a spectral regime where localization 
occurs, that is, in contrast to our paper, a regime outside of the spectrum of the free Laplacian.
 Part of the methods  which we use  are inspired by 
 expansion techniques introduced in \cite{elgart,ErdosSalmhoferYau.2007a,ErdosSalmhoferYau.2007b,ErdosSalmhoferYau.2008a}.

We then combine these two  results. As a first application we show that there exists 
 an asymptotic expansion of the integrated density of states as $\lambda \downarrow 0$. 
The zeroth order term in this asymptotic expansion is exactly the integrated density of states of the free Laplacian. 
In this context we want to mention related  results  in the literature.
In  \cite{KLOPP1997267} an asymptotic expansion  
of the IDS for a Poisson distributed random potential was 
studied in the limit where intensity of Poisson process vanishes. In contrast to our result 
the expansion in  \cite{KLOPP1997267} was studied in the sense of distributions. 
Furthermore, in \cite{wangdos},  an  asymptotic expansion of  the IDS was studied for a magnetic Schr\"odinger operator
in the limit of a  large magnetic field.

As a second application we estimate  the change  of the IDS  as the random potential is added to the free Laplacian, i.e., the difference between $\lambda > 0$
and $\lambda = 0$.  In contrast to standard Wegner estimates this provides in addition a  lower
bound for the IDS. Furthermore, we obtain  upper bounds on the IDS in the small disorder regime,  which 
 are  optimal for small $\lambda \geq 0$, and thus complementary to  Wegner-type estimates,  which are 
well in the large $\lambda$  regime. In fact, combining our result  with a Wegner estimate
 we can infer  H\"older continuity with any exponent less than $\frac{2}{3}$ uniform in the disorder strength. This exponent is an improvement 
 compared to the  exponent $\frac{1}{2}$, which was  obtained in \cite{Schenker.2004} for a
random Schr\"odinger operators on the lattice.

Let us   mention  related results about random Schr\"odinger operators at  small disorder.
Expectations of products  of  resolvents were 
studied in the context of time evolution  \cite{ErdosSalmhoferYau.2007a,ErdosSalmhoferYau.2007b,ErdosSalmhoferYau.2008a,Black}, 
on the Bethe lattice  \cite{AcostaKlein.1992,Klein.1998,FroeseHaslerSpitzer.2006,AizenmanSimsWarzel.2006,FroeseHaslerSpitzer.2007,AizenmanWarzel.2013}, 
and  for quantum trees \cite{AnantharamIngremeauSabriWinn.2021}. Moreover, we want to mention related deterministic results  about  distributions of eigenvalues for perturbations of periodic Schr\"odinger operators \cite{FeldmannKnoerrerTrubowitz.1990,FeldmannKnoerrerTrubowitz.1991}

Let us give an  outline of this paper. In Section \ref{Model} we introduce the model and state the main results 
about expectations of matrix elements of resolvents. 
In Section \ref{asympsec} we give a first application of the main results and derive an asymptotic 
expansion for traces of resolvents. In Section  \ref{dosapprox} we give a second  application 
and derive upper and lower bounds for  the integrated density of states at small disorder. 
In Section \ref{PPL} we  recall formulas about   expectations of products of the random potential, which will be needed for the proofs of the main results.
In Section  \ref{sec:prooasymexp} we derive elementary integral formulas for the expansion 
coefficients for the normalized traces of resolvents.
In Section \ref{sefurtherintformuilas} we evaluate the integral formulas.
In Section \ref{sec:boundsonexpcoeff}  we use the evaluated integral formulas to derive 
bounds, which will be needed to establish the infinite volume limit shown in  Section \ref{secnine}.
In Section \ref{sec:boundaryvalues} we show the existence of boundary values for
the infinite volume limits of the  expansion coefficients, which establishes the first main result. 
In Section  \ref{errorestimate}   we derive  an  error estimate used to show the  second main result.  
In the appendix we collect a few elementary estimates and in Appendix 
\ref{appendixDboundonres}  we show the trace class properties of  the resolvent.

\section{Model and Statements of Main Results} 

\label{Model}
In this section we  introduce  the model and state the main results about the expectation 
of the resolvent.  We note that the definition of the model follows  the one given in \cite{ErdosSalmhoferYau.2008a} closely.
We consider   a finite box of size $L$,  with $L  > 0$, and define the  box .
 $\Lambda_L := \left[ - \frac{1}{2} L , \frac{1}{2} L \right)^d \subset \R^d$. 
Let $\inn{ \cdot , \cdot }_{L^2(\Lambda_L)}$ and $\| \cdot \|_{L^2(\Lambda_L)}$ denote the canonical scalar product and the norm of the Hilbert space  $L^2(\Lambda_L)$. If it is clear from the 
context what the inner product or  the norm is,   we shall occasionally drop the subscript $L^2(\Lambda_L)$.
 The kinetic energy will be given by the periodic Laplacian. To introduce this operator it is convenient to 
work in terms of Fourier series. Let  
$\Lambda_L^*  = ( \frac{1}{L} \Z  )^d=\{ (k_1 ,  \ldots ,  k_d) : \forall j=1, \ldots , d,  \exists m_j \in \Z,: k_j =\frac{m_j}{L} \}$ denote  
the so-called dual lattice. We introduce the notation 
\begin{equation} \label{defofdiscretefouriersum} 
\int_{\Lambda_L^*} f(p) dp  := \frac{1}{|\Lambda_L |} \sum_{p \in \Lambda_L^*} f(p) . 
\end{equation} 
The sum  $\int_{\Lambda_L^*}f(p)dp$ can be interpreted as a Riemann-sum, which  converges to the integral $\int f(p) dp $ 
as   $L \to \infty$, provided $f$ has  sufficient  decay at infinity and sufficient regularity.  
For any $f \in L^1(\Lambda_L)$ we define the Fourier series 
$$
\hat{f}(p) = \int_{\Lambda_L} e^{ - 2 \pi i p \cdot x}  f(x) dx \text{ for } p \in \Lambda_L^*  . 
$$
  For $g    \in \ell^1(\Lambda_L^*)$ we define 
$$
\check{g}(x) = \int_{\Lambda^*_L} e^{  2 \pi i  p \cdot x } g(p) dp  = \frac{1}{|\Lambda_L |} \sum_{p \in \Lambda_L^*} e^{  2 \pi i  p \cdot x } g(p) \text{ for }  x \in \Lambda_L. 
$$
Note that 
$\check{(\cdot)}$ extends uniquely to a continuous linear map  on $ \ell^2(\Lambda_L^*)$.  
We shall denote this extension again by the same symbol. This extension maps $\ell^2(\Lambda_L^*)$ unitarily to $L^2(\Lambda_L)$
and is the inverse of $\hat{(\cdot)}|_{L^2(\Lambda_L)}$, see for example \cite[Theorem 8.20]{folland}.  From a physics background one might be familiar with the notation where  $L^2(\Lambda_L)$ represents the position space,  while $\ell^2(\Lambda_L^*)$ would be called momentum space,  or Fourier space.
For $ p \in \Lambda_L^*$  and $x \in \Lambda_L$ we define 
\begin{align} \label{eq:deofonb}
 \varphi_p(x) ={|  \Lambda_L|}^{-1/2} e^{ i 2 \pi p \cdot x } .
\end{align}
Observe that 
$
\{ \varphi_p :p \in \Lambda_L^* \} 
$
is an orthonormal Basis (ONB) of $L^2(\Lambda_L)$ \cite[Theorem 8.20]{folland} and that 
$\hat{f}(p) = | \Lambda_L|^{1/2}  \inn{  \varphi_p , f}_{L^2(\Lambda_L)} $.

When defining the Laplacian restricted to a finite box,  we have to choose boundary condition to describe how our operator behaves at the border of the box.  
We are presented with the choice between 
Dirichlet, Neumann and periodic boundary conditions.
Note that in  the  sense of forms,  periodic boundary conditions lie  between Dirichlet and Neumann boundary conditions,  where the latter two are better suited,  when working in position space.  
While we will be starting in position space, the majority of this work
will be performed in momentum space,  where  using periodic boundary conditions is technically convenient.  
Therefore,  we are going to introduce the Laplacian with periodic boundary conditions on $L^2(\Lambda_L)$ by means of Fourier series.   
In this paper we shall adopt the physics convention that for a vector $a =(a_1,...,a_d)$ we use the notation $a^2 := |a|^2$ where $|a| :=  \sqrt{ \sum_{j=1}^d |a_j|^2}$.  
Using this notation,  we define the energy function
\begin{equation}
\label{energiyfunction}
\nu : \R^d  \to \R_+,  \quad p \mapsto \nu(p) := \frac{1}{2} p^2,   
\end{equation} 
which we can now use to
 define $-\Delta_L$ as the linear operator with domain 
\begin{align*}
D(-\Delta_L) = \{ f \in L^2(\Lambda_L) : \nu \hat{f} \in \ell^2(\Lambda_L^*) \},
\end{align*} 
and the mapping rule
\begin{equation} \label{eq:deoflaplac0} 
 -\Delta_L : D(-\Delta_L)  \to L^2(\Lambda_L), \quad f \mapsto -\Delta_L f   =  (2 \pi )^2 ( \nu  \hat{f})^\vee . 
\end{equation} 
Observe that  $-\Delta_L$ is selfadjoint, since it is  unitary equivalent to  a multiplication operator by a real-valued function,  and that we have the identity 
\begin{equation} \label{eq:deoflaplac} 
\left[ - \frac{1}{2 (2 \pi)^2}   \Delta_L  f \right]^\wedge(p)  = \nu (p) \hat{f}(p)  .
\end{equation} 
By  
\begin{equation} \label{defofmainH}  
 H_{L} := H_{\lambda,L} := - \frac{\hbar^2}{ 2 m } \Delta_L + \lambda V_{L}
\end{equation} 
we denote a random Schr\"odinger operator acting in $L^2(\Lambda_L)$ with   a random potential
$
V_{L} = V_{L,\omega}(x)
$, defined below,  and a  coupling constant $\lambda \geq 0$.   As in  \cite{ErdosSalmhoferYau.2008a}  we choose units for the mass $m$ and Planck's constant  so that $\frac{\hbar^2}{2m} = [2 (2\pi)^2]^{-1}$.
For a function $h : \R^d \to \C$ we denote by $h_\#$ the $L$--periodic extension of $h|_{\Lambda_L}$ to $\R^d$. 
The potential is given by
\begin{align} \label{defofpot101} 
V_{L,\omega}(x) := \int_{\Lambda_L} B_\#(x - y) d \mu_{L,\omega}(y),
\end{align} 
where $B$,  having the physical interpretation as   a single-site potential profile,   is  assumed  to be a  real-valued Schwartz function on $\R^d$.
Moreover, we assume that either $B$  has   compact support
or that  $B$  is  symmetric with respect to the reflections of the coordinate axis,  i.e.  for  $j=1,...,d$
\muun{ \label{SC}
\mathfrak{S}_j:  (x_1,...,x_j,...,x_d) &\mapsto (x_1,...,-x_j,...,x_d) 
\\ 
B \circ \mathfrak{S}_j &= B. \nn
}  
\begin{remark} {\rm 
Each of the two conditions is mathematically convenient in the sense that they ensure sufficiently fast decay of the Fourier transform.  While the reflection symmetry condition is satisfied by rotationally invariant potentials, which occur naturally. }
\end{remark}
Furthermore,  $\mu_{L,\omega}$ is a Poisson point measure  on $\Lambda_L$ with homogeneous unit density and 
with independent identically distributed (i.i.d.) random masses.  
More precisely, for almost all events $\omega$, it consists of $M$ points $\{ y_{L,\gamma}(\omega) \in \Lambda_L : \gamma=1,2,...,M_L\}$, where $M_L = M_L(\omega)$ is a Poisson variable with expectation $|\Lambda_L|$, and $\{ y_{L,\gamma}(\omega) \}$ are i.i.d. random variables uniformly distributed on $\Lambda_L$.   Both are independent of  the random i.i.d. weights  $\{ v_\gamma : \gamma=1,2,... \}$ and 
the random measure is given by 
$$
\mu_{L,\omega} = \sum_{\gamma=1}^{M_L(\omega)} v_\gamma(\omega) \delta_{y_{L,\gamma}(\omega)} , 
$$ 
where $\delta_y$ denotes the Dirac mass at  the point $y$.
Note that the case where $M_L(\omega )=0$ corresponds to a vanishing potential.
Furthermore we denote the common  distribution weights  $\{ v_\gamma \}$  by  ${\bf P}_v$ and assume that the moments 
\muun {\label{moments}
m_k :=  {\bf E}_v v_\gamma^k
}
 satisfy
\begin{equation} \label{2.4} 
 m_{k } < \infty     , \quad  \text{ for all }  k \in \N.
\end{equation} 
For some of the results we will use that the first moment $m_1 = {\bf E}_v v_\gamma$ vanishes.
Note that for convenience the notation  will not reflect the   dependence on the specific  choice of the random distribution ${\bf P}_v$. 
Observe that we can write   \eqref{defofpot101} as  
\begin{equation}\label{bumpnotationres} 
 V_{L,\omega}(x) =  \sum_{\gamma=1}^M V_{L,\gamma}(x)   \quad \text{ with } \quad V_{L,\gamma}(x) := v_\gamma B_\#(x-y_{L,\gamma}) .
\end{equation} 
The expectation with respect to the joint measure of $\{ M , y_{L,\gamma}, v_\gamma\}$ is denoted by ${\bf E}_L$.
Sometimes we will use the notation  
\begin{equation}\label{esy:3.24} 
{\bf E}_L = {\bf E}_M {\bf E}_{y_L}^{\otimes M } {\bf E}_v^{\otimes M } 
\end{equation} 
referring to the expectation of $M$, $\{ y_{L,\gamma} \}$ and $\{ v_\gamma \}$ separately. 
In particular, ${\bf E}_y^{\otimes M }$ stands for the normalized integral 
\begin{equation}\label{esy:3.240} 
\frac{1}{|\Lambda_L|^M} \int_{(\Lambda_L)^M} dy_1 \cdots dy_M . 
 \end{equation} 
If it is clear from the context over which probability measure the expectation is taken, we shall use the short hand notation ${\bf E}$  for ${\bf E}_L,  \, {\bf E}_M,  \, {\bf E}_{y_L}$, and  ${\bf E}_v$, respectively.
Since the potential is almost surely bounded it follows  by standard perturbation theorems, e.g. the Kato-Rellich theorem \cite{SR2},  that   the operator \eqref{defofmainH}  is almost 
surely self-adjoint  for all $\lambda \geq 0$.

The main object of interest is the expectation of the 
trace of the imaginary part of the  resolvent of $H_L$ for  a spectral parameter in  $\C \setminus \R$. 
We will show in the appendix that this quantity 
indeed exists for $d \leq 3$. This is the content of the following lemma. 

Let $\mathcal{B} (L^2(\Lambda_L))$ be the space of bounded linear operators.  We define the space of trace class operators as
\muu{
\mathcal{B}_1 (L^2(\Lambda_L)) := \{ A \in  \mathcal{B} (L^2(\Lambda_L)) : {\rm tr} (|A|) < \infty \},
}
where ${\rm tr}$ denotes the trace and  $|A|= \sqrt{A^*A}$.
\begin{lemma} \label{existenceoftracebaby}  Let $d \leq 3$. Then ${\rm Im} (H_L - z)^{-1}$ is trace class for 
$z \in  \C \setminus \R$ and any $L \geq 1$
almost surely.
Furthermore, for any $\lambda \geq 0$ the expression 
\begin{align}\label{GFct-2} 
|\Lambda_L|^{-1} {\bf E}_L {\rm tr} {\rm Im} (H_L - z)^{-1} 
\end{align} 
is uniformly bounded in $L \geq 1$. 
\end{lemma} 
\begin{proof} This is an immediate consequence of  Proposition  \ref{uniformboundontrace}.
\end{proof} 

We  study the quantity \eqref{GFct-2} as $L \to \infty$,  as $z$ approaches the real axis, and as $\lambda$ 
becomes small. 
For this   we will  use a Neumann-Type expansion to express the interacting  resolvent as the sum of products of powers of the potential and the resolvent of the periodic Laplacian, $\Delta_L$. This  is the content of   Lemma \ref{lem:resexp} below. 
For notational compactness   we  shall denote   the resolvent of $-\Delta_L$  by 
$$
R_L(z ) :=  \left(- \frac{\hbar^2 }{2 m} \Delta_L - z \right)^{-1} , 
$$
where $z \in \C \setminus [0,\infty)$.
Note that    we use a   notation for the   resolvent of the free Laplacian, which   one might expect  for the interacting  one.  
Iterating the second resolvent identity \ref{ResId}  yields the following lemma.
\begin{lemma}
\label{lem:resexp} 
For $z \in \C \setminus \R$, $L > 0$, and $n \in \N$ we have
\begin{align} \label{eq:resexp} 
( H_L  - z)^{-1} & = \sum_{j=0}^n R_L(z) [ \lambda V_L R_L(z)  ]^j +  [R_L(z)  \lambda V_L  ]^{n+1}  ( H_L  - z)^{-1} . 
\end{align}
\end{lemma}
\begin{proof}
This follows directly from Lemma \ref{resolventenformel_interiert}  with $A=- \frac{\hbar}{2m}\Delta_L - z$ and $B=\lambda V_L$, since  by self-adjointness $\sigma (H_L) \subset \R$, see for example  \cite{SR1}.
\end{proof}
To determine the integrated density of states we will work with  the trace of  $R_L(z) [  V_L  R_L(z)  ]^n$. 
For a set $\Omega \subset \R^d$ we shall denote the trace of trace class operators in $L^2(\Omega)$ by $ {\rm tr}_{L^2(\Omega)}$.
For notational compactness  we shall write ${\rm tr}$ for  $ {\rm tr}_{L^2(\Lambda_L)}$.
We define the normed trace $\widetilde{{\rm tr}}  : \mathcal{B}_1 (L^2 (\Lambda_L)) \to \C $ as
\muu{
A \mapsto  \widetilde{{\rm tr}} (A) : =|\Lambda_L|^{-1} {\rm tr} (A).
}
The next proposition shows that the trace indeed exists for $d \leq 3$. 

\begin{proposition}\label{thm:dosasy0-0-0}      Let  $d \leq 3$.
Let $z \in \C \setminus [0,\infty)$. Then for all $n \in \N$ the  operators $   |R_L(z) [  V_L  R_L(z)  ]^n|$ 
are almost surely trace class  in $L^2(\Lambda_L)$ and the expectation of the trace is bounded. 
Furthermore,  ${\rm Im} R_L(z)$ is trace class
and 
$ \widetilde{{\rm tr}} {\rm Im} R_L(z) = \int_{\Lambda_L^*}  \frac{{\rm Im} z }{|\nu(p) - z |^2} dp $.
\end{proposition}
\begin{proof}[Proof of \ref{thm:dosasy0-0-0}]  The proposition follows directly from the Lemmas 
 \ref{finitetracesumm},  \ref{HilbertSchmidt},   \ref{fourierspacekernel}, and \ref{boundonexppotenyt}.
 \end{proof}

We  can now study the following expressions for $z\in \C \setminus (0,\infty)$. 
Thus  for   $z \in \C \setminus [0,\infty)$,  $\psi_1, \psi_2 \in L^2(\Lambda_L)$,  and  $n \in \N_0$ we define  
 \begin{align} \label{defoffinT} 
  T_{n,L}[z;\psi_1,\psi_2] 
& :=  {\bf E}_L \inn{ \psi_{1},    R_L(z) [  V_L  R_L(z)  ]^n \psi_{2} }  . 
\end{align}
By identifying the Hilbert space with its dual we define by  $ T_{n,L}[z]$ the unique linear transformation 
such that $ T_{n,L}[z;\psi_1,\psi_2]  = \langle  \psi_1, T_{n,L}[z] \psi_2 \rangle$ for all 
$\psi_1, \psi_2 \in L^2(\Lambda_L)$. 
In other words 
 \begin{align} \label{defoffinT} 
  T_{n,L}[z] 
& =  {\bf E}_L   ( R_L(z) [  V_L  R_L(z)  ]^n )  .
\end{align}
By  Proposition \ref{thm:dosasy0-0-0}   we can define  for $d \leq 3$ the normalized trace of the 
expectation of  $R_L(z) [  V_L  R_L(z)  ]^n$ for $n \in \N$ 
 \begin{align} \label{defoffinT} 
  \widetilde{{\rm tr}}   T_{n,L}[z]   =  {\bf E}_L \widetilde{{\rm tr}}  R_L(z) [  V_L  R_L(z)  ]^n.
\end{align}
Observe that by  Proposition  \ref{thm:dosasy0-0-0} we can interchange expectation and trace summation and  write 
for $n \in \N$ 
 \begin{align} \label{defoffinT-0} 
&  {\bf E}_L  \widetilde{{\rm tr}} R_L(z) [  V_L  R_L(z)  ]^n \\
& = 
{\rm tr} \widetilde{T}_{n,L}[z]  =    |\Lambda_L|^{-1}  \sum_{\Lambda_L^*}   T_{n,L}[z;\varphi_p,\varphi_p] = \int_{\Lambda_L^*}   T_{n,L}[z;\varphi_p,\varphi_p]  dp .  \nonumber 
\end{align}
We note that  the term $n=0$ requires some care, since for $d \leq 3$ the   trace only exists for the imaginary part.
By means of the definition and   Proposition  \ref{thm:dosasy0-0-0} we can write  
\begin{align} \label{zerotraceimterm} 
\tilde{{\rm tr}} {\rm Im}{T}_{0,L}[z]    = \widetilde{ {\rm tr}} {\rm Im} R_L(z)  
  = \int_{\Lambda_L^*}  \frac{{\rm Im} z }{|\nu(p) - z |^2} dp .
\end{align}

We are now going to work towards the next theorem, Theorem \ref{lem:formulaforexpcoeff},  which will express   \eqref{defoffinT}  in momentum space and evaluate the expectation over the resulting integral kernels.
For this, we  introduce  the discrete delta function  in momentum space  for $u \in \Lambda_L^*$
\begin{equation}\label{delta-0} 
\delta_*(u) = 1_{\{0\}}(u) 
\end{equation} 
and a  normalized discrete delta function 
\begin{equation} \label{deltaL}  
\delta_{*,L}( u  ) = |\Lambda_L|  \delta_*(u),
\end{equation} 
where we used  the notation that for a set $A$ we write $1_A(x) = 1$, if $x \in A$, and $1_A(x) = 0$, if $x \notin A$.  \\

We now need to consider  the expectation of the product of the   random variable  $V_{L}$ i.e.
$
 \textbf{E}_v \left[ \prod_{j=1}^n v_{\gamma_j}\right]
$,  for they will appear from \eqref{defoffinT}.
This topic will be discussed in more detail in Section  
\ref{pairing}. The short idea is that we can use independence of the variables as long as they are different, while for the ones that are the same,  we count how often each of the variables occurs: Multiple appearances of the same random variable will result in a higher moment of the respective variable after evaluating the expected value.

To express the result of this procedure  we introduce the following partition function.  It (in a sense) filters all possible partition and leave only one remaining,  which represents the configuration of the momenta.  Meaning that every set $a$ in the partition $A$ has to correspond  to exactly one site in the sense that for  all $ a\in A $ there is a unique $y_a \in \Lambda_L$ such that $y_\gamma = y_a$ for all $\gamma \in a$.
To do so let $\mathcal{A}_n$,  denote the set of partitions of the set $\{1,...,n\}$.  
For $A \in \mathcal{A}_n$ and $k= (k_1,...,k_{n+1})$ we define the partition function
\begin{align} \label{eq:defofdeltapi0}
 \mathcal{P}_{A,L} : (\Lambda_L^*)^{n+1} & \to \R,  \\  k  &\mapsto
\mathcal{P}_{A,L}(k)  :=  \prod_{a \in A} \left\{ m_{|a|}  \delta_{*, L} \left( \sum_{l \in a} (k_l - k_{l+1}) \right) 
  \prod_{l \in a}  \widehat{B}_\#(k_l - k_{l+1} ) \right\}. \nn
 \end{align}
Recall  that  $m_{|a|}$ is the $|a|$-th moment defined in \eqref{moments}.
To state the results we  will  define the infinite volume Fourier transform and inverse Fourier transform for $f \in L^1 (\R^d) $
\begin{align} \label{defoffourier} 
\hat{f}(k) = \int_{\R^d} f(x) e^{ - 2 \pi i k \cdot  x } dx  , \quad k \in \R^d , 
\end{align}
and 
\begin{align*}
\check{f}(x) =\int_{\R^d}   f(k) e^{  2 \pi i k \cdot  x } dk, \quad x \in \R^d , 
\end{align*}
respectively.

For the next theorem we transition over to the momentum space, where we  give an explicit expression for the  expansion coefficients.  Remember that we write $k= (k_1, \ldots k_{n+1}) \in (\Lambda_L^*)^{n+1}$ and that the real energy function $\nu$ defined in \eqref{energiyfunction} is non-negative.

\begin{theorem} \label{lem:formulaforexpcoeff} 
For all $z \in \C \setminus [0,\infty)$ and $n \in \N$ we have 
\begin{align*}
&\widetilde{\tr}  T_{n,L}[z] \\
& =   \int_{(\Lambda_L^*)^{n+1}} \sum_{A  \in \mathcal{A}_n} \mathcal{P}_{A,L}(k_1, \ldots , k_{n+1})   
\prod_{j=1}^{n+1}  (\nu( k_j) -  z)^{-1} \delta_{*,L}(k_1 - k_{n+1}) d (k_1, \ldots , k_{n+1}).
\end{align*}
\end{theorem} 

The proof of Theorem \ref{lem:formulaforexpcoeff}  will be  given in Section \ref{sec:prooasymexp}. 
In fact, it follows directly from Theorem  \ref{lem:formulaforexpcoeff-1}. 
The next theorem shows that 
for each expansion coefficient  the infinite volume 
limit   exists and  it gives for each expansion coefficient an upper bound dependent on  the spectral parameter.

\begin{theorem}  \label{exinflim} For   $n \in \N_0$ and 
 $z \in \C \setminus [0,\infty)$ the   limit 
$ \lim_{L \to \infty} \widetilde{\rm tr} {\rm Im} T_{n,L}[z] $ exists in $\C$. 
\end{theorem} 

Theorem    \ref{exinflim}  follows from  Proposition \ref{propexofcinfDOS}  proven in Section \ref{sec:9}.  
It  allows us to  define the infinite volume 
 limit of the expansion coefficients  
\begin{equation} \label{eq:defoflimL} 
  \widetilde{\rm tr} {\rm Im}{T}_{n,\infty}[z ] := \lim_{L \to \infty}    \widetilde{\rm tr} {\rm Im}  {T}_{n,L}[z]  ,
\end{equation}
for all  $z \in \C \setminus [0,\infty)$. 
For the infinite volume limit we have the following bounds. 
\begin{proposition}    \label{exinflim-bound}
For $n \in \N_0$ and any compact set $I \subset  \R$ there exists a constant 
$K_{n,I}$ such that  the inequality 
\begin{align} \label{boundonTs} 
| \widetilde{\rm tr} {\rm Im} {T}_{n,L}[z]  | \leq \frac{ K_{n,I} ( 1 + {\rm dist}(z,[0,\infty))^{-2} ) }{ {\rm dist}(z,[0,\infty))^{n-1} }    
\end{align} 
holds for  $z \in \C  \setminus [0,\infty)$ with ${\rm Re} z \in I$ and $L \geq 1$. 
\end{proposition} 
Proposition    \ref{exinflim-bound}  will  be proven in Section \ref{sec:9}.  
To state the   main result  about the 
expansion coefficients of the resolvent  \eqref{eq:defoflimL},      
  we need some regularity assumptions for the profile function and the wave functions with respect to which 
  we calculate the matrix element of the resolvent. These regularity assumptions are   
  formulated in terms of so-called analytic dilations, cf. \cite{CombesThomas.1973,SR4}.  Let us first introduce the group of dilations.
 
 \begin{definition}
 \label{defutheta}
The group of unitary operators $u(\theta)$ on $L^2(\R^d)$ given by 
\muun{
\label{equtheta}
(u(\theta) \psi)(x) = e^{ \frac{d \theta}{2} } \psi(e^\theta x ) \quad \text{for} \ \theta \in \R , 
}
 is called the {\bf group of dilation operators on} $\R^d$. 
 \end{definition}
 
 It is straightforward to verify that $\R \mapsto \mathcal{B}(L^2(\R^d))$, $\theta \mapsto u(\theta)$ is indeed a strongly continuous one parameter group of  unitary operators, cf. \cite[Page 265]{SR1}. 
 Observe that from the definition of the Fourier transform and the dilation operator it is easy to see that  for any $\psi \in L^2(\R^d)$ and all $\theta \in \R$ 
 \begin{equation} \label{eq:dilfour}  
 (u(\theta) \psi)^{\wedge} = u(-\theta) \hat{\psi}. 
 \end{equation} 
Let us first state the hypothesis about the profile function needed for Theorem  \ref{thm:boundaryvelueexpcoeff}, below.

\begin{hyp} \label{H1} There exists $\vartheta_B > 0$ such that the Fourier transform, $\hat{B}$, of the Schwartz function $B$ satisfies that for $p=1 $ and $p=\infty$ the  function $\R \to L^p(\R^d),  \theta \mapsto  \widehat{B}_{\theta} := (u(\theta) B)^\wedge$  has  an extension 
to an analytic function
$D_{\vartheta_B} := \{ z \in \C : |z| < \vartheta_B\}\to L^p(\R^d).$
\end{hyp}

We can now state the first main result, which establishes the existence of the boundary value of
 the expansion coefficients in Theorem    \ref{thm:boundaryvelueexpcoeff} we find from 

\begin{theorem} \label{thm:boundaryvelueexpcoeff} Let $n\in \N_0$. Suppose  Hypothesis \ref{H1} holds. Then for any $E > 0$ the following limit exists    as a finite complex number 
$$
\widetilde{\rm tr} {\rm Im} {T}_{n,\infty}[E \pm  i 0^+] := \lim_{\eta \downarrow 0 } \widetilde{\rm tr}  {\rm Im} {T}_{n,\infty}[E \pm  i \eta] .
$$
Moreover,  $z \mapsto  \widetilde{\rm tr}  {\rm Im} {T}_{n,\infty}[z  ] $  as a function on $\HS := \{ z \in \C : {\rm Im } z  > 0 \}$  has a continuous extension to a function on $\HS  \cup  \{ z \in \C : {\rm Re}   (z) > 0 \text{ and } {\rm Im } z = 0 \}$.
\end{theorem} 
The proof of  Theorem \ref{thm:boundaryvelueexpcoeff}  will be given in Section     \ref{sec:10}.  
We will use that notation $$\R_> :=\{ x \in \R : x > 0 \}  , \qquad \R_{\geq} := \{ x \in \R : x \geq 0 \} . $$

The next theorem is  the second main result about  expectations of matrix elements  of the resolvent.   It shows  
that the expectations of matrix elements of the resolvent have an asymptotic expansion with expansion coefficients given by \eqref{defoffinT}. The estimate is uniform 
in the size  $L \geq 1$  of the box and yields an asymptotic expansion  as the spectral parameter approaches the positive real axis  in the weak coupling limit where $\eta= \lambda^{2-\epsilon}$.

\begin{theorem} \label{thm:maintec000}  Assume  $m_1 = \E_v  v_\gamma = 0$.  Then there exists an $L_0 \geq 1$ with the following property. 
Let $E > 0$.  Then the following holds.  
\begin{itemize}
\item[(a)] For each $n \in \N$ there exists a constant $K_{n,d,E,B} $  such that for  $\eta \in  (0,1]$ and $L \geq L_0$  we have  
\begin{align}
& \left|  \E_L  \widetilde{\tr}  {\rm Im}  (H_{  \lambda ,  L} - E  \mp  i  \eta  )^{-1}      -  \sum_{j=0}^{n-1}\widetilde{\tr} {\rm Im} T_{j,L}[E  \pm    i \eta  ]   \lambda^j  \right| \nonumber \\
&  \leq K_{n,d,E,B}  \left( \frac{ \lambda^2  }{\eta}   \right)^{n/2} \langle \lambda \rangle \langle \eta^{-1} \rangle \left(1 +   \ln(   \eta^{-1} + 1 )\right)^n    \eta^{-{3/2}}   . \label{eq:boundonexp}  
\end{align} 
The constant $K_{n,d,E,B}$ can be chosen as a function of $E$  to depend
continuously on compact subsets of $\R_>$. 
\item[(b)] 
For  any  $\epsilon \in ( 0,2)$ and  $N  > 0$ we have 
$$
 \E_L  \widetilde{\tr}  {\rm Im}  (H_{  \lambda ,  L} - E  \mp  i   \lambda^{2 - \epsilon}   )^{-1}  =  
 \sum_{n=0}^{ \lceil (20 +4 N)/\epsilon \rceil} \widetilde{\tr} {\rm Im} T_{n,L}[E \pm    i \lambda^{2 - \epsilon} ]    \lambda^n + O(\lambda^N)   
$$
for  $\lambda \downarrow 0$ uniformly in $L \geq L_0$ and $E$ in compact subsets of $\R_>$.
\end{itemize} 
\end{theorem}

The proof of  Theorem \ref{thm:maintec000}  will be given in Section     \ref{sec:proooferror}.

\begin{remark} {\rm  The energy  $E > 0$ lies inside the spectrum of the free Laplacian which  is the main  challenge we aim to address in this paper. 
Note that for negative energies $E < 0$ outside of the free Laplacian, the corresponding  estimates and limits of Theorems \ref{thm:boundaryvelueexpcoeff} and \ref{thm:maintec000}  are in fact 
straightforward  to establish   by means of  elementary resolvent identities for $E < 0$. }
\end{remark}

\begin{remark} {\rm  It would be interesting to investigate to what extent the  expansion  \eqref{eq:boundonexp}   in Theorem \ref{thm:maintec000} could be improved 
using  so-called tadpole renormalization \cite{elgart} and an analysis of crossing graphs,  see also  \cite{ErdosSalmhoferYau.2008a}.  In \cite{ErdosSalmhoferYau.2008a}
such an analysis  was successfully  used to establish quantum diffusion in a diffusive  scaling limit  for a model similar to the one introduced in the present paper.  
We did not perform such an analysis in this work,  
since we wanted to  focus   on    the expansion coefficients  and their approximation to the density of states.
}
\end{remark} 

\begin{remark} 
{\rm  An estimate similar to  \eqref{eq:boundonexp}   was obtained in  \cite{MagnenPoirotRivasseau.1998} for a related model. In  that work
a Laplacian with  a Gaussian random potential in two dimensions was studied and it was shown that the difference between the expectation of the  resolvent  and a renormalized free resolvent  is  operator norm bounded by   $\left( \sfrac{\lambda^{2 + \delta}}{\eta} \right)^3 \frac{1}{\eta}$ for some 
$\delta > 0$,  provided $0 \leq \eta \leq \lambda^2$.   The bound   \eqref{eq:boundonexp}  in Theorem \ref{thm:maintec000} can be made 
much smaller than $\left( \sfrac{\lambda^{2 + \delta}}{\eta} \right)^3 \frac{1}{\eta}$ in  the regime where $\eta = \lambda^{2 - \epsilon}$  by  choosing  $n$ sufficiently large.
 }
\end{remark}

\section{Asymptotic Expansion}

\label{asympsec}

In this section we combine the  results of the main Theorems \ref{thm:boundaryvelueexpcoeff} and  \ref{thm:maintec000}  
 to   derive concrete  asymptotic expansions   for  infinite volume limits of  matrix elements of resolvents. 
For this we shall make the following assumption that for all $z \in \C \setminus \R$  the limit 
\begin{align} \label{defoflimit-00}
\mathcal{R}(\lambda;z) :=\lim_{L \to \infty} |\Lambda_L|^{-1}  \tr {\bf E}_L {\rm Im} (H_{\lambda,L}-z)^{-1} 
\end{align} 
exists. Such a working assumption is natural  and is  used in \cite{Schenker.2004,CombesHislopKlopp.2007}  as well. 
 Moreover, we  note that results of this type have been shown in \cite{KirschMartinelli.1982} for a class of  operators
 containing the present operator  up to the choice of the boundary conditions and the periodic extension of the random potential. 

\begin{corollary}\label{asymptotictheorem} Assume Hypothesis \ref{H1},  $\E_v v_\gamma = 0$.   Assume that the limit \eqref{defoflimit-00}
exists. 
Let $\epsilon \in (0,2)$
and $\eta = \lambda^{2 - \epsilon}$.
Let $N \in \N$ and $I$ be a compact subset of $\R_>$.  Then there exists 
a constant $C$ such that for all  $\lambda > 0$ and $E \in I$
\begin{align} \label{expansionasymptoticmain}
 &  \Bigg|  \mathcal{R}(\lambda;E \pm i \eta)  -     \sum_{n=0}^{N-1} \widetilde{\rm tr} {\rm Im} {T}_{n,\infty}[  E \pm i \eta ] \lambda^n
\Bigg|   \leq  C \lambda^{N}  .
\end{align} 
\end{corollary}

\begin{proof} In view of the bounds \eqref{boundonTs}  and  \eqref{ln-1} and the spectral theorem
it suffices to assume  $0   <   \lambda\leq 1$. 
Now  using Theorem   \ref{thm:maintec000} (b) and then taking  subsequently the limit  
 $L \to \infty$  by  means of \eqref{defoflimit-00} and  Theorem   \ref{exinflim}
we see that  
 \begin{align}  \label{bounddensity-0} 
 & \left|  \mathcal{R}(\lambda;E\pm i \eta)   - 
  \sum_{n=0}^{ \lceil (20 +4 N)/\epsilon \rceil}\widetilde{\rm tr} {\rm Im} {T}_{n,\infty}[E \pm    i \eta ]   \lambda^n 
 \right| \leq   C_0 \lambda^{N} 
\end{align} 
for some constant $C_0$ and $\lambda \geq 0$. 
Now by means of  Proposition  \ref{exinflim-bound} and Theorem \ref{thm:boundaryvelueexpcoeff}  we know that
\begin{align}  \label{regofotherTs-2-0}
C_{n,I} := \sup_{\eta \in \R \setminus\{0 \}, E \in I  }  |\widetilde{\rm tr}  {\rm Im} {T}_{n,\infty}[E +    i \eta ] | < \infty   , \qquad n \in \N . 
\end{align} 
Now      using  \eqref{regofotherTs-2-0},  \eqref{bounddensity-0},   and 
 the triangle inequality we arrive at 
\begin{align*} 
 & \left| \mathcal{R}(\lambda;E\pm i \eta)    
-    \sum_{n=0}^{N-1} \lambda^n  \widetilde{\rm tr} {\rm Im} {T}_{n,\infty}[E \pm    i \eta   ]     \right|    \\
& \leq  \left|  \mathcal{R}(\lambda;E \pm i \eta)  
-    \sum_{n=0}^{ \lceil (20 +4 N)/\epsilon \rceil}  \lambda^n   \widetilde{\rm tr}  {\rm Im} {T}_{n,\infty}[E \pm    i \eta   ] \right|  +  \sum_{n=N}^{ \lceil (20 +4 N)/\epsilon \rceil}   C_{n,I}  \lambda^n \nn  \\ 
&   \leq   C_0     \lambda^N  +   \sum_{n=N}^{ \lceil (20 +4 N)/\epsilon \rceil}   C_{n,I}  \lambda^n  \leq C \lambda^{N}    , \nn 
\end{align*}  
this shows the claim. 
\end{proof}

\begin{remark}{\rm 
We note that  Corollary  \ref{asymptotictheorem} gives  an   asymptotic expansion 
of  $ \mathcal{R}(\lambda;E + i \eta)$ with   $\lambda$-dependent expansion coefficients, which are 
 continuous and bounded. Explicitly,  it follows immediately from  Theorem   \ref{thm:boundaryvelueexpcoeff}  and  Corollary   \ref{asymptotictheorem}  that   for $\eta = \lambda^{2 - \epsilon}$
\begin{align} \label{asympprop} 
\lim_{\lambda \to 0}\frac{ \mathcal{R}(\lambda;E + i \eta) -\sum_{n=0}^{N-1} \widetilde{\rm tr} {\rm Im} {T}_{n,\infty}[E+i\eta] \lambda^n }{  \lambda^{N}}  =  \widetilde{\rm tr} {\rm Im} {T}_{N,\infty}[E+i 0^+] . 
\end{align} 
In fact, it follows from \eqref{asympprop}, that whenever the r.h.s of  \eqref{asympprop} 
is nonzero, the functions 
$$
\left( \mathcal{R}(\lambda;E + i \eta) -\sum_{n=0}^{N-1}  \widetilde{\rm tr}   {\rm Im}{T}_{n,\infty}[E+i\eta] \lambda^n \right)   \quad \text{ and } \quad  \widetilde{\rm tr}
{\rm Im} {T}_{N,\infty}[E+i 0^+]  \lambda^N 
$$
 are  asymptotically equivalent   according to the standard definition, cf.  \cite{deBruijn.1981}. 
In fact, the expansion \eqref{expansionasymptoticmain} given in  Corollary  \ref{asymptotictheorem} is a so-called  extended asymptotic expansion with respect to the asymptotic sequence $( \lambda^n )_{n \in \N_0}$  as $\lambda \downarrow 0$, cf. \cite{EstradaKanwal.2002}.

}
\end{remark}

\section{Estimates of the Integrated Density of States}

\label{dosapprox} 

In this section we show how the main results can be used to show H\"older continuity 
of the density of states.

Let us first introduce the following notation. For  $E \in \R$ and $\eta \neq  0$ define the function 
\begin{align} \label{eq:convfunc-0} 
 \gamma_{\eta,E} : \R \to \R, \quad t \mapsto \gamma_{\eta,E}(t)  := \frac{1}{\pi}\frac{\eta}{\eta^2  + (E-t)^2}  
\end{align} 
Moreover we set $\gamma_{\eta} := \gamma_{\eta,0}$. 
We note that 
\begin{align} \label{relgammimres} 
{\rm Im}(t -z)^{-1} = \pi  \gamma_{{\rm Im}z ,{\rm Re} z }(t) ,
\end{align}   
and 
for $\eta > 0$ we have the normalization$\int_\R \gamma_{\eta,E}(t) dt = 1$.
Moreover we introduce the spaces 
\begin{align*}
\mathcal{C}_{\langle \cdot \rangle^2} &:= \{ f : \R \to \C : \text{ continuous, }  \sup_{x \in \R} |f(x) \langle x \rangle^2 |< \infty   \}  \\ 
\mathcal{M}_{\langle \cdot \rangle^2} & := \{ f : \R \to \C : \text{ Borel measurable, }  \sup_{x \in \R} |f(x) \langle x \rangle^2 |< \infty   \}  . 
\end{align*}
In particular we have $\gamma_{\eta,E} \in \mathcal{C}_{\langle \cdot \rangle^2}  \subset \mathcal{M}_{\langle \cdot \rangle^2} $.

\begin{lemma}  \quad \label{existofdosint-0} 
There exists  a unique Borel measure $\mu_{\lambda,L}$ on $\R$ such that for all
 $ f \in \mathcal{M}_{\langle \cdot \rangle^{2}} (\R)$ 
\begin{align} \label{relfinspecmeas} 
\int_\R f(s)  d \mu_{\lambda,L}(s)  = \widetilde{\rm tr} {\bf E}_L f(H_{\lambda,L } ) 
\end{align} 
The Borel measure $\mu_{\lambda,L}$ is uniquely determined by requiring that \eqref{relfinspecmeas} 
holds for all $f \in C_c(\R)$. 
\end{lemma}

\begin{proof}
Let $\varphi_j$, $j \in \N$,  be an o.n.b. of $L^2(\Lambda_L)$. By the spectral theorem 
 there exists  for each $j$ and $\omega$   a unique Radon  measure $\nu_{j,\omega,\lambda,L}$ such that 
$$
\int f(s) \nu_{j,\omega,\lambda,L}(s) = \langle \varphi_j , f(H_{\lambda,L,\omega}) \varphi_j \rangle   
$$
for all  bounded measurable $f : \R \to \C$, cf.  \cite[VII.2]{SR1}] \cite[Chapter 7]{folland}. 
It follows immediately from the definitions that $\nu_{j,\omega,\lambda,L}$ is a random measure 
\cite[Definition 24.3]{Klenke.2014}. Thus there exists a measure $\overline{\nu}_{j,\lambda,L}$, called  the intensity measure, such that 
$$
\int f(s)    d \overline{\nu}_{j,\lambda,L}  (s) = {\bf E} \langle \varphi_j , f(H_{\lambda,L,\omega}) \varphi_j  \rangle 
$$
cf. \cite[Theorem 24.4]{Klenke.2014}.
Since sums of measures are measures, it follows 
that  ${\mu}_{\lambda, L} := |\Lambda_L|^{-1} \sum_{j=1}^\infty \overline{\nu}_{j,\lambda,L}$ is a measure.
In view of  Proposition \ref{lem:proofofexdens00-0}, for all $f  \in \mathcal{M}_{\langle \cdot \rangle^2}$ we have 
 $$
\int f(s)    d {\mu}_{\lambda,L}  (s) = |\Lambda_L|^{-1} {\bf E} {\rm tr} f(H_{\lambda,L,\omega})  < \infty  .  
$$
In particular, ${\mu}_{\lambda,L}$ is a locally finite Borel measure on the real line. 
As such it is uniquely determined by its values on continuous functions with compact support.
\end{proof} 
From the lemma it follows in particular, that    \eqref{relfinspecmeas}  holds for functions $x \mapsto  {\rm Im} (x-z)^{-1}$ with $z \in \C \setminus \R$.
 In fact, such functions determine the measure  $\mu_{\lambda,L}$ completely, which 
is the content of the next lemma. 

\begin{lemma} \label{lemuniquenessfinite} 
The measure $\mu_{\lambda,L}$ is  the unique Borel measure such that  
\begin{align} \label{measuniqufinitbor-0} 
\int_\R  {\rm Im}(s - z )^{-1}  d \mu_{\lambda,L}(s)  = {\bf E}_L  \widetilde{\rm tr} {\rm Im} ( H_{\lambda,L} - z )^{-1}   
\end{align} 
for all $z \in \C \setminus \R$. 
\end{lemma} 
\begin{proof}
By  Proposition \ref{lem:proofofexdens00-0}, we know that the right hand side of  \eqref{measuniqufinitbor-0} 
is finite. In particular $\int_\R  \langle x \rangle^{-2} d \mu_{\lambda,L}(x) < \infty$.  
A Borel measure on $\R$  is uniquely determined by  its integral with continuous functions with compact support. 
Now let $f \in C_c(\R)$. Then for any $\eta >  0$, we have $f * \gamma_\eta \in C_{\langle \cdot \rangle^2}(\R)$
and $f * \gamma_\eta \to f$ pointwise as $\eta \downarrow 0$. Moreover, 
$| f* \gamma_\eta | \leq C \langle \cdot  \rangle^{-2}$   uniformly in $\eta \in (0,1]$, where the right hand side 
is a  majorant. 
So it follows by dominated convergence that $\lim_{\eta \downarrow} \int f * \gamma_\eta  = \int f $. 
Thus functions of the form  $f * \gamma_\eta$ with $f \in C_c(\R)$ and $\eta > 0$ determine the measure completely.
Next we write the integral in the definition of the convolution $f * \gamma_\eta(x) =  \int f(y ) \gamma_{\eta,y}(x) dy $ as a Riemann sum
$ \sum_{j =1}^N f(y_j) \gamma_{\eta,y_j}(x) \Delta y $, where $y_0 = \inf \supp f$, $\Delta y = N^{-1} {\rm diam}(\supp f)$
and $y_j = y_0 + j \Delta y$.  By standard estimates of Riemann sums for 
uniformly continuous functions the Riemann sum converges for  $N \to \infty$  
pointwise to  $f * \gamma_\eta(x)$. On the other hand elementary estimates show that there  exists a constant 
such that for all $x \in \R$ and $\eta \in (0,1]$ 
$$ 
\left| \sum_{j =1}^N f(y_j) \gamma_{\eta,y_j}(x) \Delta y  \right| \leq \frac{C}{\langle x \rangle^2}
$$
Thus by dominated convergence  $\int \sum_{j =1}^N f(y_j) \gamma_{\eta,y_j}(x) \Delta y  d\mu(x) \to \int f * \gamma_{\eta,0} d\mu(x)$.  
This shows that it is sufficient to consider finite linear combinations of the functions $f_{\eta, E}$. This shows 
the claim. 
 \end{proof}

As in Section \ref{asympsec} we shall make the assumption that the limit  \eqref{defoflimit-00} exists for all $z \in \C \setminus \R$.

\begin{lemma} \label{exofdosmeasure}  Assume that the limit  \eqref{defoflimit-00} exists  for all $z \in \C \setminus \R$. 
Then for all 
$f \in \mathcal{C}_{\langle \cdot \rangle^2}(\R)$
also the limit 
$
 \lim_{L \to \infty} {\bf E}_L \widetilde{\rm tr} f(H_L) 
$
exists, and there exists a unique Borel measure $\mu_{\lambda,\infty}$ on $\R$ such that 
\begin{align} \label{relfinspecmeas-2} 
\int_\R  f(s)  d \mu_{\lambda,\infty}(s)  =  \lim_{L \to \infty} {\bf E}_L \widetilde{\rm tr} f(H_L)  .
\end{align} 
for all $f \in \mathcal{C}_{\langle \cdot \rangle^2}(\R)$. The Borel measure $\mu_{\lambda,\infty}$ is uniquely determined by requiring that \eqref{relfinspecmeas-2} 
holds for all $f \in C_c(\R)$. 
\end{lemma} 
We shall call $\mu_{\lambda,\infty}$ the density of states measure.

\begin{proof} Let $\alpha \in (3/2,2)$.  
 First we show that for any $f \in C_c(\R)$  and $\epsilon > 0$ there exists
 an element $R \in \mathcal{R} := {\rm lin} \{  \gamma_{\eta,E}  : \,   \eta > 0 , \,  E \in \R   \}$ in the linear span such that for all $x \in \R$ 
 \begin{align} \label{densitywdecay} 
  | f(x) - R(x) | \leq \frac{\epsilon}{\langle x \rangle^\alpha} .
 \end{align} 
  To see \eqref{densitywdecay}, let  $\tilde{\epsilon} > 0$.
Then 
\begin{align} 
&|{f}(x) - \int {f}(y) \gamma_\eta(x-y) dy  | = | \int [  f(x) - f(x+ \eta z ) ] \gamma_1(z) dz |  \nn \\
& \leq 2 \| f \|_\infty \int_{|z| \geq R}   \gamma_1(z) dz  +  \sup_{x \in \R} \sup_{|z| \leq R} |f(x) - f(x + \eta z ) | \label{trivialintest}  . 
\end{align} 
Now  choosing  $R > 0$ sufficiently large such that 
the first term in line  \eqref{trivialintest} is less than $\tilde{\epsilon}/2$ and then choosing (by the uniform continuity of $f$)   $\eta > 0$ 
sufficiently small such that also the second term in line  \eqref{trivialintest} is less than  $\tilde{\epsilon}/2$,
we obtain that 
\begin{align} 
&|{f}(x) - \int {f}(y) \gamma_\eta(x-y) dy  |  < \tilde{\epsilon} , \nn 
\end{align} 
for $\eta > 0$ sufficiently small. 
  Now  by  the uniform continuity of ${f}$ and $\gamma_\eta$ and  by well-known estimates  concerning  Riemann integration \cite{Rudin.1976} 
 there exist  $N$, $\delta >0$, and  $y_0 \in \R$ such that with  $y_j = y_0 + j \delta$  we have
   $\sup_x | \int {f}(y) \gamma_\eta(x-y) dy  -   \sum_{j=0}^N {f}(y_j) \gamma_\eta(x-y_j) \delta    |  < \tilde{\epsilon}$.  
  Thus by the triangle inequality for all $x \in \R$ we find that 
 \begin{align}  |{f}(x) -   \sum_{j=0}^N {f}(y_j) \gamma_\eta(x-y_j) \delta  | < 2 \tilde{\epsilon}.
 \end{align} 
On the other hand let $a > 0$ such that ${\rm supp} f \subset [-a,a]$. For $x \notin [-a-1,a+1]$ we find for all $0 < \eta \leq 1$ 
   $$
|{f}(x) -   \sum_{j=0}^N {f}(y_j) \gamma_\eta(x-y_j) \delta  |  = |   \sum_{j=0}^N {f}(y_j) \gamma_\eta(x-y_j) \delta  | \leq
\frac{    \| f \|_\infty  2a  }{({\rm dist}(x,\{-a,a\} ))^2 } 
$$
Since $\alpha < 2$ there exists a sufficiently large   $R > a +1$ so   that  for $|x| \geq R$ 
\begin{align}
 \frac{    \| f \|_\infty 2a  }{ ({\rm dist}(x,\{-a,a\} ))^2 } 
 \leq   \frac{\epsilon}{\langle x \rangle^\alpha} 
\end{align} 
Furthermore,  we can now choose $\tilde{\epsilon} = \epsilon/ (2 \langle R \rangle^\alpha)$. 
With these choices  we find 
   $$
|{f}(x) -   \sum_{j=0}^N {f}(y_j) \gamma_\eta(x-y_j) \delta  | \leq 1_{|x| <  R}  2 \tilde{\epsilon}  +  1_{|x| \geq R} 
\frac{\epsilon}{\langle x \rangle^\alpha}  \leq \frac{\epsilon}{\langle x \rangle^\alpha} 
$$
Hence we have shown inequality \eqref{densitywdecay}.

  Next we use the density \eqref{densitywdecay}  to  show that for any $f \in C_c(\R)$ the limit 
  \begin{align} \label{limofEanyf}
  \lim_{L \to \infty} {\bf E} \widetilde{\rm tr} f(H_L) 
  \end{align} 
  exists. Let $\epsilon > 0$. Then there exists an $R \in \mathcal{R}$ such that  $| f(x) - R(x) | \leq \frac{\epsilon/3}{\langle x \rangle^\alpha}$. By the assumption  \eqref{defoflimit-00} it follows by  the Cauchy criterion for convergent sequences
  and linearity, 
  that there exists an  $L_0 \geq 0$ such that for all $L,L' \geq L_0$ 
  \begin{align}
  | {\bf E}   \widetilde{\rm tr}   R (H_{\lambda,L})  -  {\bf E} \widetilde{\rm tr}   R (H_{\lambda,L'})  | 
  \leq \epsilon/3 .
  \end{align} 
  Thus for all $L,L' \geq L_0$ we find using the triangle inequality 
  \begin{align}
& | {\bf E} \widetilde{\rm tr} f(H_L) -{\bf E} \widetilde{\rm tr} f(H_{L'})  | \nonumber  \\
 & \leq  | {\bf E} \widetilde{\rm tr} f(H_L) -{\bf E} \widetilde{\rm tr} R(H_{L})  |  + 
  | {\bf E} \widetilde{\rm tr} R(H_L) -{\bf E} \widetilde{\rm tr} R(H_{L'})  | + | {\bf E} \widetilde{\rm tr} f(H_{L'}) -{\bf E} \widetilde{\rm tr} R(H_{L'})  | \nonumber \\
 & \leq   {\bf E} \widetilde{\rm tr} | f(H_L) -  R(H_{L})  |  + 
  | {\bf E} \widetilde{\rm tr} R(H_L) -{\bf E} \widetilde{\rm tr} R(H_{L'})  | +  {\bf E} \widetilde{\rm tr} |  f(H_{L'}) - R(H_{L'})  | \nonumber
\\ & \leq C \epsilon/3 + \epsilon/3 + C \epsilon/3 , \label{epsilon3arg} 
  \end{align} 
where we used \eqref{densitywdecay} and   Proposition \ref{lem:proofofexdens00-0}   in the last line.
  This shows the convergence by the Cauchy criterion. This shows the convergence for $f \in C_c(\R)$.
  Next we show the convergence for $f \in \mathcal{C}_{\langle \cdot \rangle^2}(\R) := \{ g \in C(\R):  \| \langle \cdot \rangle^2 g \|_\infty < \infty\}$.  
  To this end we use a similar $\epsilon/3$-argument as above. 
  First observe that there exists an   $R \in C_c(\R)$ such that 
  \begin{align} | f(x) - R(x)  |   \leq  
  \frac{\epsilon/3}{\langle x \rangle^\alpha} , 
  \end{align} 
  which can be found by choosing $R = f \chi(\tau \cdot)$
  for a symmetric cutoff function  $\chi \in C_c(\R)$ with $0 \leq \chi \leq 1$,  $\chi =1$
 in a neighborhood of the origin, 
monotonically decreasing on $[0,\infty)$,  and
  $\tau >0$ sufficiently small.  Now an argument analogous to 
  \eqref{epsilon3arg}, shows  that the limit  \eqref{limofEanyf}
  exists for all   $f \in \mathcal{C}_{\langle \cdot \rangle^2}(\R)$.
   \\
Now define the linear functional for  $f \in \mathcal{C}_{\langle \cdot \rangle^2}(\R)$ 
 $$
L(f) :=     \lim_{L \to \infty} {\bf E} \widetilde{\rm tr} f(H_L) .
$$
By the spectral theorem the functional is positive. It follows as a conclusion of the  Riesz representation theorem, that
there exists a unique Borel measure $\mu_{\lambda,\infty}$ on $\R$  such that  for all  $f \in \mathcal{C}_{\langle \cdot \rangle^2}(\R)$ 
\begin{align} \label{riesz}  
L(f) = \int_\R f(s) d \mu_{\lambda,\infty}(s) .
\end{align} 
To see this, we note that by the classical Riesz representation theorem 
the existence of a Borel measure follows if we consider only   $f \in C_c(\R)$
and  then   \eqref{riesz}   follows  for $f \in C_c(\R)$.
Next  we use Proposition \ref{lem:proofofexdens00-0} to obtain the bound for all $f \in \mathcal{C}_{\langle \cdot \rangle^2}(\R)$
\begin{align} \label{condfuncbound} 
| L(f  )  | \leq     \lim_{L \to \infty} | {\bf E} \widetilde{\rm tr} f(H_L) |  =    \lim_{L \to \infty} |  {\bf E} \widetilde{\rm tr} f  \langle \cdot \rangle^\alpha  \langle \cdot \rangle^{-\alpha} (H_L)|   \leq  
C \| f  \langle \cdot \rangle^\alpha \|_\infty 
\end{align} 
Assume $f \geq 0$, and let   $R_\tau = f \chi(\tau \cdot)$ as above. Then it follows 
by monotone convergence and  continuity of $L$ (by means of  \eqref{condfuncbound})  
\begin{align}
\int f(s) d\mu_{\lambda,\infty}(s) = \lim_{\tau \to  0} \int R_\tau(s) d\mu_{\lambda,\infty}(s) =  \lim_{\tau \to  0} L(R_\tau) = L(f)  
\end{align} 
The general case of complex valued $f$ then follows by decomposing it into positive and negative real and imaginary parts and using linearity. 
\end{proof} 
 
Note that \eqref{relfinspecmeas-2}  holds in particular for functions of the form $x \mapsto {\rm Im} (x-z)^{-1}$ with $z \in \C \setminus \R$. 
In fact these functions determine the density of states measure uniquely, this is the content of the following lemma.

\begin{lemma} \label{lemuniquenessinf} 
The measure $\mu_{\lambda,\infty}$ is  the unique Borel measure such that  
\begin{align} \label{measuniqufinitbor} 
\int_\R  {\rm Im}(s - z )^{-1}  d \mu_{\lambda,\infty}(s)  = \lim_{L \to \infty} {\bf E}_L \widetilde{\rm tr}  {\rm Im} ( H_{\lambda,L} - z )^{-1}   
\end{align} 
for all $z \in \C \setminus \R$. 
\end{lemma} 
The proof is analogous to that of Lemma \ref{lemuniquenessfinite}  with the obvious modifications. 
Now combining the two lemmas we obtain the following result 

\begin{proposition}\label{portmanteux} 
 Assume that the limit 
 \eqref{defoflimit-00}  exists  for all $z \in \C \setminus \R$. 
Then for all 
$f \in \mathcal{C}_{\langle \cdot \rangle^2}(\R)$ 
\begin{align} \label{limirofmeasures}
\int_\R  f(s)  d \mu_{\lambda,\infty}(s)  
=  \lim_{L \to \infty}  \int_\R  f(s)   d  \mu_{\lambda,L}(s)  =  \lim_{L \to \infty} {\bf E}_L \widetilde{\rm tr} f(H_L) 
  .
\end{align} 
For every bounded open set $O \subset \R$ and every bounded closed set $C \subset \R$  
\begin{align}
\liminf_{L \to \infty} \mu_{\lambda,L}(O) & \geq \mu_{\lambda,\infty}(O)   \\
\limsup_{L \to \infty} \mu_{\lambda,L}(C)  & \leq \mu_{\lambda,\infty}(C)  .
\end{align} 
\end{proposition} 
\begin{proof} First observe that  \eqref{limirofmeasures} is an immediate consequence of Lemma  \ref{exofdosmeasure} and Lemma \ref{existofdosint-0}. 
Suppose that the open $O$ and closed $C$ are contained in the ball $B$. Let  $\chi \in C_c(\R)$  with $0 \leq \chi$
and $\chi =1$ on $B$. Define the measures $d\nu_{\lambda, L,\chi}(s) = \chi(s) d \mu_{\lambda,L}(s)$ for $L  \geq 1$.    
Then it follows from \eqref{limirofmeasures} that the measure $\nu_{\lambda,L,\chi}$ converges weakly 
to $\nu_{\lambda,\infty,\chi}$. Now by the Portmanteau theorem we find that 
\begin{align*}
\liminf_{L \to \infty} \mu_{\lambda,L}(O) & = \liminf_{L \to \infty} \nu_{\lambda,L,\chi}(O) \geq  \nu_{\lambda,\infty,\chi}(O) =  \mu_{\lambda,\infty}(O)  \\
\limsup_{L \to \infty} \mu_{\lambda,L}(C)  & = \limsup_{L \to \infty} \nu_{\lambda,L,\chi}(C) \leq  \nu_{\lambda,\infty,\chi}(C) =  \mu_{\lambda,\infty}(C)  
\end{align*} 
This shows the claim. 
\end{proof} 
Henceforth we will 
assume that the limit 
 \eqref{defoflimit-00}  exists  for all $z \in \C \setminus \R$. 
Furthermore, we will make use of the following identity for $f \in C_c(\R)$
\begin{align} 
&\int  f(s) \frac{1}{|\Lambda_L|} {\rm tr} {\rm Im} (  H_{\lambda,L} - s - i \eta  )^{-1} ds
\nonumber  \\
& = \int  f(s) \int {\rm Im}   ( x  - s - i \eta  )^{-1} d \mu_{\lambda,L}(x) ds  \label{eqintdensrel} \\
& = \int  \int   f(s) {\rm Im}   ( x  - s - i \eta  )^{-1} ds d \mu_{\lambda,L}(x)   \label{eqintdensrel-2} \\
& =
\int  (  f  * {\rm Im}   ( \cdot  - i \eta  )^{-1} )(x) d  \mu_{\lambda,L}(x)  ,\label{eqintdensrel-3-f} 
\end{align}
where the first identity follows from the definition of the measure $\mu_{\lambda,L}$,
the second follows from Fubini and the last from the definition of the convolution of 
two functions.

From the main result we immediately obtain the following result using the corresponding relation for 
the case $\lambda=0$. 
For notational compactness we introduce the error term in the main theorem 
\begin{align}
\mathcal{E}_{n,d,E,B} ( \lambda  ,  \eta  ) :=  K_{n,d,E,B}  \left( \frac{ \lambda^2  }{\eta}   \right)^{n/2} \langle \eta^{-1} \rangle \langle \lambda \rangle  \left(1 +   \ln(   \eta^{-1} + 1 )\right)^n    \eta^{-{3/2}}  . \label{eq:boundonexp-2}  
\end{align} 

\begin{corollary}   Assume  $m_1 = \E_v v_\gamma = 0$. Let  $\chi \in C(\R) \cap L^1(\R)$
with ${\rm supp}\chi$ a compact subset of  $(0,\infty)$. 
Then for $\lambda \geq 0$
we find 
\begin{align*} 
& \left| \int   \chi    *   \gamma_{\eta}(e)  d  \mu_{\lambda,\infty}(e)   - \int \chi * \gamma_{\eta}(e)         d \mu_{0,\infty}(e)  -  \frac{1}{\pi} \int \chi(e) \sum_{j=1}^{n-1} \widehat{{\rm tr}} {\rm Im}  T_j[ e + i \eta]   \lambda^j de 
\right|  \\
&  \leq  \| \chi \|_1  \frac{1}{\pi} \sup_{E \in \supp \chi} \mathcal{E}_{n,d,E,B} ( \lambda  ,  \eta  ) .
\end{align*}  
\end{corollary} 
\begin{proof} Using \eqref{limirofmeasures} of Proposition \ref{portmanteux}  
for the first two terms and dominated convergence by means of     Theorem \ref{exinflim}  and the bound of Proposition  \ref{exinflim-bound}
for the third term
we find  
\begin{align*} 
& \left| \int   \chi    *   \gamma_{\eta}(e)  d  \mu_{\lambda,\infty}(e)   - \int \chi * \gamma_{\eta}(e)         d \mu_{0,\infty}(e)  -  \frac{1}{\pi} \int \chi(e) \sum_{j=1}^{n-1} \widehat{{\rm tr}} {\rm Im}  T_j[ e + i \eta] \lambda^j  de 
\right|  \\
& =\lim_L \left| \int   \chi    *   \gamma_{\eta}(e)  d  \mu_{\lambda,L}(e)   - \int \chi * \gamma_{\eta}(e)         d \mu_{0,L}(e)  -  \frac{1}{\pi} \int \chi(e) \sum_{j=1}^{n-1} \widehat{{\rm tr}} {\rm Im}  T_{j,L}[ e + i \eta]  \lambda^j  de 
\right|  \\
&  \leq  \| \chi \|_1  \frac{1}{\pi} \sup_{E \in {\rm supp}\chi } \mathcal{E}_{n,d,E,B} ( \lambda  ,  \eta  ) ,
\end{align*} 
where in the last line we  used  \eqref{relgammimres}  \eqref{eqintdensrel-3-f}  and Theorem  \ref{thm:maintec000} .
\end{proof} 

Next we want express   the convolution in the previous corollary in terms of the original function, up to some error.
For this we use the identity of the following Lemma. 

\begin{lemma} \label{trivialidentcauchy}  For $f \in L^1(\R)$ and $\epsilon > 0$ 
\begin{align}
(\gamma_{\epsilon } * f)(x) =  \frac{1}{2\pi i}\int_\R  \left( \frac{1}{x-t-i \epsilon} - \frac{1}{x-t+ i \epsilon} 
\right) f(t) d t 
= f(x) +\frac{1}{\pi }\int_\R \frac{ f(x+\epsilon t)-f(x)  }{t^2 +1} \label{identcauchy-0} 
d t
\end{align}
\end{lemma}
\begin{proof} The first identity follows from \eqref{relgammimres}  and the definition of the convolution.
Using a change of variables  and  $\int_\R (t^2+1)^{-1} dt = \pi$ we find  
\begin{align*}
\text{ L.H.S. of \eqref{identcauchy-0} }   = \frac{1}{\pi }\int_\R  \frac{\epsilon}{(x-t)^2 +\epsilon^2} f(t)  d t 
= \frac{1}{\pi }\int_\R \frac{1}{t^2 +1}  f(x+\epsilon t) 
d t
 = \text{ R.H.S. of \eqref{identcauchy-0} } 
\end{align*}

\end{proof} 

Now using Lemma \ref{trivialidentcauchy} in Theorem  \ref{thm:maintec000} we find 

\begin{lemma}  \label{corollarygenapprox}  Assume  $m_1 = \E_v v_\gamma = 0$. Let $n \in \N$ and  $\eta  >0 $.
Then  for $\chi \in C(\R) \cap L^1(\R)$ with ${\rm supp} \chi$ a compact subset of $(0,\infty)$  and  $\lambda \geq 0$ 
\begin{align*} 
& \left| \int   \chi(e)  d  \mu_{\lambda,\infty}(e)   - \int ( \chi * \gamma_\eta)        d \mu_{0,\infty}(e)  - \frac{1}{\pi}  \int \chi(e) \sum_{j=1}^{n-1} \widetilde{{\rm tr}} {\rm Im}  T_j[ e + i \eta] \lambda^j de 
\right|  \\
&  \leq   \frac{1}{\pi } \| \chi\|_1 \sup_{E \in {\rm supp}\chi}  \mathcal{E}_{n,d,E,B} ( \lambda  ,  \eta  )  + \left|  \frac{1}{\pi} \int \int_\R \frac{\chi(x+\eta  t)-\chi(x)}{t^2 + 1 } dt d \mu_{\lambda,\infty}(x)    \right| .
\end{align*} 
\end{lemma} 
\begin{proof}
Using  Lemma \ref{trivialidentcauchy}  and then \eqref{limirofmeasures} of Proposition  \ref{portmanteux} 
\begin{align}
\int \chi(e)  d  \mu_{\lambda,\infty}(e) 
& +   \frac{1}{\pi } \int  \int_\R \frac{ \chi(e+\eta t)-\chi(e)  }{t^2 +1}  d  \mu_{\lambda,\infty}(e) \nonumber \\
 &= \frac{1}{2\pi i}\int_\R \int  \left( \frac{1}{e -t-i \eta} - \frac{1}{e-t+ i \eta} d  \mu_{\lambda,\infty}(e) 
\right) \chi(t) d t \nonumber \\
& = \frac{1}{ \pi }\int_\R\lim_{L \to \infty} {\bf E}\widetilde{\rm tr} {\rm Im}  (H_{\lambda,L}- t  - i \eta )^{-1}  \chi(t) dt  \label{someestinfordos} 
\end{align} 
Using \eqref{someestinfordos} and the triangle inequality as well as   dominated convergence by means of     Theorem \ref{exinflim} 
 and the bound of Proposition  \ref{exinflim-bound}
\begin{align*} 
& \left| \int   \chi(e)  d  \mu_{\lambda,\infty}(e)   - \frac{1}{\pi}   \int \chi(e) \sum_{j=0}^{n-1} \widetilde{{\rm tr}} {\rm Im}  T_{j,\infty}[ e + i \eta]  \lambda^j de 
\right|  \\
&  \leq   \left|   \frac{1}{ \pi }\int_\R \lim_{L \to \infty} \left( {\bf E}\widetilde{\rm tr} {\rm Im}  (H_{\lambda,L}- t  - i \eta )^{-1}  -  \sum_{j=0}^{n-1} \widetilde{{\rm tr}} {\rm Im}  T_{j,L}[ t + i \eta]  \lambda^j \right)  \chi(t) dt  
\right|  \\
& + |  \frac{1}{\pi } \int  \int_\R \frac{ \chi(e+\eta t)-\chi(e)  }{t^2 +1}  dt d  \mu_{\lambda,\infty}(e)  | \\
& \leq   \frac{1}{\pi }  \| \chi\|_1 \sup_{E \in {\rm supp}\chi } \mathcal{E}_{n,d,E,B} ( \lambda  ,  \eta  )  +  \left|  \frac{1}{\pi } \int  \int_\R \frac{ \chi(e+\eta t)-\chi(e)  }{t^2 +1}  dt  d  \mu_{\lambda,\infty}(e)  \right| ,
\end{align*} 
where in the  last line we used Theorem  \ref{thm:maintec000}. This shows the claim. 
\end{proof}

\begin{corollary} \label{maincorollary1}   Assume  $m_1 = \E_v v_\gamma = 0$. 
Let $\mathcal{K} \subset (0,\infty)$ be a compact set. Let  $n \in \N$,  $s \in [0,1)$, and  $\eta \in (0,1]$.
Then  for  any  interval $I \subset \mathcal{K}$ and  $\lambda \geq 0$ 
\begin{align} 
& \left|  \mu_{\lambda,\infty}(I)   -  \mu_{0,\infty}(I)  - \frac{1}{\pi}  \int_I \sum_{j=1}^{n-1} \widetilde{{\rm tr}} {\rm Im}  T_j[ e + i \eta] \lambda^j de 
\right| \nonumber  \\
&  \leq  \sup_{E \in \mathcal{K}, \eta \in (0,1]}  \left| \sum_{j=1}^{n-1} \widetilde{{\rm tr}} {\rm Im}  T_j[ E + i \eta] \lambda^j  \right|  4 \delta
+   \frac{1}{\pi }  ( |I| + \delta )  \sup_{E \in \mathcal{K}} \mathcal{E}_{n,d,E,B} ( \lambda  ,  \eta  )   \nonumber \\
& +\mu_{0,\infty}([\inf I -\delta,\inf I+\delta] \cup [\sup I -\delta , \sup I  + \delta]) \nonumber \\
& + \eta C _2(\lambda)(|I|+\delta) + C_{1,s}(\lambda)  (\eta/\delta)^s + (\eta/\delta)^s C_{1,s}(0)  + \eta C_{2}(0) 2 \delta 
 \label{someboundondos} 
\end{align} 
 for  all $\delta  > 0$ such that $[\inf I - \delta , \sup I + \delta]  \subset  \mathcal{K}$, where we defined 
\begin{align}
&   C_{1,s}(\lambda) :=    \frac{ \mu_{\lambda,\infty}(-R,R)     }{\pi }  2^{1-s}\int_\R \frac{|t|^s}{t^2+1} dt \label{constsome1}  \\
& C_2(\lambda) :=    \frac{4}{\pi }   \int_{|x| \geq R}  x^{-2} d  \mu_{\lambda,\infty}(x) . \label{constsome2} 
\end{align} 
for some $R$ such that $\mathcal{K} \subset (-R/2,R/2)$. 
\end{corollary} 
Note that 
the constants in  \eqref{constsome1}  \eqref{constsome2}  can be estimated by means of  Proposition \ref{lem:proofofexdens00-0}.
\begin{proof} We set  $a = \inf I$ and $b = \sup I$.
Take a continuous  cutoff function $\chi_+ \geq 1_I $. Then we find from  Lemma  \ref{corollarygenapprox} 
\begin{align*} 
&   \int  1_I(e)  d  \mu_{\lambda,\infty}(e)   - \int ( \chi_+ * \gamma_\eta)        d \mu_{0,\infty}(e)  - \frac{1}{\pi}  \int \chi_+(e) \sum_{j=1}^{n-1} \widetilde{{\rm tr}} {\rm Im} T_{j,\infty}[ e + i \eta]\lambda^j   de   \\
&  \leq  \int  \chi_+(e)  d  \mu_{\lambda,\infty}(e)   - \int ( \chi_+ * \gamma_\eta)        d \mu_{0,\infty}(e)  - \frac{1}{\pi}  \int \chi_+(e) \sum_{j=1}^{n-1} \widetilde{{\rm tr}} {\rm Im}  T_{j,\infty}[ e + i \eta] \lambda^j  de   \\
&  \leq  \frac{1}{\pi}\| \chi_+\|_1 \sup_{E \in \mathcal{K}} \mathcal{E}_{n,d,E,B} ( \lambda  ,  \eta  )  + \left|  \frac{1}{\pi} \int \int_\R \frac{\chi_+(x+\eta  t)-\chi_+(x)}{t^2 + 1 } dt d \mu_{\lambda,\infty}(x)    \right| 
\end{align*} 
For a bound in the opposite direction choose a continuous $\chi_-$ with  $0 \leq \chi_- \leq 1_I$. Then an analogous estimate 
 using Lemma  \ref{corollarygenapprox} 
\begin{align*} 
&   \int  1_I(e)  d  \mu_{\lambda,\infty}(e)   - \int ( \chi_- * \gamma_\eta)        d \mu_{0,\infty}(e)  - \frac{1}{\pi}  \int \chi_-(e) \sum_{j=1}^{n-1} \widetilde{{\rm tr}}  {\rm Im}  T_{j,\infty}[ e + i \eta] \lambda^j  de   \\
&  \geq  \int  \chi_-(e)  d  \mu_{\lambda,\infty}(e)   - \int ( \chi_- * \gamma_\eta)        d \mu_{0,\infty}(e)  - \frac{1}{\pi}  \int \chi_-(e) \sum_{j=1}^{n-1} \widetilde{{\rm tr}} {\rm Im}  T_{j,\infty}[ e + i \eta]\lambda^j   de   \\
&  \geq - ( \frac{1}{\pi} \| \chi_-\|_1 \sup_{E \in \mathcal{K}} \mathcal{E}_{n,d,E,B} ( \lambda  ,  \eta  )  + |  \frac{1}{\pi} \int \int_\R \frac{\chi_-(x+\eta  t)-\chi_-(x)}{t^2 + 1 } dt d \mu_{\lambda,\infty}(x)    |  ) 
\end{align*} 
Now we choose  
$ \chi_-(x) = [\min\{ \delta^{-1}(x-  a ) , 1 , \delta^{-1} (b -x)  \} ]_+$
and $\chi_+(x) =   [\min\{ \delta^{-1}(x-  a + \delta ) , 1 , \delta^{-1} (b +\delta  -x)  \} ]_+$.  
That is $\chi_+, \chi_- \in C(\R)$ are piecewise linear  functions,  absolutely continuous, and satisfy
 $0 \leq \chi_\pm \leq 1$, 
   $\chi'_\pm = O( \delta^{-1})$. Moreover  $\chi_+$  has support within a $\delta$ distance of $I$, and 
 $\chi_-$ is  one on  the interval $(\inf I + \delta, \sup I - \delta)$. 
Thus, by construction and the assumptions  $\supp \chi_\pm \subset \mathcal{K}$.  
Then we find from Lemma \ref{diffestforint}   that 
\begin{align*}
 &\left|  \int \int_\R \frac{\chi_\pm(x+\eta  t)-\chi_\pm(x)}{t^2 + 1 } dt d \mu_{\lambda,\infty}(x)    \right|   \\
 & 
 \leq \left|  
 \int_{|x| <  R}  \int_\R \frac{\chi_\pm(x+\eta  t)-\chi_\pm(x)}{t^2 + 1 } dt d \mu_{\lambda,\infty}(x)    \right|  
 + \left|  
 \int_{|x| \geq   R}  \int_\R \frac{\chi_\pm(x+\eta  t)-\chi_\pm(x)}{t^2 + 1 } dt d \mu_{\lambda,\infty}(x)    \right|  
  \\
 & \leq  \left|  
 \int_{|x| <  R}  \int_\R \frac{ (\eta |t| \|\chi_\pm'\|_\infty)^s (2 \| \chi_\pm\|_\infty)^{1-s}}{t^2 + 1 } dt d \mu_{\lambda,\infty}(x)    \right| + \left| \int_{|x| \geq   R}  \int_\R \frac{\chi_\pm(s )}{(s-x)^2 + \eta^2 } \eta ds d \mu_{\lambda,\infty}(x) \right| \\
&  \leq ( \eta/ \delta)^s 2^{1-s} \mu_{\lambda,\infty}(-R,R) \int_\R |t|^s(t^2+1)^{-1} dt +  4 \eta  \| \chi_\pm \|_1
\int 1_{|x| \geq R}  x^{-2} d  \mu_{\lambda,\infty}(x) 
\end{align*} 
where we used that for $x \geq R$ and $s \in \supp \chi_\pm \subset  (- R/2,R/2)$ 
we have $(s-x)^2 \geq x^2 + s^2 - 2 s x \geq \frac{1}{2}{x^2}  - s^2 \geq \frac{1}{4} x^2 + \frac{1}{4} R^2 - R^2/4 =  \frac{1}{4} x^2$. 
Now to estimate the difference we use again Lemma  \ref{trivialidentcauchy} as well as Lemma   \ref{diffestforint} 
and an analogous argument (where we first use monotonicity) gives 
\begin{align*} 
&   |  \int ( 1_I  - \chi_\pm )  * \gamma_\eta)        d \mu_{0,\infty}(e)  |  \leq |  \int ( \chi_+  - \chi_- )  * \gamma_\eta)        d \mu_{0,\infty}(e)  |  \\
& =  |  \int( \chi_+ - \chi_- )(x)  d \mu_{0,\infty}(e)  
  +\frac{1}{\pi } \int \int_\R \frac{( \chi_+ - \chi_- )  (x+\eta t)- ( \chi_+ - \chi_- ) (x)  }{t^2 +1}  dt d \mu_{0,\infty}(x)   | \\
  & \leq  \mu_{0,\infty}([a-\delta,a+\delta] \cup [b-\delta , b + \delta]) \\
& \quad + \frac{1}{\pi}\left[ ( \eta / \delta)^s 2^{1-s} \mu_{0,\infty}(-R,R) \int_\R |t|^s(t^2+1)^{-1} dt +  4 \eta  \| \chi_+-\chi_- \|_1
\int_{|x| \geq R}  x^{-2} d  \mu_{0,\infty}(x) \right] . 
\end{align*}
Moreover, by  Theorem \ref{thm:boundaryvelueexpcoeff}
\begin{align*}
& \int  \left| ( 1_I - \chi_\pm)(e) \sum_{j=1}^{n-1} \widetilde{{\rm tr}} {\rm Im}  T_j[ e + i \eta] \lambda^j  \right|  de  
 \leq 
\int_{[a -\delta, a +\delta] \cup [ b -\delta, b +\delta]}   \left| \sum_{j=1}^{n-1} \widetilde{{\rm tr}} {\rm Im}  T_j[ e + i \eta]\lambda^j   \right|  de  \\
&\leq    \sup_{E \in \mathcal{K}, \eta \in (0,1]}  \left| \sum_{j=1}^{n-1} \widetilde{{\rm tr}} {\rm Im}  T_j[ e + i \eta] \lambda^j  \right|  4 \delta 
\end{align*} 
Now \eqref{someboundondos}  follows by collecting  inequalities and a  simple application the triangle inequality. 
\end{proof}

\begin{lemma} \label{diffestforint}  Let $f \in C(\R)$  be a absolutely continuous with $\| f' \|_\infty < \infty$. Then for any $s  \in [0,1]$, 
$\eta \geq  0$, and $x \in \R$ 
\begin{align} 
 & | f(x+\eta  t)-f(x)| \leq ( 2 \| f \|_\infty)^{1-s} ( \eta |t| \|f' \|_\infty)^{s}   . 
\end{align} 
\end{lemma} 
\begin{proof} Using that $f$ is absolutely continuous we find 
\begin{align} 
 &   | f(x+\eta  t)-f(x)| \leq     \int_0^{|t|\eta} | f'(x+ \xi) |d\xi   \leq  \eta | t |  \|f' \|_\infty 
\end{align} 
Furthermore, we have the trivial bound  $| f(x+\eta  t)-f(x)| \leq 2 \| f \|_\infty$.
Interpolation of these two bounds now yields the claim.
\end{proof} 

\begin{remark} {\rm 
As an application of Corollary  \ref{maincorollary1} one can obtain various bounds. 
For instance, let  $\lambda_0 > 0$ and let $\mathcal{K}$ 
be a compact subset of $(0,\infty)$. Then for any $\epsilon \in(0,1]$ 
 there exists a  constant $C$ 
such that  for all $I \subset \mathcal{K}$ 
and $\lambda \in [0,\lambda_0]$ 
\begin{align} \label{hoelderupperlowerbound} 
& \left|  \mu_{\lambda,\infty}(I)   -  \mu_{0,\infty}(I) 
\right|   \leq  C \lambda^{1-\epsilon}   . 
\end{align} 
To derive   \eqref{hoelderupperlowerbound} we choose the constants in Corollary  \ref{maincorollary1}
as follows:  $\delta = c_\delta \lambda^{1-\epsilon/2}$, $\eta = \lambda^{2-\epsilon}$ $s=1-\epsilon/2$,
and $n \in \N$ with $n > 10/\epsilon$. We choose $c_\delta > 0$ sufficiently small to ensure that 
$\delta \leq \inf \mathcal{K}/2$ for all $\lambda \in [0,\lambda_0]$, so that 
$[\inf \mathcal{K} - \delta , \sup \mathcal{K} + \delta ]$ is a compact subset of $(0,\infty)$. 
Moreover, we used that $\mu_{0,\infty}$ the  density of states measure of the free Laplacian 
has a density which is uniformly bounded on compact subsets of $(0,\infty)$, \cite[Page 32]{Stollmann.2001}.
We note that \eqref{hoelderupperlowerbound}  provides an upper and lower bound 
on the integrated density of states in terms of the density of states of the free Laplacian.
 }
\end{remark}

	If one is   merely interested in an upper bound one can proceed as follows. 
	
\begin{lemma} \label{anarctanspecbound} 
	Let $I=[a,b]$ and $\eta>0$. Then we have 
\begin{align} \label{ineqwitharctan1}
	\langle \phi, 1_I(H) \phi\rangle  \,\arctan(|I|/\eta)\le \int_I  \mathrm{Im}\langle \phi, (H-E-i\eta)^{-1} \phi \rangle dE
\end{align}
for any self-adjoint operator $H$ in a Hilbert space $\mathcal{H}$ and  $\phi \in \mathcal{H}$.
\end{lemma}
\begin{proof}
First note that by spectral calculus
\begin{align} \label{speccalc33356}
	 \int_I  \mathrm{Im}\langle \phi, (H-E-i\eta)^{-1} \phi \rangle dE = \left\langle \phi, \arctan\left(\frac{b-H}{\eta}\right) -\arctan\left(\frac{a-H}{\eta}\right) \phi\right\rangle
\end{align}
Now we use the following inequality for $\lambda \in \R$ 
\begin{align} \label{ineqarctan} 
	 \arctan\left(\frac{b-\lambda}{\eta}\right) -\arctan\left(\frac{a-\lambda}{\eta}\right) \ge  1_I(\lambda) \arctan\left(\frac{|I|}{\eta}\right)
\end{align}
To see this, we first observe that by monotonicity of $\arctan$ the left hand side is always non-negative.
So the inequality 
holds  for   $\lambda \notin [a,b]$.  On the other hand  let  $\lambda \in [a,b]$.  
Define for a fixed $t>0$ and $x\in [0,t]$ 
\begin{align*}
	f_t(x) =\arctan(x) -\arctan(x-t) . 
\end{align*}
The  function $x \mapsto f_t(x)$ is continuous and strictly concave on $(0,t)$. Thus the  minimum is attained on the boundary and 
\begin{align*}
	\min_{x\in[0,t]}f_t(x) = f_t(0) = f_t(t) =\arctan{t}.
\end{align*}
Hence for all $x\in [0,t]$
\begin{align} \label{ineqofarctan-2} 
	\arctan(x) -\arctan(x-t) \ge \arctan(t) .
\end{align}
Now inserting  $t= |I|/\eta$ and $x= (b-\lambda)/\eta$ into  \eqref{ineqofarctan-2}  shows \eqref{ineqarctan}
 in the case  $\lambda \in I$. 
 Inequality \eqref{ineqwitharctan1} now follows by means of the spectral theorem
  from \eqref{speccalc33356} and \eqref{ineqarctan}. 
\end{proof}
\begin{corollary}\label{newbound-3}
	Assume $m_1=\mathbf{E}_v v_\gamma=0$. Let $n\in\mathbb{N}$. Let $\eta>0$. Then for a compact interval $[a,b]\subset (0,\infty)$ we find for $\lambda\ge0$
\begin{align}
	\mu_{\lambda,\infty}([a,b]) \le \frac{1}{\arctan((b-a)/\eta)}\int_{a}^{b} \left[\left|\sum_{j=0}^{n-1} \widetilde{\mathrm{tr}}\mathrm{Im} T_{j,\infty}[E+i\eta]\lambda^j\right| +\mathcal{E}_{n,d,E,B}(\lambda,\eta)\right] dE
\end{align}
\end{corollary}
\begin{proof}
 Apply  Lemma  \ref{anarctanspecbound}  to the averaged trace and use Fubini to interchange the integral with expectation
\begin{align}
	|\Lambda_L|^{-1} \mathbf{E}\mathrm{tr} 1_{[a,b]}(H_{\lambda,L}) &\le  \frac{1}{\arctan((b-a)/\eta)}|\Lambda_L|^{-1}\mathbf{E}\int_a^b\mathrm{tr}\, \mathrm{Im}(H_{\lambda,L}-E-i\eta)^{-1} dE \nn  \\
&= \frac{1}{\arctan((b-a)/\eta)}\int_a^b |\Lambda_L|^{-1} \mathbf{E}\mathrm{tr}\, \mathrm{Im}(H_{\lambda,L}-E-i\eta)^{-1} dE  \label{intboundforlemma} 
\end{align}
Now we use   Theorem \ref{thm:maintec000}  and   the triangle inequality to find 
\begin{align*} 
	 |\Lambda_L|^{-1} \mathbf{E}\mathrm{tr}\, \mathrm{Im}(H_{\lambda,L}-E-i\eta)^{-1} \le \left|\sum_{j=0}^{n-1}\widetilde{\mathrm{tr}}\mathrm{Im} T_{j,L}[E+i\eta]\lambda^j\right| +\mathcal{E}_{n,d,E,B}(\lambda,\eta) . 
\end{align*}
and insert this  into the integral in   \eqref{intboundforlemma}. Then, taking the limit $L\to\infty$ and using Theorem   \ref{exinflim} as well as  Proposition  \ref{portmanteux},
we complete the proof. 
\end{proof}

\begin{remark} {\rm  Let us now show an application of Corollary \ref{newbound-3}.   
	We apply the corollary with $\eta=|I|$. Recall the uniformity of the coefficients and the bound for the error term. Then there exists constants $C_{0,n}$ and $ C_{1,n}$, such that
\begin{align}
	\mu_{\lambda,\infty}(I)&\le C_{0,n}|I|+ C_{1,n} \left(\frac{\lambda^2}{|I|}\right)^{n/2} |I|^{-3/2}(1+\ln( \,|I|^{-1} +1))^n
\end{align}
for $0\le \lambda\le \lambda_0$.
Let us now assume as in   \cite{Schenker.2004} that there exists an optimal volume  Wegner estimate, 
\begin{align}
 \mu_{\lambda,\infty}(I)\le C_{\rm Wegner}\frac{|I|}{\lambda} ,  \quad  \lambda > 0  .
\end{align}
 Let $I= [E-\delta, E+\delta]$ with $0<\delta<1$ small enough. Then by Corollary \ref{newbound-3} 
\begin{align*}
	\mu_{\lambda,\infty}(I) \le  \frac{1}{\arctan(2 \epsilon/\eta)}\int_{E-\delta-\epsilon}^{E+\delta+\epsilon}  \left[\left|\sum_{j=0}^{n-1} \widetilde{\mathrm{tr}}\mathrm{Im} T_{j,\infty}[E+i\eta]\lambda^j\right| +\mathcal{E}_{n,d,E,B}(\lambda,\eta)\right] dE
\end{align*}
Choose $\epsilon=\delta^\beta$ for some $0<\beta<1$. Now we estimate for $0\le \lambda\le \lambda_0$ and $\eta=\delta^\beta$
\begin{align} \label{bound0.10}
	\mu_{\lambda,\infty}([E-\delta,E+\delta]) \le C_{n,0} \delta^{\beta}+ C_{n,1} \left(\frac{\lambda^2}{\delta^\beta}\right)^{n/2} \delta^{- \frac{3\beta}{2} }(1+\beta |\ln \delta|)^n
\end{align}
for some constants $C_{n,0}$ and $C_{n,1}$. We have to choose $\delta<1$ small enough, such that 
\begin{align} \label{constraintondelta-893}
[E-\delta-\delta^{\beta},E+\delta+\delta^{\beta}] \subset (0,\infty).
\end{align} 
We compare with the Wegner estimate. The  bound \eqref{bound0.10} will be stronger if
\begin{align*}
	\frac{\delta}{\lambda}\ge\lambda^{n} \delta^{-\beta \frac{n+3}{2} }(1+\beta |\ln \delta|)^n  \quad \text{i.e.} \quad   \delta^{\frac{2+n\beta+3\beta}{2(n+1)}}(1+\beta |\ln \delta|)^{-\frac{n}{n+1}} \ge \lambda .
\end{align*}
Now we insert in the above estimate
\begin{align*}
	\mu_{\lambda,\infty}([E-\delta,E+\delta]) \le C_{n,0} \delta^{\beta}+  C_{n,1} \delta^{\frac{2n-\beta n -3\beta }{2(n+1)}}(1+\beta |\ln \delta|)^{\frac{n}{n+1}} .
\end{align*}
We can estimate further above by multiplying with the same logarithmic factor on the first term. Then we choose $\beta$, such that both have the same order. Doing this we find $\beta= \frac{2n}{3n+5}$  giving
\begin{align*}
	\mu_{\lambda,\infty}([E-\delta,E+\delta]) \le ( C_{n,0}+  C_{n,1})\delta^{\frac{2n }{3n+5}} (1+2n/(3n+5) |\ln \delta|)^{\frac{n}{n+1}}  .
\end{align*}

For the other region we estimate using  the Wegner bound instead and get the same orders up to some constant. Hence there exists a constant $C_n$, such that
\begin{align} \label{schenkerestimate} 
	\mu_{\lambda,\infty}([E-\delta,E+\delta]) \le C_n\delta^{\frac{2n }{3n+5}}  \left(1+\frac{2n}{3n+5}  |\ln \delta|\right)^{\frac{n}{n+1}}
\end{align} 
for all $0\le \lambda\le \lambda_0$ and $\delta < 1$ sufficiently small such that \eqref{constraintondelta-893}
holds. 
We can choose $n$ large such that the exponent will get close to $\frac{2}{3}$.  
This is analogous to the result  in   \cite{Schenker.2004} for a random Schrödinger operator on a lattice.
 In fact, the exponent of $\frac{2}{3}$ obtained by   \eqref{schenkerestimate}  improves the exponent of $\frac{1}{2}$  obtained  in    \cite{Schenker.2004}.
}
\end{remark}

\section{Pairing Potential Labels}   

\label{pairing} \label{PPL}

In this section we give formulas about expectations of products of the random potential, which will be used in the following 
sections.  The main  formulas   have been shown   in \cite{ErdosSalmhoferYau.2008a}.   
We  are going to work with partitions of the set $ \{1,2,...,n\}$ of natural numbers. 
In the definition given in \eqref{defoffinT} we are faced with the expectation  value of powers of the random potentials, i.e.,  terms of the form
\begin{align}
\label{expvCP}
{\bf E}_v^{\otimes M} \ {\bf E}_{y_L}^{\otimes M} \sum_{\gamma_1,...,\gamma_n = 1}^M \prod_{j=1}^n  v_{\gamma_j}.
\end{align}
We are therefore going to consider products of the random potentials $v_{\gamma}$ at the sites $y_{L,\gamma} \in \Lambda_L$. Note,  that for  a given configuration $\{ \gamma_1,...,\gamma_n \}$, different indices $j_1, \ldots , j_l$ may label the same sites, i.e. $\gamma_{j_1} = \ldots = \gamma_{j_l}$.  This case would produce a factor
\begin{align*}
\prod_{i=1}^l  v_{\gamma_{j_i}}=v_{\gamma_{j_1}}^l , 
\end{align*}
since all  $v_{\gamma_j}$  are the same. 
Thus  the expectation value of this factor yields  an $l$-th moment. 

To control this  technically, we introduce the following notation.
Let $\mathcal{A}_n$ be the set of partitions of $ \{1,2,...,n\}$.
For a given configuration $\{ \gamma_1,...,\gamma_n \}$ we are seeking to  filter for the unique partition $A$ that sorts the indices $ 1,2,...,n$ in such a way, that the following two conditions  are satisfied. 
\begin{itemize}
\item[(a)] All indicies in each $a \in A$ correspond to the same site.
\item[(b)] All sets $a$ correspond to different sites.
\end{itemize}
To express this, we define for $A \in \mathcal{A}_n$ the two functions $\chi_A$ and $\tilde{\chi}_A$ on $\N^n$
as follows.   First,  we set  
\muun{
\label{X1}
\chi_A(z_1,....,z_n) := \prod_{a \in A} \bigg[ \prod_{(j,l) \in a \times a } 1_{\{  z_j = z_l \}} \bigg] ,
}
which ensures condition (a).   Second,  we define  
\muun{
\label{X2}
\tilde{\chi}_A(z_1,....,z_n) := \chi_A(z_1,...,z_n)  \prod_{\substack{ ( a, b )   \in A \times A :  \\ a \neq b}}  \bigg[ \prod_{(j,l) \in a \times b} 1_{\{  z_j \neq  z_l \}} \bigg]   ,
}
which ensures that condition (b) is additionally met.
Observe that for any $(z_1,...,z_n) \in \N^n$ we have 
\begin{align}
\label{upartition}
\sum_{A \in \mathcal{A}_n}  \tilde{\chi}(z_1,...,z_n)  = 1,
\end{align}
since there is exactly one unique partition that satisfies  both conditions (a) and (b).
With this small preamble, we can now express the expected value in term (\ref{expvCP}) and show the following lemma.

\begin{lemma} \label{lem:combprob1} 
For any fixed $L > 0$,  integers  $n, M  \in \N$,  and any fixed momenta $q_j \in \Lambda_L^*$, $j=1,...,n$,  the identity 
\begin{align*}
 {\bf E}_v^{\otimes M} \ {\bf E}_y^{\otimes M} \sum_{\gamma_1,...,\gamma_n = 1}^M \prod_{j=1}^n  v_{\gamma_j} \exp(2 \pi i q_j \cdot y_{L,\gamma_j}) =  \sum_{A \in \mathcal{A}_n} 1_{\{ |A| \leq M  \}} \frac{M!}{(M-|A|)!}
\prod_{a \in A} \left\{ m_{|a|} \delta_* \left( \sum_{l \in a} q_l \right)  \right\} 
\end{align*} 
holds.
\end{lemma} 
For a proof we refer the reader to  \cite{ErdosSalmhoferYau.2008a} see also \cite{HaslerKoberstein.2025-2}.
We will use the following Lemma. 

\begin{lemma}\label{lemP}  For a Poisson distributed random variable  $N$ with mean $\lambda$ we  
 have for any $k \in \N$ 
\begin{align} 
\E \left[ \prod_{j=0}^{k-1} (N-j) \right]   
& = \lambda^k . \label{eq:poissonexp} 
\end{align} 
\end{lemma} 
We note that  Lemma \ref{lemP}   can alternatively be proven  using  the probability-generating function,
see for example \cite[Section 1.1]{LastPenrose.2018}  or \cite{Kingman.1993}.
 The following lemma is a continuation of Lemma \ref{lem:combprob1} in the sense that we evaluate the $\E_M$ expectation as well.

\begin{lemma} \label{lem:combprob2} Let $L > 0$,  $n \in \N$, and  $q_j \in \Lambda_L^*$, $j=1,...,n$.  Then
\begin{align*}
{\bf E}_M  {\bf E}_v^{\otimes M} \ {\bf E}_{y_L}^{\otimes M} \sum_{\gamma_1,...,\gamma_n = 1}^M \prod_{j=1}^n  v_{\gamma_j} \exp(2 \pi i q_j \cdot  y_{L,\gamma_j}) =  \sum_{A \in \mathcal{A}_n} 
\prod_{a \in A} \left\{ m_{|a|} \delta_{*,L} \left( \sum_{l \in a} q_l \right)  \right\} 
\end{align*}   
holds.
\end{lemma} 
\begin{proof}
The statement follows from \Cref{lem:combprob1} and the identity 
$$
{\bf E}_M   1_{|A| \leq M } \frac{M!}{(M-|A|)!} =  |\Lambda_L|^{|A|},
$$
which follows from  \Cref{eq:poissonexp}. 
\end{proof} 
To transform into momentum space we shall expand with respect to the ONB $\varphi_p$ defined in \eqref{eq:deofonb}. That is we  shall  insert the following  identity, which holds in strong operator topology 
\begin{equation} \label{fourierident} 
\psi= \sum_{p \in \Lambda_L^*}  \SP{\varphi_p ,  \psi}   \varphi_p   =  \int_{\Lambda_L^*}  \SP{e^{  2 \pi i \inn{ p  , \cdot } } , \psi (\cdot)} e^{  2 \pi i \inn{ p , \cdot}  } dp,
\end{equation} 
for all $\psi \in L^2 (\Lambda_L)$, which provides  the  representation via the ONB $\{\varphi_p : p \in \Lambda_L^* \}$. 
Moreover, for $\psi\in L^2 (\Lambda_L)$ and $ W \in L^\infty (\Lambda_L)$ (since $\Lambda_L$ is compact they are both also in $L^1 (\Lambda_L)$)  we will use the identities  
\begin{align} 
 \langle e^{  2 \pi i \inn{ p  , \cdot  } },  \psi \rangle & = \int_{\Lambda_L}  e^{ - 2 \pi i p \cdot x } \psi(x) dx = \widehat{\psi}(p) \label{ftident1}  \\
\langle e^{  2 \pi i \inn{ p , \cdot }  },   W(\cdot)   e^{  2 \pi i \inn{  q , \cdot } } \rangle & =   \int_{\Lambda_L}    W(x)   e^{ - i 2 \pi (p-q)\cdot x } dx  = \widehat{W}(p-q). \label{ftident2} 
\end{align} 
 First we use the probabilistic structure of the potential \eqref{bumpnotationres} and calculate  the Fourier transform   
\muun{ \label{bumpnotationresdd}   
\widehat{ V}_{L,\omega}(p) 
&=  \int_{\Lambda_L} \sum_{\gamma=1}^M  v_\gamma B_\#(x-y_{L,\gamma})  e^{  -2 \pi i p \cdot x } d x = \sum_{\gamma=1}^M v_\gamma \int_{\Lambda_L}  B_\#(z)  e^{ - 2 \pi i p \cdot (z+y_{L,\gamma})} d z \nn \\
&=  \sum_{\gamma=1}^M  v_\gamma \widehat{B}_\#(p) e^{ -2 \pi  i p \cdot  y_{L,\gamma}} .
}
In the following sections we will make use of the following lemma. 
 
\begin{lemma} \label{expProd00123} For the model introduced in Section \ref{Model} we have 
\begin{align}
   { \bf E}_L \left(   \prod_{j=1}^n  \hat{V}_{L,\omega}(p_j - p_{j+1} )    \right)
& =  \sum_{A \in \mathcal{A}_n}  \prod_{a \in A} \left\{ m_{|a|}  \delta_{*,L} \left( \sum_{l \in a} (p_l - p_{l+1}) \right) 
  \prod_{l \in a}  \hat{B}_\#(p_l - p_{l+1} ) \right\} \nn \\
& = \sum_{A \in \mathcal{A}_n}  \mathcal{P}_{A,L}(p_1,...,p_{n+1})
 .  \label{eq:ofexpofrandpot2} 
\end{align} 

\end{lemma} 

For a proof we refer the reader to  \cite{ErdosSalmhoferYau.2008a} see also \cite{HaslerKoberstein.2025-2}.

\section{Expectation of Expansion Coefficients}

\label{sec:prooasymexp}

In this section, we prove Theorem   \ref{lem:formulaforexpcoeff} and lay the ground work for the proofs of Theorems     \ref{exinflim}   and   \ref{thm:boundaryvelueexpcoeff}.
We introduce the notation for the canonical norm in the $\ell^q$-spaces.  For  $f \in \ell^q(\Lambda_L^*)$ 
with $q  \in [1,  \infty]$,
we write 
  \begin{align*} \| f \|_{*, q} = \left( \int_{\Lambda_L^*} |f(k)|^q dk \right)^{1/q}  \ ( 1 \leq q < \infty)   , \quad   \| f \|_{*, \infty } = \sup_{k \in \Lambda_L^*} |f(k)| .
  \end{align*}  
We start with the following Lemma which expresses the  terms in the Neumann expansion with respect 
to the eigenbasis \eqref{eq:deofonb}.

\begin{lemma} \label{finitetracesumm} Let $d \leq 3$. 
There exists a  constant $K_{0,E}$  uniformly bounded for $E$ in compact subsets of $\R$ such that   
\begin{align} \label{traceofboundonfreeres} 
  \int_{\Lambda_L^*} |\nu( k) -  z|^{-2}  dk \leq  K_{0,E} ( 1   + {\rm dist}(z,[0,\infty))^{-2} ) 
\end{align} 
holds for all $L \geq 1$ and  $z \in \C \setminus [0,\infty)$ with ${\rm Re} z = E$.
\end{lemma}

\begin{proof}
Let $I $ be a compact interval of $\R$. 
Using Lemma  \ref{lem:discest} we find  that there exists a constant $C_I$  such that 
\begin{align} \label{eq:defofCpi-1.1}
 &   \int_{\Lambda_L^*} |\nu( k) -  z|^{-2}  dk 
 =   \int_{\Lambda_L^*}   \frac{1 }{   ( \nu(p) - {\rm Re} z )^2   +   ({\rm Im } z)^2}    dp
 \leq    C_I (   1   + {\rm dist}(z,[0,\infty))^{-2} )  
\end{align}
for  all  $L \geq 1$ and  $z \in \C \setminus [0,\infty)$ with ${\rm Re } z \in I$.  This shows \eqref{traceofboundonfreeres}.
\end{proof}

\begin{lemma}\label{HilbertSchmidt}  Let $d \leq 3$. For $z \in \C \setminus [0,\infty)$ the operator $R_L(z)$ is Hilbert-Schmidt
and 
\begin{align} \label{freereshs} 
 \widetilde{{\rm tr}} (  R_L(z)^* R_L(z) )  =  \int_{\Lambda_L^*} |\nu( k) -  z|^{-2}  dk < \infty ,
\end{align} 
furthermore for $z \in \C \setminus \R$ 
\begin{align} \label{freereshs-old} 
\frac{1}{{\rm Im} z}\widetilde{{\rm tr}}{\rm Im} R_L(z)    =  \widetilde{{\rm tr}} (  R_L(z)^* R_L(z) ) .
\end{align} 
\end{lemma}

\begin{remark}{\rm 
We see that   $d \leq   3$ is necessary to see that the sum in  \eqref{freereshs} converges. In case we had $d  > 3$, 
one would have to work with higher powers of the resolvents, which would make the analysis notationally much 
more involved. 
}
\end{remark}

\begin{proof} We note that the r.h.s of  \eqref{freereshs} 
is finite in view of   Lemma \ref{finitetracesumm}. 
For the o.n.b. $(\varphi_p)_{p \in \Lambda_L^*}$    introduced in \eqref{eq:deofonb}.
 we find from the spectral theorem 
\begin{align*}
 \langle \varphi_p , R_L(z)^* R_L(z) \varphi_p \rangle = | \nu(p) - z |^{-2}
\end{align*} 
and so 
\begin{align*} \widetilde{{\rm tr}}| R_L(z)|^2 & =
 |\Lambda_L|^{-1} \sum_{p \in \Lambda_L^* } \langle \varphi_p , | R_L(z)|^2  \varphi_p \rangle  =  |\Lambda_L|^{-1} \sum_{p \in \Lambda_L^* } | \nu(p) - z |^{-2} \\
& = \int_{\Lambda_L^*} |\nu( k) -  z|^{-2}  dk .
\end{align*} 
This proves  \eqref{freereshs}.  To show   \eqref{freereshs-old}, we use the spectral theorem and the resolvent identity 
\begin{align*}
{\rm Im} R_L(z) & = \frac{1}{2i} (R_L(z)  - R_L(z)^*) = \frac{1}{2i} (R_L(z)  - R_L(\overline{z})) = 
\frac{1}{2i}  R_L(\overline{z}) (z-\overline{z}) R_L(z)  \\
& = {\rm Im} z R_L(z)^* R_L(z) .
\end{align*} 
\end{proof}

\begin{lemma}  \label{fourierspacekernel}  Let  $z \in \C \setminus [0,\infty)$. 
Then for all  $n \in \N$ the operator $R_L(z) [  V_L  R_L(z)  ]^n$ is trace class  and 
\begin{align} \label{defoffinTnoexp-0} 
& \widetilde{{\rm tr}} R_L(z) [  V_L  R_L(z)  ]^n \\ 
& =   \int_{(\Lambda_L^*)^{n+1}} \prod_{j=1}^n \widehat{V}(k_j-k_{j+1})    
 \prod_{j=1}^{n+1}  (\nu( k_j) -  z)^{-1}  \delta_{k_1,k_{n+1} }  d (k_1, \ldots , k_{n+1})  . \nn
\end{align}
The sum converges absolutely and satisfies  the bound 
\begin{align} \label{defoffinTnoexp-abs} 
  & \int_{(\Lambda_L^*)^{n+1}} \prod_{j=1}^n| \widehat{V}(k_j-k_{j+1})  |  
 \prod_{j=1}^{n+1}  |\nu( k_j) -  z|^{-1}   \delta_{k_1,k_{n+1} }  d (k_1, \ldots , k_{n+1})  \\
& \leq   |{\rm dist}(z,[0,\infty))|^{-n+1}   \| \widehat{V} \|_{1,*}^n  \int_{\Lambda_L^*} |\nu( k) -  z|^{-2}  dk, \nn
\end{align}
where we used the notation in \eqref{defofdiscretefouriersum}. 
\end{lemma}

\begin{proof}
The trace class property follows since $R_L(z)$ is Hilbert-Schmidt, cf.  Lemma \ref{HilbertSchmidt},  the fact that 
Hilbert-Schmidt operators are an ideal and that  the product of two Hilbert-Schmidt operators is trace class. 
To calculate the trace we  use  the o.n.b. $(\varphi_p)_{p \in \Lambda_L^*}$,
which constitutes an  eigenbasis  
of   $R_L(z)$.   Using the definition of Fourier transform of the potential, cf.   \eqref{ftident2}, and the notation  \eqref{delta-0} we find as an iterated summation 
\begin{align*}
&   |\Lambda_L|^{-1}  \sum_{p}  \inn{ \varphi_p,    R_L(z) [  V_L  R_L(z)  ]^n \varphi_p }  \\
& = |\Lambda_L|^{-1}  \sum_{p}\sum_{k_1} \cdots \sum_{k_{n+1}}  \langle  \varphi_p , R_L(z) \varphi_{k_1} \rangle  \langle \varphi_{k_1} , V_L  \varphi_{k_2}  \rangle \langle  \varphi_{k_2} ,  R_L(z)  \varphi_{k_2} \rangle \langle \varphi_{k_2} ,  V_L \varphi_{k_3} \rangle   \\
& 
\cdots \times  \langle  \varphi_{k_{n}} ,  R_L(z)  \varphi_{k_{n}} \rangle \langle  \varphi_{k_{n}}  , V_L \varphi_{k_{n+1}} \rangle \langle 
\varphi_{k_{n+1}} ,  R_L(z)   \varphi_p
\rangle \\
& = |\Lambda_L|^{-n-1} \sum_{p}  \sum_{k_1} \cdots \sum_{k_{n+1}}\delta_{p,k_1}
   (\nu(k_1) - z)^{-1}   \widehat{V}_L(k_1-k_2)   (\nu(k_2) - z)^{-1}   
\widehat{V}_L(k_{2} - k_{3})  \\ 
&   \cdots \times  (\nu({k}_{n}) - z)^{-1}  \widehat{V}_L({k}_{n} - {k}_{n+1}) 
  (\nu({k}_{n+1}) - z)^{-1} \delta_{p, k_{n+1}}  \\
& = |\Lambda_L|^{-n-1}  \sum_{k_1} \cdots \sum_{k_{n+1}}
   (\nu(k_1) - z)^{-1}   \widehat{V}_L(k_1-k_2)   (\nu(k_2) - z)^{-1}   
\widehat{V}_L(k_{2} - k_{3})  \\ 
&   \cdots \times  (\nu({k}_{n}) - z)^{-1}  \widehat{V}_L({k}_{n} - {k}_{n+1}) 
  (\nu({k}_{n+1}) - z)^{-1} \delta_{k_1, k_{n+1}}  . 
\end{align*}
Now writing this in terms of the notation \eqref{deltaL}  and  \eqref{defofdiscretefouriersum} 
we obtain 
\begin{align} \label{defoffinTnoexp} 
   & |\Lambda_L|^{-1}  \sum_{p}  \inn{ \varphi_p,    R_L(z) [  V_L  R_L(z)  ]^n \varphi_p }  \\  & = \prod_{j=1}^{n+1}  \int_{\Lambda_L^*} \prod_{j=1}^n \widehat{V}(k_j-k_{j+1})    
 \prod_{j=1}^{n+1}  (\nu( k_j) -  z)^{-1}   \delta_{k_1, k_{n+1} } d k_{n+1} \cdots d k_1   .
\end{align}
If  the sum converges absolutely we can sum in any order and obtain  \eqref{defoffinTnoexp-0}. 
Thus it remains to show \eqref{defoffinTnoexp-abs}. For this we estimate using    Lemma \ref{finitetracesumm}.
\begin{align*} 
  &\prod_{j=1}^{n+1}  \int_{\Lambda_L^*} \prod_{j=1}^n |\widehat{V}(k_j-k_{j+1})   | 
 \prod_{j=1}^{n+1}  |\nu( k_j) -  z|^{-1}   \delta_{k_1, k_{n+1} } d k_{n+1} \cdots d k_1   \\
& \leq  \int_{(\Lambda_L^*)^{n+1}} \prod_{j=1}^n| \widehat{V}(k_j-k_{j+1})  |  
 |\nu( k_1) -  z|^{-2}    |{\rm dist}(z,[0,\infty))|^{-n+1}    d (k_1, \ldots , k_{n+1}) \\
&  \leq  |{\rm dist}(z,[0,\infty))|^{-n+1}   \| \widehat{V} \|_{1,*}^n  \int_{\Lambda_L^*} |\nu( k) -  z|^{-2}  dk .
\end{align*}
\end{proof}

Next we want to calculate the expectation of   \eqref{defoffinTnoexp-0} by taking 
the expectation into the integral  and use  Lemma  \ref{expProd00123}.
To justify the interchange of integration we use Fubini's theorem which is 
justified by means of \eqref{defoffinTnoexp-0} and the bound of the following Lemma. 

\begin{lemma} \label{boundonexppotenyt}  For $m \in \N$ there exists a constant $C_m$ such that 
\begin{align}
{\bf E}_L ( \|\widehat{V}_{L}\|_{1,*}^{m}  )  
 \leq C_m  |\Lambda_L|^m \|     \widehat{B}_\# \|_{1,*}^{m} {\bf E}_v  (|v|^{m} ) 
\end{align} 
 for all $L \geq 1$. 
\end{lemma} 
\begin{proof}
The bound is  obtained by inserting 
the definition \eqref{bumpnotationres} using the triangle inequality and convexity 
\begin{align}
{\bf E}_L ( \|\widehat{V}_{L}\|_{1,*}^{m}  )  & \leq  {\bf E}_L  \left( \sum_{\gamma=1}^M |v_\gamma| \|     \widehat{B}_\# \|_{1,*}\right)^{m}    \leq  {\bf E}_L  \left( \sum_{\gamma=1}^M |v_\gamma| \right)^{m}  \|     \widehat{B}_\# \|_{1,*}^{m} \nn  \\
& \leq  {\bf E}_L  M^{m-1} \sum_{\gamma=1}^M    |v_\gamma|^{m}  \|     \widehat{B}_\# \|_{1,*}^{m}
\leq  \|     \widehat{B}_\# \|_{1,*}^{m} {\bf E}_v (|v|^{m} ) {\bf E}_M ( M^{m} ) . \label{expofV} 
\end{align} 
Now it follows from  Eq. \eqref{eq:poissonexp} of Lemma \ref{lemP}  that there exists a constant $C_m$ such that for Poisson distributed 
random variable  $N$ with mean $\lambda \geq 1$ 
\begin{align} 
{\bf E} (N^{m}   ) 
 \leq  C_m  \lambda^m   . 
\end{align} 
This gives for $L \geq 1$ that 
\begin{align} 
{\bf E}_M (M^{m}   ) 
 \leq   C_m |\Lambda_L|^m  . 
\end{align} 
Inserting this into \eqref{expofV}    yields the claim. 
\end{proof} 

This leads immediately to the following theorem. To formulate it  we introduce the following notation. 
For $A  \in \mathcal{A}_n $ and $z \in \C \setminus [0,\infty)$ we define 
\begin{align} \label{eq:defofCpi-0}
\widetilde{C}_{n,A,L}[z ] & :=      \int_{(\Lambda_L^*)^{n+1}}   \mathcal{P}_{A,L}(k_1, \ldots , k_{n+1})    
 \prod_{j=1}^{n+1}  (\nu( k_j) -  z)^{-1}   \delta_{k_1,k_{n+1}}   d (k_1, \ldots , k_{n+1})  , 
\end{align}
recalling the definition in  \eqref{eq:defofdeltapi0}.

\begin{theorem} \label{lem:formulaforexpcoeff-1} 
Let  $z \in \C \setminus [0,\infty)$.    Then 
 we have 
\begin{align} 
   {\bf E}_L
  \widetilde{\rm tr}   R_L(z) [  V_L  R_L(z)  ]^n  
 =     \sum_{A \in \mathcal{A}_n}  \widetilde{C}_{n, A, L }[ z ]  \label{TintermsofC} 
\end{align}
and the sum  \eqref{eq:defofCpi-0} converges absolutely.
\end{theorem}

\begin{proof}   Calculate the expectation of   \eqref{defoffinTnoexp-0}, using Fubini justified 
by means of \eqref{defoffinTnoexp-abs} and  Lemma \ref{boundonexppotenyt} 
we find with Lemma  \ref{expProd00123}
\begin{align} \label{defoffinTnoexp} 
&{\bf E} \widetilde{{\rm tr}} R_L(z) [  V_L  R_L(z)  ]^n \\ 
& =   \int_{(\Lambda_L^*)^{n+1}}  {\bf E}  \prod_{j=1}^n \widehat{V}(k_j-k_{j+1})    
 \prod_{j=1}^{n+1}  (\nu( k_j) -  z)^{-1}  \delta_{k_1,k_{n+1} }  d (k_1, \ldots , k_{n+1})   \nn \\
& =     \int_{(\Lambda_L^*)^{n+1}}   \sum_{A  \in \mathcal{A}_n} \mathcal{P}_{A,L}(k_1, \ldots , k_{n+1})    
 \prod_{j=1}^{n+1}  (\nu( k_j) -  z)^{-1}   \delta_{k_1,k_{n+1}}   d (k_1, \ldots , k_{n+1})   . \nn
\end{align}
Furthermore, the absolute convergence follows from  \eqref{defoffinTnoexp-abs} and  Lemma \ref{boundonexppotenyt} and Fubini's theorem. 
\end{proof}

\section{Integral Formula for Expansion Coefficients} 

\label{sefurtherintformuilas} 

Before we can show the other two theorems we will need some preparatory work. 
For this we introduce  the notation  
 \begin{align}
 \label{Klammern}
 \langle x \rangle = (1 + x^2)^{1/2},
 \end{align} which we shall use throughout the paper.

 The remaining part of this section is devoted to the proofs  of  Theorems  \ref{exinflim} and   \ref{thm:boundaryvelueexpcoeff},
 which state that the infinite volume limit for each expansion coefficient exists and that for this  infinite volume limit we can take the limit  as the spectral 
 parameter approaches the positive real axis from the complex upper or lower half plane.

Technically,  it is more convenient to first sum over the discrete delta functions and then take the limit $L \to \infty$, rather than taking  the limit first and then 
integrating out the delta distributions, cf.  Corollary \ref{remdeltdirac} below.\\ 

\label{sec:transformU} 
To  sum over the discrete delta functions, we make a change of the  summation variables. 
First we introduce the  change of variables given by
\begin{align} \label{changevar1}
u_0  = k_1 ,  \qquad  u_{s} = k_{s+1} - k_s,  \quad s=1,...,n
\end{align}
to be able to resolve the discrete delta functions one at a time.
Expressing the variables $k$  in terms of the variables $u = (u_0,  \ldots u_n) \in (\R^d)^{n+1}$ we find
$$
k_j = \sum_{l=0}^{j-1} u_l .
$$
Applying this change of variables to  \eqref{eq:defofCpi-0} and  inserting  \eqref{eq:defofdeltapi0} produces 
\begin{align} 
 \widetilde{C}_{n,A,L}[z ] & =   \int_{(\Lambda_L^*)^{n+1}}   \prod_{ a \in A} \left\{  \delta_{*,L} \left( \sum_{s  \in a } u_s  \right) \prod_{s \in a} \widehat{B}_\#(-u_s)  \right\}    \nonumber  \\ & \quad  
\prod_{j=1}^{n+1}  \left( \nu \left( \sum_{l=0}^{j-1} u_l \right)  -z  \right)^{-1} \delta_{u_0, \sum_{l=0}^{n} u_l } d (u_0,  \ldots, u_n) , \quad z \in \C \setminus [0,\infty) \label{eq:defofCpi}
\end{align}
For a given partition $A$ (hence given $a \in A$) we are now going to introduce a function $M_A$ to express the dependence of indices belonging to the same $a \in A$ which is caused by the term 
$$
\delta_{*,L}\left( \sum_{s  \in a } u_s  \right). 
$$
This term  vanishes unless  $\sum_{s  \in a } u_s  =0$, which means that for every $a \in A$ we can express one of the variables by the negative sum of the others,  or  in the case,  where $a$ consists of only one element,  that element has to be equal to $0$.    
We choose to express the highest number of each set $a \in A$ by the negative sum of all the other numbers, which gives us 
\begin{align*}
u_{\max a}=- \sum_{j \in a \setminus \{ \max a \} } u_j .
\end{align*}
  In the case where $a$ consists of only one element,  $ a \setminus \{ \max a \}= \emptyset$ which produces an empty sum, which is $0$ by convention.
To make the notation more usable we define the set of all  indices, which are the maximum of a set $a$ by

\begin{align} \label{defofJA}
J_A := \{ \max a : a \in A\}
\end{align}
as well as its complement
\begin{align}\label{defofIA}
I_A := \{ 1 ,  \ldots ,  n \} \setminus J_A.
\end{align}
Since $A$ is a partition, for any $j \in \{1 , \ldots , n\}$ there is a unique set $a \in A$ such that $j \in a$, we denote this set by $a (j)$. 
We define the map $M_A :( \R^d)^{| I_A |} \to (\R^d)^n$ as 
\begin{align}\label{changevar1next} 
[M_A (v)]_j :=
\begin{cases}
v_j &: j \in I_A
\\
- \sum\limits_{ l \in a (j)  \setminus \{ j \} }  v_l &: j \in J_A,
\end{cases}
\end{align}
where $v=(v_j)_{j \in I_A}$ with  $v_j  \in \R^d$ for all $j \in I_A$. 
Note that \eqref{changevar1next} contains the case in which $j$ is the only element of $a(j)$ and  $[M_A (v)]_j =0$ holds.
This definition implies the following lemma.

\begin{lemma} \label{sum0}
Let $A$ be a partition of  $\{ 1 ,  \ldots ,  n \}$  and $M_A$ be defined as in  \eqref{changevar1next}.
Then we have
\begin{align}\label{totsumzero} 
\sum_{j=1}^n  [M_A (v)]_j = 0
\end{align}
for all $v  \in ( \R^d)^{| I_A |}$. 
\end{lemma}
The proof is straightforward to see. For details we refer the reader to \cite{HaslerKoberstein.2025-2}. 
With the notation  introduced  \eqref{changevar1next} summing over  the variables in  \eqref{eq:defofCpi} and evaluating the delta function we obtain the identity of the following theorem.
Here and henceforth we adopt a notation where we write the integration variables right after the integral sign for notational compactness.

\begin{theorem} \label{thm:dosasy0-0-Chat}  
 Let  $z \in \C \setminus [0,\infty)$ and  $n \in \N$. Then  
\begin{align*}  
  {\bf E}_L   \widetilde{\rm tr}    R_L(z) [  V_L  R_L(z)  ]^n   = \sum_{A \in \mathcal{A}_n}  \widetilde{C}_{n, A, L }[ z  ]   ,  
\end{align*} 
where for $A  \in \mathcal{A}_n $
\begin{align} \label{eq:defofCpi2}
& \widetilde{C}_{n,A,L}[z ]  \\
 & =   \int_{\Lambda_L^*} du_0 \prod_{l \in I_A} \left( \int_{\Lambda_L^*}  dv_l \right)   \prod_{j  = 1}^n   \widehat{B}_\#(-[M_A(v)]_j)               
 \prod_{j=1}^{n+1}  \left( \nu \left( u_0 +   \sum_{l=1}^{j-1} [M_A v]_l  \right) - z  \right)^{-1}  . \nonumber
\end{align}
and the sum converges absolutely. 
\end{theorem}

\section{Bounds on Expansion Coefficients} 

\label{sec:boundsonexpcoeff}

In this section we will use  the explicit formula in  \eqref{eq:defofCpi2} to obtain 
 bounds, which are better than the bounds in   \eqref{defoffinTnoexp-abs}.

\begin{lemma}\label{thm:dosasy0-0}   \label{boundonCsDOS} 
For $n \in  \N$  and  any compact set  $I \subset \R$
there exists a  constant $K_{n,I}$   such that   
\begin{align} \label{boundonTs-2-0}
& | \widetilde{C}_{n,A,L}[z]  | \leq   \frac{ K_{n,I}  ( 1  + {\rm dist}(z,\R_\geq)^{-2} )  }{ {\rm dist}(z,\R_\geq)^{n-1}}  ,  \quad   n \geq 1  
\end{align} 
holds for all $L \geq 1$,   $z \in \C \setminus [0,\infty)$ with ${\rm Re} z \in I$,  and  $A \in  \mathcal{A}_n $.
\end{lemma}

\begin{proof} 
We use  identity    \eqref{eq:defofCpi2}.
Let $z \in \C \setminus [0,\infty)$ and $A \in \mathcal{A}_n$.  
Then for  all 
resolvents except the first and the last we use the trivial bound    \eqref{ln-2}   (recall  that by Lemma  \ref{sum0} we
have   $\sum_{l=1}^{n} [M_Av]_l = 0$)
we find   summing over  all $v_l$ variables that 
\begin{align} \label{eq:exoflimit02}
&  \left|   \widetilde{C}_{n,A,L}[z   ]   \right|    \\
& =   \int_{\Lambda_L^*} dq  \left| \prod_{l \in I_A} \left( \int_{\Lambda_L^*}  dv_l \right)   \prod_{j  = 1}^n   \widehat{B}_\#(-[M_A(v)]_j)    \prod_{j=1}^{n+1}  \left( \nu \left( q +   \sum_{l=1}^{j-1} [M_A v]_l  \right) - z \right)^{-1} \right|   \nonumber \\
& \leq    {\rm dist}(z, \R_\geq)^{-n+1} \| \hat{B}_\# \|_{*,\infty}^{n-|I_A|}  \int_{\Lambda_L^*} dq  \prod_{l \in I_A} \left( \int_{\Lambda_L^*}  dv_l \right)   \prod_{\substack{ j  = 1 \\ j \in I_A} }^n  | \widehat{B}_\#(v_j)   |
|  \nu(q ) - z |^{-2}       \nn \\
& \leq  {\rm dist}(z,\R_\geq )^{-n+1} \| \hat{B}_\# \|_{*,\infty}^{n-|I_A|} \| \hat{B}_\# \|_{*,1}^{|I_A|}   \int_{\Lambda_L^*}  
| \nu (q ) - z |^{-2}    dq    \nn \\  
& \leq  {\rm dist}(z,\R_\geq)^{-n+1} \| \hat{B}_\# \|_{*,\infty}^{n-|I_A|}   \| \hat{B}_\# \|_{*,1}^{|I_A|}  C_I  ( 1  + {\rm dist}(z,\R_\geq)^{-2} ) , \nn 
\end{align}
for some constant $C_I$,  which exists by  Lemma  \ref{lem:discest} (with $\tau=0$),
uniformly in $L \geq 1$ and ${\rm Re } z \in I$. Note that  $\| \hat{B}_\# \|_{*,1}$ as well as $\| \hat{B}_\# \|_{*,\infty}$ are both dependent on $L$ indicated by the indexed $\#$.
Now to find  $L$-independent bounds,  we use Lemma \ref{lemestfourdisc}, which gives us an $L$-independent bound on $\| \hat{B}_\# \|_{*,1}$.  By Lemma  \ref{fourierseries}
\end{proof}

\section{Infinite Volume Limit of Expansion Coefficients} 

\label{secnine}

In this section we take the infinite volume limit of the expressions in \eqref{thm:dosasy0-0}.

\begin{proposition} \label{propexofcinfDOS} Let $d \leq 3$. Let  $z \in \C \setminus [0,\infty)$ and $A \in \mathcal{A}_n$.
Let   $n \in \N$. Then  the following limits exist 
\begin{align} \label{eq:exoflimit2DOS}
&   \widetilde{C}_{n,A, \infty }[z ] :=  \lim_{L \to \infty} \widetilde{C}_{n,A,L}[z]   \\
& =   \int_{\R^d} du_0 \prod_{l \in I_A} \left( \int_{\R^d}  dv_l \right)   \prod_{j  = 1}^n   \widehat{B}(-[M_A(v)]_j)   \prod_{j=1}^{n+1}  \left( \nu \left( u_0 +   \sum_{l=1}^{j-1} [M_A (v)]_l  \right) - z \right)^{-1}  ,  \nonumber
\end{align}
where $M_A$ is defined in  \eqref{changevar1next} and the integral converges absolutely. 
\begin{align} \label{eq:exoflimit2DOS}
& \widetilde{\rm tr}  {T}_{n,\infty}[z]   :=  \lim_{L \to \infty}  {\bf E}_L  \widetilde{\rm tr}  R_L(z) [  V_L  R_L(z)  ]^n  
= \sum_{A \in \mathcal{A}_n}     \widetilde{C}_{n,A, \infty }[z ] 
\end{align}
The following limit exists and equals 
\begin{align} \label{tracexpcoeff00-0-0} 
 \widetilde{\rm tr}  {\rm Im} {T}_{0,\infty}[z] & := 
 \lim_{L \to \infty} \widetilde{\rm tr}  {\rm Im} {T}_{0,L}[z]  = 
 \lim_{L \to \infty}\int_{\Lambda_L^*}  \frac{ {\rm Im } z }{ (\nu(p) - {\rm Re} z )^2 +  ({\rm Im } z)^2} dp \\
& = \int_{\R^d}  \frac{ {\rm Im } z }{ (\nu(p) - {\rm Re} z )^2 +  ({\rm Im } z)^2} dp  . \nn
\end{align} 
\end{proposition} 
\begin{proof}
This will be shown using the dominated convergence theorem writing the sum as an integral over  simple functions and by finding a suitable integrable majorant. 
   For a function $f : (\Lambda_L^*)^{m} \to \C$, $m \in \N$ we define an extension  $\mathcal{E}_L[f] :  (\R^d)^{m}  \to \C$ by $\mathcal{E}_L[f](x) = f(k)$ if 
$x_j = k_j + \xi_j$ for $k_j \in \Lambda_L^*$ and $\xi_j \in [-\frac{1}{2L},\frac{1}{2L})^d$, $j=1,...,m$.  With this definition we can express   the sum in \eqref{tracexpcoeff00-0-0}   as an integral over simple functions.

First we show  \eqref{tracexpcoeff00-0-0}.
With the function 
\begin{align} \label{eq:defofCpi2-1-00}
Q_1(u) &:=  \frac{ 1 }{ \nu ( u)   - z }  ,  \quad u \in \R^d , 
\end{align}
we obtain 
\begin{align} 
 \label{eq:defofCpi2-1-1}
{\rm Im} Q_1(u) &=  \frac{ {\rm Im} z   }{( \nu ( u)  - {\rm Re} z )^2 + ({\rm Im} z)^2 }  . 
\end{align}
By definition of the extension $\mathcal{E}_L$ we find 
\begin{align} \label{eq:defofCpi2-1-1-1}
 \int_{\Lambda_L^*}   {\rm Im} Q_1(u_0)  du_0 = 
 \int_{\R^d} \mathcal{E}_L[{\rm Im} Q_1](u_0)   du_0 . 
\end{align}
Since  $Q_1$ is continuous it follows that  $\mathcal{E}_L[{\rm Im} Q_1] \to {\rm Im} Q_1$ pointwise as $L \to \infty$ (in fact one has uniform convergence).  
 An application of Lemma \ref{lem:boundonsquaring}  or 
using elementary properties of polynomials, it is straightforward to see that for $a , b \in \R$   with $b \neq 0$ or $a < 0$  there 
exists a constant $c>0$, such that 
 for all $\xi \in [-\frac{1}{2},\frac{1}{2}]^d$ and $p \in \R^d$ 
\begin{align} \label{basicbound} 
(\nu(p - \xi)  - a)^2 + b^2 \geq  c ( p^4  + 1)  .
\end{align} 
Thus by monotonicity there exist $C$ such that for $L \geq 1$,  $\xi \in [-\frac{1}{2},\frac{1}{2}]^d$, and $u \in \R^d$
\begin{align} \label{EofQ1est} 
| \mathcal{E}_L [{\rm Im} Q_{1}](u) | \leq  \frac{C}{u^4 + 1} .
\end{align} 
Now the r.h.s. of \eqref{EofQ1est} is an integrable majorant.  
Thus it follows by dominated convergence that 
\begin{align} \label{eq:defofCpi2-1-1-1}
\lim_{L \to \infty}    \int_{\R^d}  \mathcal{E}_L[{\rm Im} Q_1](u_0) du_0   & =   \int_{\R^d}  {\rm Im} Q_1(u_0) du_0  , 
\end{align}
which is in view of \eqref{eq:defofCpi2-1-1}  and \eqref{eq:defofCpi2-1-1-1} exactly  \eqref{tracexpcoeff00-0-0}.

 Next  we show \eqref{eq:exoflimit2DOS}.  For this we  express   \eqref{eq:defofCpi2} similarly as an integral over simple functions. 
That is with   
$v = (v_1,....,v_{|I_A|}) \in \R^{d |I_A|}$ 
and 
$$
Q_{n+1,A}(u_0,v)  :=   \prod_{j=1}^{n+1}   \left( \nu \left( u +   \sum_{l=1}^{j-1} [M_A v]_l  \right)  - z  \right)^{-1} , \quad (u_0,v) \in \R^{d |I_A|} \times \R^d  , 
$$
it follows from the definition of the extension $E_L$ that 
\begin{align} \label{eq:defofCpi3}
\widetilde{C}_{n,A,L}[z ]  & =   \int_{\R^d}  du_0 \prod_{l \in I_A} \left( \int_{\R^d}  dv_l \right)   \prod_{j  = 1}^n 
  \mathcal{E}_L [\widehat{B}_\#(-[M_A( \cdot )]_j)](v)   
   \mathcal{E}_L[Q_{n+1,A}](u_0,v) .
\end{align}

First, we observe  that as $L \to \infty$ the integrand of \eqref{eq:defofCpi3} converges  to the integrand of  \eqref{eq:exoflimit2DOS} pointwise.  To see this, we first  note that $Q_{n+1,A}$ is 
continuous. Moreover,  from Lemma \ref{fourierseries} (c) and the continuity of $[M_A(\cdot)]_j$
 we see that $\mathcal{E}_L [\widehat{B}_\#(-[M_A( \cdot )]_j)]$ converges  to $\widehat{B}(-[M_A( \cdot )]_j) $ pointwise.  
Next,  we  want to give a convergent majorant. For this we show that the absolute value of the  integrand  of \eqref{eq:defofCpi3} is uniformly bounded  for  $L \geq 1$ by an integrable function.
Using the trivial bound  \eqref{ln-2} for  all 
resolvents except the first and the last (recall  that by Lemma  \ref{sum0}    $\sum_{l=1}^{n} [M_Av]_l = 0$)
we have
\begin{align}
 \label{eq:estobmaj}  
& \left| \prod_{j  = 1}^n   \mathcal{E}_L [\widehat{B}_\#(-[M_A( \cdot )]_j)](v)     \mathcal{E}_L[Q_{n+1,A}](u_0,v) \right|  \nn    \\
& \leq
\|{B}\|_1^{n-|I_A|}   \prod_{j \in I_A}  \left\{ c_B \mathcal{E}_L[ \langle v_j \rangle^{-2d} ] \right\} ({\rm dist}(z,\R_\geq))^{-n+1} \mathcal{E}_L[|Q_{1}|^2](u_0) \nn \\
& 
\leq  C  ({\rm dist}(z,\R_\geq))^{-n+1} \|{B}\|_1^{n-|I_A|}   \prod_{j \in I_A}  \left\{\langle v_j \rangle^{-2d} \right\}  \langle u_0 \rangle^{-4}   , 
\end{align}
where in the second inequality we used Lemma \ref{fourierseries} (b) to bound the terms involving $B$ for $j \in J_A$,  we used Lemma \ref{lemestfourdisc}
 to bound the terms involving $B$ for $j \in I_A$ with  some constant $c_B$.
Note the  last inequality in  \eqref{eq:estobmaj}   follows since by  \eqref{basicbound}  
there exist a constant such that 
\begin{align}
| \mathcal{E}_L[| Q_{1}|^2](u) | \leq  \frac{C}{u^4 + 1} .
\end{align} 
Furthermore, it follows from  
\eqref{elemboundonsquaring} 
 that 
 $ \langle v_j \rangle^{2d} \mathcal{E}_L[ \langle v_j \rangle^{-2d} ] $  is uniformly bounded. 
We see that  \eqref{eq:estobmaj}  is integrable with respect to  the measure $du_0 \prod_{j \in I_A} d v_j $.
Now applying  the dominated convergence theorem shows \eqref{eq:exoflimit2DOS}.
\end{proof}

 \label{secproofofthem28} \label{sec:9}

 \begin{proof}[Proof of Proposition  \ref{exinflim-bound}] The bound in \eqref{boundonTs}
 follows 
for $n \in \N$ from Lemma \ref{thm:dosasy0-0} and for $n =0$ from 
 \eqref{zerotraceimterm}  and Lemma \ref{lem:discest}.
\end{proof} 

As an immediate consequence of Proposition   \ref{propexofcinfDOS} and the   coordinate transformation \eqref{changevar1} back to  the original
 variables, we obtain the  result of the following corollary, which is interesting of its own but  will not be needed in the sequel.

\begin{corollary}\label{remdeltdirac} For $z \in \C \setminus  [0,\infty)$ and $n \in \N$ we have 
\begin{align} \label{eq:exoflimit3}
& 
\lim_{L\to \infty} {\bf E}_L  \widetilde{\rm tr}  R_L(z) [  V_L  R_L(z)  ]^n \\
& = \sum_{A \in \mathcal{A}_n } \int_{(\R^d)^{n+1}}   \mathcal{P}_{A,\infty}(k) \delta(k_1 -k_{n+1}) \prod_{j=1}^{n+1}  (\nu( k_j) - z )^{-1} d(k_1, \cdots,  k_{n+1})  , \nn
\end{align} 
where 
\begin{align} \label{eq:defofdeltainf}
\mathcal{P}_{A,\infty}(k) & :=  \prod_{a \in A} \left\{ m_{|a|}  \delta \left( \sum_{l \in a} (k_l - k_{l+1}) \right) 
  \prod_{l \in a}  \widehat{B}(k_l - k_{l+1} ) \right\}  .
 \end{align}
  Here $\delta$ denotes the usual  Dirac measure  at the origin in $\R^d$.    
  \end{corollary}

\section{Boundary Values of Expansion Coefficients} 

\label{sec:boundaryvalues}

 In this section we consider the boundary value of the expressions of Proposition  \ref{propexofcinfDOS}  as $z$ approaches the positive real axis. 
In particular,   we prove Theorem  \ref{thm:boundaryvelueexpcoeff}. 
Henceforth,  we assume that the spectral parameter is of the form 
$z = E + i \eta$ with $$
E > 0  \text{ and } \eta > 0  ,
$$
the case $z = E - i \eta$ will follow by taking complex conjugates. 
Moreover we assume that   Hypothesis \ref{H1} holds.

The main idea of the proof is to use a so-called analytic dilation (see for example \cite{CombesThomas.1973,SR4}) to control the limit as $\eta \downarrow 0$.  Using the 
right hand side of \eqref{eq:exoflimit2DOS} and  the change of variables $u_0 \mapsto e^{\theta} u_0$ and  $v_j \mapsto e^{\theta} v_j$ for $\theta \in \R$,  we find that 
\begin{align} \label{eq:defofCpi01}
&\widetilde{C}_{n,A,\infty}[E + i \eta ]  \\
& =   e^{ - d ( |I_A|+1 )\theta}  \int_{(\R^d)^{|I_A|+1}} du_0  \prod_{l \in I_A}  dv_l   \prod_{j = 1}^n  \widehat{B}(- e^{ -  \theta} [M_A(v)]_j )  \nonumber  \\
 & \quad \times  \prod_{j=1}^{n+1}  \left(e^{ - 2  \theta}  \nu \left( u_0 +   \sum_{l=1}^{j-1} [M_A v]_l  \right)  - E -  i \eta \right)^{-1}  . \nonumber
\end{align}
Note that the left hand side of \eqref{eq:defofCpi01} is independent of $\theta \in \R$.  By Hypothesis \ref{H1} and the argument presented below, we will first be   able to analytically extend the right hand side for  $\theta$ into  the  open set  
$$S := \left\{ \theta  \in \C :  - \frac{1}{2} \arctan(\eta/E)  <   {\rm Im} \theta <   \pi / 2  \right\} \cap \left\{ \theta \in \C : |  \theta | < \vartheta_{\rm B}  \right\},$$  
which is 
a  neighborhood of zero. As displayed in the Figure \ref{PictureDilatation} below, for $E,\eta > 0$, the point $E + i \eta$ is located in the first quadrant. At the same time $\nu \left( u_0 +   \sum_{l=1}^{j-1} [M_A (v)]_l  \right)$ is non-negative real-valued (for any $j=1, \ldots , n+1$), hence located on the non-negative real axis and gets turned by the factor $e^{ - 2  \theta}$ in negative direction by the angle $2 {\rm Im} \theta$. Since $\arctan(\eta/E)$  represents the angle of the polar form  of $E+i \eta$ and a rotation by $-\frac{3 \pi}{2}$ marks the transition between second and first quadrant,  
the first set in the intersection in the definition of $S$  ensures the invertibility of 
\begin{align} 
\label{TermPictureDilatation}
\left(e^{ - 2  \theta} \nu  \left( u_0 +   \sum_{l=1}^{j-1} [M_A v]_l  \right)  - E -  i \eta \right)
\end{align}
for any $j=1, \ldots , n+1$. 
Once we have analytically continued $\theta$ into the  complex lower half plane, in a second step,  we can 
analytically continue $z$ from the complex upper half plane to the complex lower half plane. 

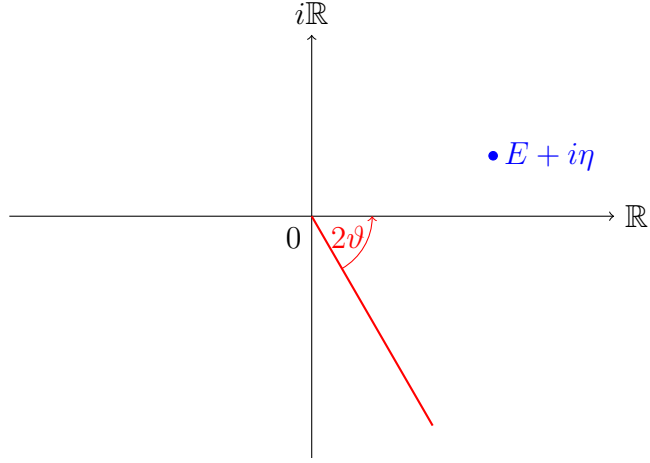
\begin{figure}[H]
\begin{center}

\begin{tikzpicture}[scale=0.8]
  \draw[->] (-5, 0) -- (5, 0) node[right] {$\R$};
  \draw[->] (0,  - 4) -- (0,3) node[above] {$i \R$};
  \draw[ red, thick] (0, 0) -- (2,-2*3^0.5) ;  
\filldraw[blue]  (3,1)  circle[radius=2pt] node[right] {$E+i\eta$};
\draw  (-0.3,0)   node[below] {$0$};
     \draw[red, <-] (1,0) arc (0:-60:1) ;
         \draw[red] (0.6,-0.35)  node {$ 2 \vartheta$};
\end{tikzpicture}

\caption{\label{PictureDilatation}  \small The set $\{  e^{  2 \theta} r  : r \geq 0 \}$  for $\vartheta = {\rm Im} \theta$ is drawn in red. }
\end{center} 
\end{figure}

To estimate the resolvents we will use the following lemma. 

\begin{lemma}\label{lem:triivialomplex}  Let   $\alpha, E , \eta \in \R$ and $ \lambda  \geq 0$. Then    the following holds. 
\begin{enumerate}[(a)] 
\item \label{estofresanayz}  If $\alpha \in [0,\pi]$ and $\eta, E \in [ 0, \infty) $ then 
\begin{align} \label{estonresv2}  
|  e^{ - i \alpha } \lambda - E - i \eta  |  
 &  \geq  |\sin(\alpha) E|. 
 \end{align} 
\item \label{estofresanayz-0}  If $\alpha \in [0,\pi]$ and $\eta, E \in [ 0, \infty) $ then 
\begin{align} \label{estonresv2}  
|  e^{ - i \alpha } \lambda - E - i \eta  |  
 &  \geq \max\{  |\sin(\alpha) E|, \lambda - E  \}. 
 \end{align} 
 \item \label{estofresanayz-3}  If $\alpha \in [0,\pi]$ and $\eta, E \in [ 0, \infty) $ then 
\begin{align} \label{estonresv2}  
|  e^{ - i \alpha } \lambda - E - i \eta  |  
 &  \geq  \max\{ \eta , \lambda - |E-i\eta|\} . 
\end{align} 
 \item \label{estofresanayz-2}  If $\alpha \in (-\arctan(\eta/E),0]$ and $\eta, E \in ( 0, \infty) $ then 
\begin{align} \label{estonresv2}  
|  e^{ - i \alpha } \lambda - E - i \eta  |  
 &  \geq   \max\{ {\rm Im} ( e^{  i \alpha } (E + i \eta)  )  , \lambda - |E-i\eta|\} 
\end{align} 
and  ${\rm Im} ( e^{  i \alpha } (E + i \eta)) > 0$. 
\end{enumerate} 
\end{lemma} 
\begin{proof} 
 \ref{estofresanayz}  Using that $\alpha \in [0,\pi]$  and $\eta \geq 0$ we find 
\begin{align}  \label{estonres2} 
|  e^{ - i \alpha } \lambda - E - i \eta  |  
&  = \sqrt{ ( \cos(\alpha) \lambda - E)^2 + (\sin(\alpha ) \lambda + \eta)^2  }\nonumber  \\
&  \geq  \sqrt{ ( \cos(\alpha) \lambda - E)^2 + (\sin(\alpha ) \lambda )^2  }\nonumber  \\
&   =   \sqrt{  \lambda^2 + E^2 - 2 \lambda E \cos(\alpha)   } 
\end{align} 
Now the infimum is attained if $\lambda = E \cos(\alpha)$ 
so 
\begin{align}  \label{estonres2-2} 
|  e^{ - i \alpha } \lambda - E - i \eta  |  
\geq    E \sin(\alpha) 
\end{align} 
and the monotonicity of the square root. \\
 \ref{estofresanayz-0}  Using  \eqref{estonres2}  and   the triangle inequality we find 
\begin{align*}  
|  e^{ - i \alpha } \lambda - E - i \eta  |  \geq |e^{-i \alpha } \lambda - E | \geq \lambda - E .
\end{align*} 
Now the result follows with   \ref{estofresanayz}.\\
\ref{estofresanayz-3}   If $\alpha \in [0,\pi]$ then $e^{ - i \alpha } \lambda$ is in the complex lower half plane.
 Since $E + i \eta$ is in the first quadrant 
\begin{align*} 
|  e^{ - i \alpha } \lambda - E - i \eta  |   \geq|  {\rm Im}  (e^{ - i \alpha } \lambda - E - i \eta ) | \geq \eta . 
\end{align*} 
On the other hand by the triangle inequality 
\begin{align*} 
|  e^{ - i \alpha } \lambda - E - i \eta  |   \geq  |   \lambda | - |E  + i  \eta| . 
\end{align*} 
Combining the two estimates shows the claim. \\
  \ref{estofresanayz-2} First observe that   $E + i \eta = | E + i \eta | \exp( i\arctan (\eta / E))$ is in the first quadrant, and so  
  for  $\alpha \in (-\arctan(\eta/E),0]$  the number  $e^{  i \alpha } (E + i \eta)$ is still in 
the first quadrant.  In particular ${\rm Im} (e^{  i \alpha } (E + i \eta)) > 0$. Thus it follows from (c) that  
\begin{align*}  
|  e^{ - i \alpha } \lambda - E - i \eta  |   = | \lambda -  e^{  i \alpha } (E + i \eta)  |  
 &  \geq  \max\{ {\rm Im} ( e^{  i \alpha } (E + i \eta)  ) , \lambda - |E-i\eta|\} . 
\end{align*} 
\end{proof}

We are now ready to prove Theorem \ref{thm:boundaryvelueexpcoeff}.

\label{sec:10} 

\begin{proof}[Proof of Theorem \ref{thm:boundaryvelueexpcoeff}]
For $n=0$, the statement follows directly from  \eqref{tracexpcoeff00-0-0}. Hence, in the following let $n \in \N$. 

$$
$$
\underline{STEP 1:} In a first step we fix  $E , \eta > 0$   and show that we 
can analytically   continue  \eqref{eq:defofCpi01} from real $\theta$ to complex values in  $S$. 
For this  we have to show that  the  integrand has an analytic continuation and that the 
analytic continuation of the integrand is indeed integrable.
Analytic continuation of the integrand follows by means of Hypothesis \ref{H1} and the fact that 
quotients of analytic functions are analytic as long as they are well defined. 
Furthermore,  
we can bound the integrand in  \eqref{eq:defofCpi01} for $\theta \in S $  by means of Lemma \ref{lem:triivialomplex}.
 Specifically, 
it follows from Lemma \ref{lem:triivialomplex} Part  (d) that there exists a positive constant  $c_{E,\eta}$ 
depending on $E$ and $\eta$ such that for all $p \in \R^d$ and $\theta \in S$ 
$$
| e^{- 2 \theta} \nu(p) - E - i \eta | \geq  c_{E,\eta} ( e^{- 2 { \rm Re}  \theta } \nu(p) + 1 )  .
$$
Thus we obtain the following bound on the integrand  \eqref{eq:defofCpi01} for $\theta \in S$ 
\begin{align} \nn 
&\left|    \prod_{j = 1}^n  \widehat{B}(- e^{ -  \theta} [M_A(v)]_j )    \prod_{j=1}^{n+1} 
 \left(e^{ - 2  \theta}  \nu \left( u_0 +   \sum_{l=1}^{j-1} [M_A v]_l  \right)  - E -  i \eta \right)^{-1} 
\right| \\   
& \leq  \prod_{s \in I_A}     | \widehat{B}(- e^{ - i \vartheta} v_s) | 
 \| \widehat{B}(- e^{- i \vartheta } \cdot  )\|_\infty^{|A|}  c_{E,\eta}^{-n-1} \ ( e^{- 2 { \rm Re}  \theta } \nu(u_0) + 1 )^{-2}    .  
  \label{eq:defofCpi44456}
\end{align}
We see that the right hand side is integrable with respect to $ \int du_0  \prod_{s \in I_A}  dv_s $.
We conclude that  for fixed $E + i \eta$ we 
can analytically continue in $\theta$ from $\R$ to $S$, by means of  Hypothesis \ref{H1}  and the bound 
  \eqref{eq:defofCpi44456}. 
  Specifically we  choose $\theta = i \vartheta$ with 
   \begin{align} \label{anaastheta} \vartheta \in (0,{\rm min}(\vartheta_{B},\pi/4)) . 
\end{align}
Let 
 $m = |I_A|$ and let a single $\int$  be  a short hand notation  for the collection of all the respective integrals.
From  \eqref{eq:defofCpi01} we thus find 
\begin{align} \label{eq:defofCpi1}
& \widetilde{C}_{n,A,\infty }[E + i \eta]  \\
& =   e^{ -  i d ( m+1 )\vartheta}  \int du_0  \prod_{s \in I_A}  dv_s   \prod_{j = 1}^n  \widehat{B}(- e^{ - i \vartheta} [M_A(v)]_j )     \nonumber  \\
& \quad  \prod_{j=1}^{n+1}  \left(e^{ - 2 i \vartheta} \nu \left( u_0 +   \sum_{l=1}^{j-1} [M_A (v)]_l  \right)  - E -  i \eta \right)^{-1}  . \nonumber
\end{align}
$$
$$
\underline{STEP 2:}  Now for fixed $\vartheta$ as in   \eqref{anaastheta} we want to show 
that we  can  continue \eqref{eq:defofCpi1} for $z = E + i \eta$ from the 
upper half complex plane to the positive real axis. To show this we want to apply dominated 
convergence. For this we estimate the integrand of \eqref{eq:defofCpi1}
to find   an integrable majorant.
First we note that from Lemma    \ref{lem:triivialomplex} Part   \ref{estofresanayz-0} 
we see that for fixed  $\vartheta$ as in   \eqref{anaastheta} and $E$ in a compact subset $I$ of $(0,\infty)$  
 there exists a positive constant $c_{\vartheta,I}$ such that for all $p \in \R^d$, $E \in I$, and  $\eta \geq 0$ 
\begin{align}
| e^{ - 2 i \vartheta} \nu \left( p  \right)  - E -  i \eta | \geq c_{\vartheta,I} ( \nu(p) + 1 ) .
\end{align} 
 This gives 
\begin{align} \nn 
&  \prod_{j = 1}^n \left|  \widehat{B}(- e^{ - i \vartheta} [M_A(v)]_j )  \right|   \prod_{j=1}^{n+1} \left|  \left(e^{ - 2 i \vartheta} \nu \left( u_0 +   \sum_{l=1}^{j-1} [M_A (v)]_l  \right)  - E -  i \eta \right)^{-1} \right|  \\
& \leq \prod_{s \in I_A}     | \widehat{B}(- e^{ - i \vartheta} v_s) |  \| \widehat{B}(- e^{- i \vartheta } \cdot  )\|_\infty^{|A|} c_{\vartheta,I}^{-n-1}
 \left| \nu  \left( u_0  \right)  + 1   \right|^{-2}     
  \label{eq:defofCpi44456-2}
\end{align}

We see that the right hand side is integrable with respect to $ \int du_0  \prod_{s \in I_A}  dv_s $.
Thus by dominated convergence we can take the limit $\eta \downarrow 0$ in  \eqref{eq:defofCpi1}.
Analogously,  it follows that $\widetilde{\tr}C_{n,A,\infty }[z ]$ has  a continuous 
extension to $\{ w \in \C : {\rm Im } w \geq  0, \ {\rm Re } w > 0 \}$ for $z \in \{ w \in \C : {\rm Im } w > 0, \ {\rm Re } w > 0 \}$,  by choosing the same majorant and using the continuity of the integrand in $z = E + i \eta$. 
This completes the proof Theorem  \ref{thm:boundaryvelueexpcoeff}.
\end{proof}

\section{Asymptotic Error Estimate} 

\label{errorestimate} 

The goal of this  section is to  prove the asymptotic error estimate given in Theorem \ref{thm:maintec000}, which will take some 
steps as preparatory work.
We start with the following lemma.

\begin{proposition}  \label{uniformboundontrace} 
There exists a constant $C$  such that for all $E \in \R$, $\eta \neq 0$, and $\lambda \geq 0$ 
\begin{align}
  {\bf E }_L \widetilde{\tr}  |
{\rm Im} ( H_L  - E - i \eta )^{-1}   | \leq |\eta|^{-1}  (1 + 2 E^2 + 2 \eta^2)  C (1+\lambda^2) 
\end{align} 
for all $L \geq 1$. 
\end{proposition} 
\begin{proof} 
This follows from Proposition  \ref{lem:proofofexdens00-0} by  observing  that  for all $x \in \R$ 
and $\eta \neq  0$  we have 
\begin{align} \label{tracialboundprep} 
| \eta|  \left| {\rm Im} \frac{1}{ x  - E - i \eta} \right| =  \frac{\eta^2}{(x-E)^2+\eta^2} \leq  ( 1 + 2 E^2 + 2 \eta^2) \langle x \rangle^{-2} .
\end{align} 
\end{proof}

\begin{lemma} \label{lemmerorbound}  Let $z \in \C \setminus \R$ and  $\lambda \geq 0$.
Then 
\begin{align} \label{eq:expansionU-0} 
 {\bf E }_L \widetilde{\tr}  
{\rm Im} ( H_L  - z )^{-1}   & = \sum_{j=0}^{n-1}  \lambda^j \widetilde{\tr}  {\rm Im}  T_{j,L} [z ]  +  \lambda^{n} U_{n,L}[z ;  \lambda]   , 
\end{align}
with 
\begin{align}\label{eq:expansionU2-0} 
 U_{n,L}[z  ; \lambda] &  :=  {\bf E }_L  \widetilde{\tr}  {\rm Im} ( [R_L(z)   V_L  ]^{n}  ( H_L - z  )^{-1}  )  .
\end{align}
Furthermore,
\begin{align}
|   U_{n,L}[z ; \lambda]   | 
& \leq   ( \mathfrak{E}_{n,L}[z ])^{1/2} \left( |{\rm Im} z |^{-1}| {\bf E}_L \widetilde{\tr} {\rm Im} (( H_L  - z)^{-1}) |\right)^{1/2}  \label{eq:boundonU-0}  , 
\end{align}
for 
\begin{align} \label{defofENL-0} 
\mathfrak{E}_{n,L}[z]   & :=   {\bf E}_L \widetilde{\tr}(( [R_L(z)  V_L  ]^{n})^*
 [R_L(z)   V_L  ]^{n}) .   
\end{align}
\end{lemma} 
\begin{proof} W.l.o.g. we assume ${\rm Im} z  > 0$,  since the other case follows from  complex conjugation. 
Iterating the resolvent identity, cf.  Lemma \ref{resolventenformel_interiert}, we find 
\begin{align*} 
 ( H_L  - z)^{-1} & = \sum_{j=0}^{n-1} R_L(z) [ \lambda V_L R_L(z)  ]^j +  [R_L(z)  \lambda V_L  ]^{n}  ( H_L  - z)^{-1} . 
\end{align*}
Thus we find for the normalized trace and expectation that 
\begin{align*} 
&  {\bf E} {\rm \widetilde{tr}}{\rm Im} ( H_L  - z)^{-1} \\& =  {\bf E} {\rm \widetilde{tr}}{\rm Im} 
\sum_{j=0}^{n-1} R_L(z) [ \lambda V_L R_L(z)  ]^j + {\bf E} {\rm \widetilde{tr}}{\rm Im}
 [R_L(z)  \lambda V_L  ]^{n}  ( H_L  - z)^{-1} . 
\end{align*}
This shows 
 \eqref{eq:expansionU-0}  with \eqref{eq:expansionU2-0}.

To estimate the error term \eqref{eq:expansionU2-0} we use  ($|{\rm tr} {\rm Im } A| = | {\rm tr}  \frac{1}{2i} ( A - A^* ) | \leq \frac{1}{2} {\rm tr}  (|A| + |A| )  = {\rm tr} |A|$) and H\"older inequality (or Cauchy-Schwarz) for traces 
\begin{align*} 
  & | {\bf E} {\rm \widetilde{tr}}{\rm Im}
 [R_L(z)  \lambda V_L  ]^{n}  ( H_L  - z)^{-1}  |  \leq {\bf E} {\rm \widetilde{tr}} |
 [R_L(z)  \lambda V_L  ]^{n}  ( H_L  - z)^{-1} | \\
& \leq  \left( {\bf E} {\rm \widetilde{tr}}( [R_L(z)  \lambda V_L  ]^{n})^*
 [R_L(z)  \lambda V_L  ]^{n} \right)^{1/2}  
\left( {\bf E} {\rm \widetilde{tr}} (( H_L  - z)^{-1})^*  ( H_L  - z)^{-1} \right)^{1/2} 
\end{align*}
To estimate the second term we use   the following  identity
which follows from the resolvent identity 
\begin{align}
   {\rm Im}  ( H_L  - z)^{-1} & = \frac{1}{2i} ( ( H_L  - z)^{-1}  -  [( H_L  - z)^{-1} ]^*) \label{impartofiop} \\
&  = 
 \frac{1}{2i}  (( H_L  - z)^{-1})^* (z - \overline{z})  ( H_L  - z)^{-1} =    {\rm Im } z   (( H_L  - z)^{-1})^* ( H_L  - z)^{-1}  \nn
\end{align} 
Thus 
\begin{align}\label{identresforHL}
{\bf E}  \widetilde{\rm tr} (( H_L  - z)^{-1})^*  ( H_L  - z)^{-1}  = \frac{1}{ {\rm Im} z} {\bf E}  \widetilde{\rm tr} {\rm Im}  ( H_L  - z)^{-1} 
\end{align} 
Now the bound on the r.h.s follows from  Proposition \ref{uniformboundontrace} .
This shows  the bound \eqref{eq:boundonU-0}  with  \eqref{defofENL-0}. 
\end{proof}

Next we estimate \eqref{defofENL-0}. 
Inserting the onb  \eqref{eq:deofonb} we find with  \eqref{ftident2} (observing that each term
involving an interaction $V$ gives a contribution $|\Lambda_L|^{-1}$ 

\begin{lemma}\label{fourierspacekernelerror} Let $z \in \C \setminus  [0,\infty)$ and $n \in \N$. Then for almost every $\omega$  
the operator 
$ 
 [R_L(z)   V_L  ]^{n} ( [R_L(z)   V_L  ]^{n})^*
$ 
is  trace class. 
Furthermore,  for almost every $\omega$ 
\begin{align}  \nonumber 
& \widetilde{\tr}(
 [R_L(z)   V_L  ]^{n}  [R_L(z)   V_L  ]^{n})^*) 
\\ 
&  =  \int_{(\Lambda_L^*)^{n+1} }  \int_{(\Lambda_L^*)^{n+1} } 
 \prod_{j=1}^n  \widehat{V}_{L}(p_j - p_{j+1} ) \delta_{*,L}(p_{n+1} - \tilde{p}_1)  \prod_{j=1}^n  \widehat{V}_{L}(\tilde{p}_j - \tilde{p}_{j+1} )  
\nonumber  \\
& \times  \left(    \prod_{j=1}^{n}      \frac{1}{ \nu(p_j) - z}        \right)  \left(    \prod_{j=2}^{n+1}      \frac{1}{ \nu(\tilde{p}_j) - \overline{z}}        \right) \delta_{\tilde{p}_{n+1},  {p}_1}  d(p_1, \ldots , p_{n+1})  d(\tilde{p}_1, \ldots ,  \tilde{p}_{n+1}) 
\label{defofENL-0-0-0}  \\
&  = 
 \int_{(\Lambda_L^*)^{2n+1} }    \prod_{j=1}^{2n}  \widehat{V}_{L}(q_j - q_{j+1} )
\nonumber  \\
& \times  \left(    \prod_{j=1}^{n}      \frac{1}{ \nu(q_j) - z}        \right)  \left(    \prod_{j=n+2}^{2n+1}      \frac{1}{ \nu(q_j) - \overline{z}}        \right)  \delta_{q_{2n+1},{q}_1}  d (q_1 , \ldots , q_{2n+1}). \label{eq:expforTn} .
\end{align}
The sum converges absolutely and satisfies  the bound 
\begin{align} \label{defofENL-0-0} 
&  \int_{(\Lambda_L^*)^{n+1} }  \int_{(\Lambda_L^*)^{n+1} } 
 \prod_{j=1}^n  | \widehat{V}_{L}(p_j - p_{j+1} ) | \delta_{*,L}(p_{n+1} - \tilde{p}_1)  \prod_{j=1}^n   | \widehat{V}_{L}(\tilde{p}_j - \tilde{p}_{j+1} )   | 
\nonumber  \\
& \times  \left(    \prod_{j=1}^{n}      \frac{1}{| \nu(p_j) - z|}        \right)  \left(    \prod_{j=2}^{n+1}      \frac{1}{ | \nu(\tilde{p}_j) - \overline{z}|}        \right) \delta_{\tilde{p}_{n+1},  {p}_1}  d(p_1, \ldots , p_{n+1})  d(\tilde{p}_1, \ldots ,  \tilde{p}_{n+1})  
 \nonumber \\
&  \leq  {\rm dist}(z,\R_{\geq })^{-2n+2} \| \widehat{V}_L \|_{*,1}^{2n} \int_{\Lambda_L^*} |\nu(p) - z|^{-2} dp . 
\end{align}
\end{lemma} 
We note that the proof of  Lemma \ref{fourierspacekernelerror} is almost identical to 
that of Lemma  \ref{fourierspacekernel}.

\begin{proof}
Inserting iteratively the o.n.b.  \eqref{eq:deofonb}, we obtain as an iterated sum 
\begin{align*}
&  {\rm \tilde{tr}} (
 [R_L(z)   V_L  ]^{n}  ([R_L(z)  V_L  ]^{n})^* )  \\
& = |\Lambda_L|^{-1}  \sum_{\tilde{p}_{n+1}} \cdots  \sum_{p_1}\langle \varphi_{p_1} R_L(z) \varphi_{p_1} \rangle  \langle \varphi_{p_1} V_L  \varphi_{p_2}  \rangle \langle  \varphi_{p_2}  R_L(z)  \varphi_{p_2} \rangle 
\cdots \langle  \varphi_{p_{n}} , V_L \varphi_{p_{n+1}} \rangle \\
& \langle \varphi_{p_{n+1}} , \varphi_{\tilde{p}_1} \rangle   \langle  \varphi_{\tilde{p}_{1}} , V_L \varphi_{\tilde{p}_{2}} \rangle \langle  \varphi_{\tilde{p}_2}  R_L(z)  \varphi_{\tilde{p}_2} \rangle \cdots   
 \langle  \varphi_{\tilde{p}_{n}} , V_L \varphi_{\tilde{p}_{n+1}} \rangle \langle \varphi_{\tilde{p}_{n+1}}  R_L(z)  \varphi_{\tilde{p}_{n+1}} \rangle  \langle \varphi_{\tilde{p}_{n+1}} , \varphi_{{p}_1} \rangle \\
& = |\Lambda_L|^{-1}|\Lambda_L|^{-2n} \sum_{\tilde{p}_{n+1}} \cdots  \sum_{p_1}
   (\nu(p_1) - z)^{-1}   \widehat{V}_L(p_1-p_2)   (\nu(p_2) - z)^{-1}   
\cdots \widehat{V}_L(p_{n} - p_{n+1}) \\
&   \delta_{{p}_{n+1} ,\tilde{p}_1}   \widehat{V}_L(\tilde{p}_{1} - \tilde{p}_{2}) 
 (\nu(\tilde{p}_2) - z)^{-1}  \cdots     \widehat{V}_L(\tilde{p}_{n} - \tilde{p}_{n+1}) 
  (\nu(\tilde{p}_{n+1}) - z)^{-1}  \delta_{\tilde{p}_{n+1} ,{p}_1} .
\end{align*}
Thus the power counting in $|\Lambda_L|$ shows that we can write one of the Kronecker delta functions  in terms of  $\delta_{*,L}$
Thus we find as an iterated integral 
\begin{align}
 & \widetilde{\rm tr}  ( [R_L(z)   V_L  ]^{n})^*
 [R_L(z)   V_L  ]^{n} \\ 
&   \nonumber \\ 
&  = \prod_{l=1}^{n+1}  \int_{\Lambda_L^*  }  \prod_{\tilde{l}=1}^{n+1}  \int_{\Lambda_L^*} 
  \prod_{j=1}^n  \widehat{V}_{L}(p_j - p_{j+1} ) \delta_{*,L}(p_{n+1} - \tilde{p}_1)  \prod_{j=1}^n  \widehat{V}_{L}(\tilde{p}_j - \tilde{p}_{j+1} )  
\nonumber  \\
& \times  \left(    \prod_{j=1}^{n}      \frac{1}{ \nu(p_j) - z}        \right)  \left(    \prod_{j=2}^{n+1}      \frac{1}{ \nu(\tilde{p}_j) - \overline{z}}        \right) \delta_{\tilde{p}_{n+1},  {p}_1}  d p_1 \ldots  d p_{n+1}   d \tilde{p}_1 \ldots  d  \tilde{p}_{n+1} .  \nn
\end{align} 
To see that the sum converges absolutely we  first estimate  using   $|(\nu(\tilde{p}) - z)^{-1}| \leq   {\rm dist}(z,\R_{\geq })^{-1}$   and  then successively integrate out
\begin{align}
&  \prod_{l=1}^{n+1}  \int_{\Lambda_L^*  }  \prod_{\tilde{l}=1}^{n+1}  \int_{\Lambda_L^*} 
  \prod_{j=1}^n  |\widehat{V}_{L}(p_j - p_{j+1} )| \delta_{*,L}(p_{n+1} - \tilde{p}_1)  \prod_{j=1}^n  |\widehat{V}_{L}(\tilde{p}_j - \tilde{p}_{j+1} )  |
\nonumber  \\
& \times  \left(    \prod_{j=1}^{n}      \frac{1}{ |\nu(p_j) - z|}        \right)  \left(    \prod_{j=2}^{n+1}      \frac{1}{| \nu(\tilde{p}_j) - \overline{z}|}        \right) \delta_{\tilde{p}_{n+1},  {p}_1}  d p_1 \ldots  d p_{n+1}   d \tilde{p}_1 \ldots  d  \tilde{p}_{n+1}  \nn \\
& \leq   \prod_{l=1}^{n+1}  \int_{\Lambda_L^*  }  \prod_{\tilde{l}=1}^{n+1}  \int_{\Lambda_L^*} 
  \prod_{j=1}^n  |\widehat{V}_{L}(p_j - p_{j+1} )|  |\Lambda_L | \prod_{j=1}^n  |\widehat{V}_{L}(\tilde{p}_j - \tilde{p}_{j+1} )  |
\nonumber  \\
& \times {\rm dist}(z,\R_{\geq })^{-2n+2}    \frac{1}{ |\nu(p_1) - z|}            \frac{1}{| \nu(\tilde{p}_{n+1}) - \overline{z}|}         \delta_{\tilde{p}_{n+1},  {p}_1}  d p_1 \ldots  d p_{n+1}   d \tilde{p}_1 \ldots  d  \tilde{p}_{n+1} \nn \\
& \leq     \int_{\Lambda_L^*  }   \int_{\Lambda_L^*} 
 \|\widehat{V}_{L}\|_{1,*}^{2n}  |\Lambda_L | 
 {\rm dist}(z,\R_{\geq })^{-2n+2}    \frac{1}{ |\nu(p_1) - z|}            \frac{1}{| \nu(\tilde{p}_{n+1}) - \overline{z}|}         \delta_{\tilde{p}_{n+1},  {p}_1}  d p_1  \tilde{p}_{n+1} \nn  \\
& \leq    
 \|\widehat{V}_{L}\|_{1,*}^{2n}  
{\rm dist}(z,\R_{\geq })^{-2n+2}    \int_{\Lambda_L^*  }    \frac{1}{ |\nu(p) - z|^2}         dp . \nn
\end{align} 
Finally \eqref{eq:expforTn} follows by summing  out $p_{n+1}$ and  $\tilde{p}_{n+1}$ 
 using the delta function and the following simple relabeling of the integration variables 
$$
q=((q_1,...q_n),(q_{n+1},...,q_{2n+1})) = ((p_1,...,p_{n}),(\tilde{p}_1,...,\tilde{p}_{n+1})) =(p,\tilde{p}) .
$$
\end{proof}

To obtain an explicit expression for \eqref{defofENL-0}  we  calculate  the expectation of  \eqref{defofENL-0-0}.
For this we define  for $A  \in \mathcal{A}_{2n} $ and $z \in \C \setminus [0,\infty)$ 
\begin{align} \label{eq:defofCpitilde}
{D}_{n,A,L}[z] 
& :=      \int_{(\Lambda_L^*)^{2n+1}}  \mathcal{P}_{A,L}( (q_1 , \ldots , q_{2n+1})) 
\nn
\\
& \times\prod_{j=1}^{n}  (\nu( q_j) - z)^{-1}\prod_{j=n+2}^{2n+1}  (\nu( q_j) - \overline{z})^{-1} \delta_{q_{2n+1},{q}_1} d (q_1 , \ldots , q_{2n+1}) 
\end{align}
We  recall that by   the definition given in  \eqref{eq:defofdeltapi0} we have  for $A \in \mathcal{A}_{2n}$
\begin{align} \label{eq:defofdeltapi}
\mathcal{P}_{A,L}(q) & =  \prod_{a \in A} \left[ m_{|a|}   \delta_{*,L} \left( \sum_{l \in a} (q_l - q_{l+1}) \right) 
  \prod_{l \in a}  \widehat{B}_\#(q_l - q_{l+1} ) \right]  .
 \end{align}

\begin{lemma} 
Let  $z \in \C \setminus [0,\infty)$. Then the sum in  \eqref{eq:defofCpitilde}
converges absolutely.  Furthermore, we can write   \eqref{defofENL-0} in terms of \eqref{eq:defofCpitilde}
as 
\begin{align} \label{CtildeErel} 
\mathfrak{E}_{n,L}[ z  ] =  {\bf E}_L  {\rm \widetilde{tr}} (
 [R_L(z)   V_L  ]^{n}  ([R_L(z)   V_L  ]^{n})^* )   = \sum_{A \in \mathcal{A}_{2n}} {D}_{n, A,L}[ z ]  .  
\end{align} 

\end{lemma} 
\begin{proof}
Using the  definition given in  \eqref{defofENL-0}  and \eqref{defofENL-0-0} we find 
\begin{align} \label{defofENL-0-0} 
& \mathfrak{E}_{n,L}[ z  ] = {\bf E}   \widetilde{\tr}( [R_L(z)  V_L  ]^{n})^*
 [R_L(z)   V_L  ]^{n} ) \\
&  = 
 \int_{(\Lambda_L^*)^{2n+1} }  {\bf E}   \prod_{j=1}^{2n}  \widehat{V}_{L}(q_j - q_{j+1} )
\nonumber  \\
& \times  \left(    \prod_{j=1}^{n}      \frac{1}{ \nu(q_j) - z}        \right)  \left(    \prod_{j=n+2}^{2n+1}      \frac{1}{ \nu(q_j) - \overline{z}}        \right)  \delta_{q_{2n+1},{q}_1}  d (q_1 , \ldots , q_{2n+1}). \label{eq:expforTn} ,
\end{align}
where we  interchange the integrals using  Fubini, which is justified by the absolute convergence  
of the expression in \eqref{defofENL-0-0}  and the bound  obtained by Lemma  \ref{boundonexppotenyt}.
Now the claim follows using  Lemma  \ref{expProd00123}.
\end{proof}

To obtain a good   estimate  for \eqref{eq:defofCpitilde} 
we  introduce  analogously as in Section \ref{sec:transformU} the following change of variables  similar to   \eqref{changevar1}
\begin{align} \label{changevar1v2}
u_0  = q_1  \text{ and }  u_{l} = q_{l+1} - q_l \text{ for }  l=1,...,n 
\end{align}
and $u=(u_0 , \ldots ,  u_{2n})$.
 With this we can express the $q_j$'s as
$$
q_j = \sum_{l=0}^{j-1} u_l .
$$
Using the change of variables  \eqref{changevar1v2} in  \eqref{eq:defofCpitilde}  we find with \eqref{eq:defofdeltapi} expressed in the new variables 
that 
  for $A  \in \mathcal{A}_{2n} $ 
\begin{align}  \label{eq:defofCpiVarChange0}
{D}_{n,A,L}[z ] & =      \int_{(\Lambda_L^*)^{2n+1}}  \prod_{a \in A} \left[ m_{|a|}  \delta_{*,L} \left( \sum_{l \in a} u_l\right) 
  \prod_{l \in a}  \hat{B}_\#(- u_l) \right] 
\\
& \times    \prod_{j=1}^{n}  \left(\nu\left(  \sum_{l=0}^{j-1} u_l \right) - z \right)^{-1}\prod_{j=n+2}^{2n+1}  \left(\nu \left(  \sum_{l=0}^{j-1} u_l \right) - \overline{z} \right)^{-1} \delta_{0,\sum_{l=1}^{2n} u_l}  d(u_0 , \ldots , u_{2_n}).     \nn
\end{align}

We will use  the  notation introduced in Section  \ref{sec:transformU}, in particular \eqref{changevar1next},
to sum over the  delta  functions  in \eqref{eq:defofCpiVarChange0}.
 Note that for $A  \in \mathcal{A}_{2n} $  (note that in view of Lemma \ref{sum0}    we have by the delta functions 
that $\sum_{l=1}^{2n} u_l=0$  and so 
the  $\delta_{0,\sum_{l=1}^{2n} u_l} $ term  in   \eqref{eq:defofCpiVarChange0}  can be dropped). Thus we arrive at the following lemma. 

\begin{lemma} Let  $z \in \C \setminus  [0,\infty)$. Then 
\begin{align} \label{CtildeErel} 
\mathfrak{E}_{n,L}[ z   ]  & = \sum_{A \in \mathcal{A}_{2n}} {D}_{n, A,L}[ z ]  
\end{align} 
with 
\begin{align} 
& {D}_{n,A,L}[z ]\nn \\
&=      \int_{\Lambda_L^*} d u_0  \prod_{l \in I_A} \left( \int_{\Lambda_L^*} d u_l  \right)   \prod_{a \in A} \left[ m_{|a|}   \prod_{l \in a }  \hat{B}_\#(- [M_A (u)]_l) \right]  \nn
\\
& \times \prod_{j=1}^{n}  \left(\nu\left( u_0 + \sum_{l=1}^{j-1} [M_A (u)]_l  \right) - z  \right)^{-1}\prod_{j=n+2}^{2n+1}  \left (\nu \left( u_0 + \sum_{l=1}^{j-1} [M_A (u)]_l  \right) - \overline{z}  \right)^{-1}, \label{eq:defofCpiVarChange}
\end{align}
where the sum converges absolutely. We recall that  $M_A(\cdot )$ is defined in \eqref{changevar1next}.
\end{lemma}

To estimate this expression we proceed similarly as in \cite{HaslerKoberstein.2025-2}.
We keep the first   resolvents as well as the last resolvent to obtain the necessary decay in $u_0$ to show  the tracial property. 
In view of Lemma  \ref{lemmalinindep}  we know that the expression  $u_0 + \sum_{l=1}^{j-1} [M_A (u)]_l$   in the $j$-th resolvents for $j=1,...,2n+1$ with  $j-1 \in I_A$ are as functions of the coordinates  linearly independent.
The remaining  resolvents except the first and last one, i.e., $j=2,...,2n$ for  $j - 1 \in J_A$ we estimate  by $\eta^{-1}$ using   \eqref{ln-1}, assuming $0 < \eta \leq 1$.  Recall that this may include   the case where  $j = n+1$ and $j-1 \in J_A$, that is 
where we have an identity instead of a  resolvent, we   estimate trivially by $\eta^{-1}$. 
Furthermore,   for $l \in J_A$ we  estimate  $|\hat{B}_\#(-[M_A(u)]_l)| \leq  \| \hat{B}_\#\|_{*,\infty} $ for $u \in \Lambda_L^*$. 

Thus as  outlined in the previous paragraph,  to estimate \eqref{eq:defofCpiVarChange} for $z \in  \C \setminus \R$
 we organize the factors occurring in its integrand  as 
\begin{align} \label{eq:defofCpiVarChange-2-3}
& D_{n,A,L}[z ] \\
&=      \int_{\Lambda_L^*} d u_0  \prod_{l \in I_A} \left( \int_{\Lambda_L^*} d u_l  \right)   \prod_{a \in A} \left[ m_{|a|}   \prod_{l \in a \cap I_A }  \hat{B}_\#(- [M_A (u)]_l)  \prod_{l \in a \cap J_A }  \hat{B}_\#(- [M_A (u)]_l) \right]  \nn
\\
& \times \prod_{j=1 : j-1 
\in I_A \cup \{ 0 \}}^{n}  \left(\nu\left( u_0 + \sum_{l=1}^{j-1} [M_A (u)]_l  \right) - z  \right)^{-1}\prod_{j=n+2: j-1 
\in I_A \cup \{2n\} }^{2n+1}  \left (\nu \left( u_0 + \sum_{l=1}^{j-1} [M_A (u)]_l  \right) - \overline{z}  \right)^{-1}, \nn \\
& \times \prod_{j=1: j-1 
\in  J_A}^{n}  \left(\nu\left( u_0 + \sum_{l=1}^{j-1} [M_A (u)]_l  \right) - z  \right)^{-1}\prod_{j=n+2:  j-1 
\in J_A\setminus \{2n\}}^{2n+1}  \left (\nu \left( u_0 + \sum_{l=1}^{j-1} [M_A (u)]_l  \right) - \overline{z}  \right)^{-1} , \nn
\end{align}
and  estimate the factors partitioned over $J_A$ as follows: 
\begin{itemize}
\item
   for $l \in J_A$ we  estimate  $|\hat{B}_\#(-[M_A(u)]_l)| \leq  \| \hat{B}_\#\|_{*,\infty} $ for $u \in \Lambda_L^*$,
  \item  for $j-1 \in J_A \setminus \{ 2n\}$ we estimate $|(\nu( \cdots ) - z )^{-1}| \leq |{\rm Im } z |^{-1}$. 
\end{itemize} 
This gives  for $E , \eta > 0$ 
\begin{align}  \nn 
& | D_{n,A,L}[E \pm  i \eta ] |  \\
& \leq   \eta^{-|A|+1}  \| \hat{B}_\# \|_{*,\infty}^{|A|} \left|  \prod_{a \in A} m_{|a|} \right| \int_{\Lambda_L^*} d u_0 \prod_{l \in I_A}  \left( \int_{\Lambda_L^*} d u_l  \right)   \bigg|    \prod_{l=1, l \in I_A }^{2n}  \hat{B}_\#( -u_l)  \bigg|
 \nonumber 
  \\
& \times \prod_{\substack{ j=1 \\ j - 1 \in I_A  }}^{n}  \left|\nu\left( u_0   + \sum_{l=1}^{j-1}[M_A (u)]_l  \right) - E \mp  i \eta \right|^{-1}\prod_{\substack{ j=n+2 \\  j - 1 \in I_A }}^{2n+1}  \left| \nu \left( u_0 +  \sum_{l=1}^{j-1} [M_A (u)]_l   \right) - E \pm  i \eta  \right|^{-1} , \nonumber 
  \\
& \times   \left|\nu\left( u_0    \right) - E \mp  i \eta \right|^{-2}   . \label{Dtoestim1-NEW} 
\end{align}

To further estimate \eqref{Dtoestim1-NEW}  we will use the following two lemmas, similarly
as in \cite{HaslerKoberstein.2025-2}.   The first lemma is about an algebraic 
relation, which is satisfied  by the  transformed variables.

\begin{lemma} \label{lemmalinindep} 
Let $A$ be a partition  of $\{ 1, \ldots ,  2n \}$.  Let $J_A,I_A$ and $M_A$ be defined as in \eqref{defofJA},   \eqref{defofIA} and  \eqref{changevar1next}. 
Let for $A$ and $j =1,...,2n+1$ the function  $\sigma_{A,j}$ be defined as
\begin{align*}
\sigma_{A,j} : \{ 1 ,  \ldots ,  2n \} \to \{ 0,1 \} \quad \text{via} \quad
\sigma_{A,j}  (l) =
 \begin{cases}
1 \quad : \quad  \max {a(l)} > j,
\\
0 \quad : \quad \text{otherwise}.
\end{cases}
\end{align*}

Then for all $j=2,...,2n$ with $j \in I_A$ we have
\begin{align} \label{linindep} 
 \sum_{l=1}^{j}[M_A (u)]_l   = u_{j} + \sum_{l \in \{1,...,j-1\} \cap    I_A } \sigma_{A , j-1}  (l) u_l .
\end{align}
\end{lemma}

 The proof of the lemma follows from an algebraic argument, for details see   \cite{HaslerKoberstein.2025-2}.
Furthermore,  we will use  the bound in the following Lemma to estimate 
\eqref{Dtoestim1-NEW}.This  is 
the  final step before we can prove Theorem \ref{thm:maintec000}.
To show the lemma, we will use that the  relation shown  in Lemma \ref{lemmalinindep}
allows us to   perform an induction to integrate over the $u_l$ variables one after another.
\label{sec:ausintegrieren}

\begin{lemma}
\label{ausintegrieren}
Let $n \in \N$,  $A$ be a partition of the set $\{ 1, \ldots ,  2n \}$,  and  $J_A,I_A$ and $M_A$ be defined as in \eqref{defofJA},   \eqref{defofIA} and \eqref{changevar1next},
respectively.  Let  $z = E \pm i \eta$
with $E > 0$ and $\eta  > 0$.  Define $\sigma_{A, j }$ as in Lemma   \ref{lemmalinindep}  and  ${C}(E,d,L,\eta,\hat{B}_\#)$ as in  Proposition \ref{ln000}. Then for all \(u_0\in\Lambda_L^*\)  
\begin{align}   \label{LHS}
&   \prod_{l \in I_A}  \left( \int_{\Lambda_L^*} d u_l  \right)   \bigg|    \prod_{l=1, l \in I_A }^{2n}  \hat{B}_\#( -u_l)  \bigg|
 \\
 &\times
 \prod_{\substack{ j=2 \\ j - 1 \in I_A }}^{n}  \left|\nu\left( u_{j-1}  + u_0 + \sum_{l \in \{1,...,j-2\} \cap    I_A } \sigma_{A , j-2}  (l) u_l \right) - z  \right|^{-1}  \nonumber   \\
  &\times
 \prod_{\substack{ j=n+2 \\ j - 1 \in I_A }}^{2n+1}  \left|\nu\left( u_{j-1}  + u_0 + \sum_{l \in \{1,...,j-2\} \cap    I_A } \sigma_{A , j-2}  (l) u_l \right) - \overline{z} \right|^{-1}
 \nonumber \\
 &\leq   [  {C}(E,d,L, \eta,\hat{B}_\#)   ]^{2n-|A|}  .\nonumber
\end{align}

\end{lemma} 

This lemma can be shown via an induction using Lemma \ref{lemmalinindep} as well as Proposition \ref{ln000}.
The proof is  identical to the proof of   \cite[Lemma 6.2]{HaslerKoberstein.2025-2} up to the very last step, which can be omitted.
We are now able to prove Theorem \ref{thm:maintec000} as outlined below.

\label{sec:proooferror} 

\begin{proof}[Proof of Theorem \ref{thm:maintec000}]   Let $E, \eta > 0$. We assume that $z = E + i \eta$ (the case $z = E - i \eta$ is obtained by taking complex conjugates). 
We start from \eqref{Dtoestim1-NEW}. We insert    \eqref{LHS} of Lemma  \ref{ausintegrieren} into  \eqref{Dtoestim1-NEW}.
The remaining sum over $u_0$ is then estimated  by   Lemma  \ref{lem:discest} (using $d\leq 3$) giving  a constant 
$C_{E,0}$, uniformly bounded for $E$ in compact subsets of $(0,\infty)$,   such that 
\begin{align}  \label{Dtoestim1tracial} 
 C_{T,\eta}:= \sup_{L \geq 1} \int_{\Lambda_L^*} \frac{1}{|\nu\left( u_0    \right)^2 - E \pm  i \eta|^2} da  \leq C_{E,0} ( 1 + \eta^{-2} ) .
\end{align} 
This gives 
\begin{align} \label{0finalestonD}
|   D_{n,A,L}[E + i \eta ]|  & \leq    \eta^{-|A|-1}  C_{T,\eta} \| \hat{B}_\# \|_{*,\infty}^{|A|}    C(E,d, L, \eta,\hat{B}_\#)^{2n-|A|}   \left|  \prod_{a \in A} m_{|a|} \right|,
\end{align}
where $C(E,d, L, \eta,\hat{B}_\#) $ is by  Lemma \ref{ausintegrieren} the constant  \eqref{C007} given in Proposition \ref{ln000}
\begin{align} \label{CdEB0}
  &C(E,d, L, \eta,\hat{B}_\#) \\
  &=
2 \| \hat{B}_\#  \|_{*,\infty}  \left[  C_{1}(2 E,d)  \ln \left|  \eta^{-1}  +1    \right| +  2^d \sqrt{2} (4 E + 1)^{d/2} \right] + 2  \| \hat{B}_\# \|_{*,1}  .  \nn
\end{align} 

To estimate the $\| \widehat{B}_\# \|_{*,\infty}$    term appearing in \eqref{CdEB0}  we use  that by Lemma  \ref{fourierseries} (b) 
we have 
\begin{align*}
& \| \widehat{B}_{\#} \|_{*,\infty}  \leq \| B \|_1 < \infty  .
\end{align*}
Finally let us define $c_B$ as the RHS of  \eqref{secondRest}, which is an  $L$-independent upper bound on the $\| \hat{B}_\# \|_{*,1}$ term  appearing  in
 \eqref{CdEB0}.  
Now by \eqref{C007} this  shows that  \eqref{CdEB0}   is uniformly bounded in $L$.
Therefore we have
\begin{align} 
 & C(E,d, L, \eta,\hat{B}_\#) \nn \\
&\leq 2 \| B \|_1  C_{1}(2 E,d)  \ln \left|  \eta^{-1}  +1    \right| +  \left[ 2 \| B \|_12^d \sqrt{2} (4 E + 1)^{d/2}+ 2  c_B \right] \nn
\\
&\leq C^\sharp(d,E,B) \left( \ln \left|  \eta^{-1}  +1    \right|  +1 \right), \label{CdEB}
\end{align} 
for $ {C}^\sharp(d,E,B):= \max \{ 2 \| B \|_1  C_{1}(2 E,d),  2 \| B \|_12^d \sqrt{2} (4 E + 1)^{d/2}+ 2  c_B \}$.
Observe that $|A| \leq n$ since the first moment vanishes, i.e.,   $m_1 = 0$.  Now inserting \eqref{CdEB} in \eqref{0finalestonD} yields
\begin{align} 
\label{DCdEB}
 &\left( \sum_{A \in \mathcal{A}_{2n} }  {D}_{n,A,L}[E + i \eta] \right)^{1/2} \nn \\
 &\leq \left( \sum_{A \in \mathcal{A}_{2n} }    \eta^{-|A|-1}   C_{T,\eta}\| \hat{B}_\# \|_{*,\infty}^{|A|}   [  C(E,d, L, \eta,\hat{B}_\#)   ]^{2n-|A|}   \left|  \prod_{a \in A} m_{|a|}  \right|  \right)^{1/2} \nn \\
 &\leq  C_{T,\eta}^{1/2} \left( \sum_{A \in \mathcal{A}_{2n} }  \eta^{-|A|-1}   \| B \|_1^{|A|}    \left|  \prod_{a \in A}  m_{|a|}  \right|   \left[   C^\sharp(d,E,B) \left(1+ \ln \left(  \eta^{-1}  +1   \right)   \right) \right]^{2n-|A|}   \right)^{1/2} 
 \nn \\
&\leq  C_{T,\eta}^{1/2}  \eta^{-(n+1)/2} \left( 1+ \ln \left(  \eta^{-1}  +1    \right)  \right)^{n}  \left( \sum_{A \in \mathcal{A}_{2n} }   \| B \|_1^{|A|}    \left|  \prod_{a \in A}  m_{|a|} \right|    C^\sharp(d,E,B)  ^{2n-|A|}   \right)^{1/2} \nn \\
& \leq  \left(1 + \eta^{-2} \right)^{1/2}  \eta^{-(n+1)/2} \left( 1+ \ln \left(  \eta^{-1}  +1    \right)  \right)^{n} \tilde{K}^\sharp_{n,d,E,B}
\end{align} 
where we  defined 
\muu{
K^\sharp_{n,d,E,B} :=\left( C_{E,0}  \sum_{A \in \mathcal{A}_{2n} }   \| B \|_1^{|A|}    \left|  \prod_{a \in A}  m_{|a|}  \right|    C^\sharp(d,E,B)  ^{2n-|A|}   \right)^{1/2}. 
}
Note that  the notation of $K^\sharp_{n,d,E,B}$  does not reflect the dependence on the distribution ${\bf P}_v$ an hence the momenta $m_{|a|}$.

Next we use  Lemma  \ref{lemmerorbound}  and  Proposition  \ref{uniformboundontrace} 
to obtain 
\begin{align*}
& \left|  {\bf E } \widetilde{\tr}
{\rm Im} ( H_L  - z )^{-1}   - \sum_{j=0}^{n-1}  \lambda^j \widetilde{\tr}  T_{j,L} [z ]    \right| \\
& \leq   \lambda^n  \mathfrak{E}_{n,L}[z]^{1/2} \left( |{\rm Im} z |^{-1} {\bf E} \widetilde{\tr}  {\rm Im} (( H_L  - z)^{-1}) \right)^{1/2}    \\
& \leq   \lambda^n \left( \sum_{A \in \mathcal{A}_{2n} } {D}_{n,A,L}[z] \right)^{1/2}   \eta^{-1} ( \tilde{C}_E (1+\lambda^2) )^{1/2} 
\\
& \leq  \lambda^n \eta^{-(n+1)/2}   ( 1 + \eta^{-2} )^{1/2}  {K}^\sharp_{n,d,E,B}  \left(  1 + \ln (  \eta^{-1} + 1 )  \right)^{n}
\eta^{-1} ( \tilde{C}_E  (1+\lambda^2) )^{1/2}   ,
\end{align*} 
where in the last line we made use of  the bound  \eqref{Dtoestim1tracial}.
This shows (a) noting that the continuity statement in $E$ follows from the explicit expressions
of the constants.

(b)  This part follows for $\epsilon \in (0,2)$ by  inserting $\eta = \lambda^{2 - \epsilon}$ into (a). 
 In that case    we see that the RHS of 
\eqref{eq:boundonexp}
 tends to 
zero as $\lambda \to 0$ for $n$ sufficiently large,  since logarithms are bounded by  positive 
power. Thus we see that     there exists a constant $C$ such that for  $\eta = \lambda^{2 - \epsilon}$ we have for small $\lambda$ 
\begin{align}
  \left( \frac{ \lambda^2  }{\eta}   \right)^{n/2} \langle \eta^{-1} \rangle  \langle \lambda \rangle \left(1 +   \ln(   \eta^{-1} + 1 )\right)^n    \eta^{-{3/2}}   \leq C  \lambda^{  n - (2-\epsilon) (\frac{n}{2} +1) -     n  \frac{ \epsilon}{4}  - (2 - \epsilon) \frac{3}{2}    } = C  \lambda^{\epsilon (\frac{n}{4}+\frac{5}{2}) -5}.
\end{align} 
Now  $ \lambda^{\epsilon (\frac{n}{4}+\frac{5}{2}) -5} \leq   \lambda^N$ provided  
$
 \epsilon (\frac{n}{4}+\frac{5}{2}) -5   \geq N  ,
$
i.e., 
$
 n    \geq 4  \frac{ N +5}{\epsilon }  - 10   .
$ 
\end{proof}

\label{proofofsec27}

\section{Acknowledgements}

D.H. wants to thank Laszlo Erd\"os for valuable discussions about Feynman graphs. 
We thank Robert Hesse, Benjamin Hinrichs,  and Oliver Siebert for helpful comments.

\appendix

\section{Resolvent Expansion}

In this section we collect elementary resolvent identities, which are used throughout this paper.
For proofs we refer the reader to \cite[Appendix A]{HaslerKoberstein.2025-2}.
\begin{lemma}[Second resolvent identity]
\label{ResId}
Let $T,S$ be linear operators on the normed space $E$ with $D(S)=D(T)$. Let $z \in \rho (T) \cap \rho (S)$. Then 
\muu{
\frac{1}{T-z}-\frac{1}{S-z}=\frac{1}{T-z}(S-T)\frac{1}{S-z} \quad .
}
\end{lemma}
The following lemma is an immediate consequence of Lemma \ref{ResId}.
\begin{lemma}
\label{Resolventen_Anwendung}
Let $A$ and $B$ be linear operators on a normed space $E$ with $D(A)=D(A + B)$ and $0 \in \rho (A) \cap \rho (A + B)$.
Then 
\muu{
\frac{1}{A+B} = \frac{1}{A} + \frac{1}{A+B}(-B)\frac{1}{A} \quad .
}
\end{lemma}
We can now perform the resolvent expansion using an iteration of the formula of the previous lemma.
\begin{lemma}
\label{resolventenformel_interiert}
Let $A$ and $B$ be linear operators on a normed space $E$ with $D(A)=D(A + B)$ and $0 \in \rho (A) \cap \rho (A + B)$.
Then for $m \in \N$ 
\muu{
\frac{1}{A+B} = \sum_{n=0}^{m-1} \frac{1}{A} \left[ (-B) \frac{1}{A}  \right]^n  + \frac{1}{A+B} \left[ (-B) \frac{1}{A}  \right]^m .
}
\end{lemma}
Lemma \ref{resolventenformel_interiert} follows from a straightforward algebraic identity, for details see   \cite{HaslerKoberstein.2025-2}.

\section{Estimates of Fourier Series and Fourier Integrals}

The following lemma will be applied in the proof of  Lemma  
\ref{thm:dosasy0-0}. 
\begin{lemma} \label{lemestfourdisc}    There exists a constant $C_d$ (explicitly given in the proof)
 and  a constant ${C}_d^\sharp$ such that for all $L \geq 1$ and  $f \in \mathcal{S}(\R^d)$ which are 
symmetric  with respect to the reflections $\rho_j : (x_1,...,x_j,...,x_d) \mapsto (x_1,...,-x_j,...,x_d)$ for all $j=1,...,d$ (i.e.  $f \circ \rho_j = f$ for all $j=1,...,d$), 
we have 
\begin{align}  \label{eq:dgenallk0} 
\langle  k_1 \rangle^2 \cdots  \langle k_d \rangle^2 | \widehat{f}_\#(k) |  \leq C_d  \sum_{\substack{\alpha \in \N_0^d : \\  \alpha_j \leq 2}} \sup_{x \in \R^d}  |  \langle x \rangle^{2d} \partial^\alpha f(x) | 
 \end{align} 
for all $k=(k_1, \ldots , k_d) \in \R^d$.
Furthermore $\widehat{f}_\# \in \ell_*^1(\Lambda_L^*)$,  and 
 \begin{align} \label{secondRest0} 
\| \widehat{f}_\# \|_{*,1} \leq {C}^\sharp_d \sum_{\substack{ \alpha \in \N_0^d :  \\ \alpha_j \leq 2}} \sup_{x \in \R^d}| \langle x \rangle^{2d} \partial^\alpha f(x) |  . \end{align} 
\end{lemma} 
A proof of  Lemma \ref{eq:dgenallk0}  can be found in \cite{HaslerKoberstein.2024-2}[Lemma A.3]. In it's proof, one  makes use of Lemma \ref{ChSymbEst}.
The following  Remark will be used in the proof of Theorem \ref{thm:maintec000}.
\begin{remark} {\rm 
We note that if the Schwartz function  $B$  has compact support or if it is symmetric with respect to reflections as in  Lemma \ref{lemestfourdisc}, then there exists an $L_0 \geq 1$ such that for all $L \geq L_0$ we have 
\begin{align}  \label{eq:dgenallk} 
\langle  k_1 \rangle^2 \cdots  \langle k_d \rangle^2 | \widehat{B}_\#(k) |  \leq C_d  \sum_{\substack{\alpha \in \N_0^d : \\  \alpha_j \leq 2}} \sup_{x \in \R^d}  |  \langle x \rangle^{2d} \partial^\alpha B(x) |  .
 \end{align} 
In particular  for $L \geq L_0$ we have $\widehat{B}_\# \in \ell_*^1(\Lambda_L^*)$,  and 
 \begin{align} \label{secondRest} 
\| \widehat{B}_\# \|_{*,1} \leq {C}^\sharp_d \sum_{\substack{ \alpha \in \N_0^d :  \\ \alpha_j \leq 2}} \sup_{x \in \R^d}| \langle x \rangle^{2d} \partial^\alpha B(x) |  . 
\end{align} 
In the compact case \eqref{eq:dgenallk} and \eqref{secondRest} follow from choosing $L_0 \geq 1$ sufficiently large such that $\supp B \subset (-L_0/2,L_0/2)^d$ holds while using  well-known properties of  the Fourier transform of Schwartz functions.    In the symmetric case  \eqref{eq:dgenallk} and \eqref{secondRest} follow  from Lemma \ref{lemestfourdisc}.}
\end{remark}

For the next Lemma we recall the definition of the Fourier transform in  \eqref{defoffourier}  for any $f \in L^1(\R^d)$. 
The Lemma will be needed in the proof of  Proposition \ref{propexofcinfDOS},  showing   that  the limit $L \to \infty $ of   \eqref{eq:defofCpi2} exists.

\begin{lemma} \label{fourierseries}  The following holds. 
\begin{enumerate}[(a)]
\item 
Let $f , g \in L^2(\R^d)$, then 
$
\int_{\Lambda_L^*} \overline{\widehat{f}}_\#(k) \widehat{g}_\#(k) dk =  \int_{\Lambda_L} \overline{ f(x)} g(x) dx  .
$
\item 
If $f \in L^1(\R^d)$, then  
$
\| \hat{f}_\# \|_{*,\infty} \leq \int_{\Lambda_L}  |f(x)|  dx \leq \| f \|_1 .  
$
\item 
 \label{discconvgarten} If  $f \in L^1(\R^d)$,  then 
$
\| \hat{f}_\# - \hat{f} \|_{*,\infty} \to 0 \quad ( L \to \infty ) . 
$

\end{enumerate} 
\end{lemma} 
The  proof of  Lemma \ref{fourierseries}  uses elementary estimates, for details we refer the reader to  \cite{HaslerKoberstein.2025-2}.

\section{Elementary Estimates}

In this appendix we collect elementary bounds, which are needed for the infinite volume limit 
of the expansion coefficients. 
The following lemma is used in  the proof of  Proposition \ref{propexofcinfDOS}.

\begin{lemma}\label{lem:boundonsquaring}   Let  $a,b\in\R$,  $b \neq 0$, and   $\delta  = \frac{b^2}{ b^2+ 2 a^2}$.
Then for all $x \in \R$ 
\begin{align} \label{elemboundonsquaring} 
(x- a)^2 +b^2 \geq   \delta x^2 + \frac{1}{2} b^2 . 
\end{align} 
\end{lemma}
\begin{proof}
We estimate by means of the squaring inequality  and find by inserting the value of $\delta$ that for all $x \in \R$ 
\begin{align*}
(x- a)^2 +b^2  & =     x^2 - 2 a x + a^2 + b^2 \\
& \geq  x^2 - (1-\delta) x^2 +  ( 1 - (1-\delta)^{-1} ) a^2 + b^2 = \delta x^2 + \frac{1}{2} b^2. 
\end{align*} 
\end{proof}

\section{Bound on the Trace of the Resolvent}

\label{appendixDboundonres} 


For an open set $\Omega \subset \R^d$ we shall denote by $\Delta_{\Omega,N}$ the Neumann Laplacian
on $\Omega$.

\begin{proposition}\label{lem:proofofexdens00-0}  Let $d \leq 3$, $c > 0$, and $\alpha \in (3/2 ,2]$. Then there exists a constant $C$ 
such that  for all $L \geq 1$ and $\lambda \geq 0$ 
\begin{align} \label{tracialboundneededforasy}
|\Lambda_L|^{-1}{\bf E}_L {\rm tr}   \frac{1}{\langle - c \Delta_L + \lambda V \rangle^{\alpha}}\leq C (1+\lambda^2) .  
\end{align} 
\end{proposition}

\begin{proof} First we partition  the set $\Lambda_L$ into subcubes of size $1$ and boundary cubes. 
Explicitly, we 
define  $I_L := \{ m \in \Z^d : ( \Lambda_1 + m )  \cap \Lambda_L \neq \emptyset \}$, see  Figure \ref{Decomposition}.
 Then for 
$$
Q_{m,L} :=   ( \Lambda_1 + m ) \cap \Lambda_L 
$$
we have 
$$
\bigcup_{n \in I_L} Q_{m,L}  =  \Lambda_L . 
$$
For a set $S \subset \R^d$ we shall denote by $S^\circ$ its interior. 
We define 
$$
S_L :=  \bigcup_{n \in I_L} Q_{m,L}^\circ .
$$
It  follows  from  well-known properties \cite{SR4} \cite[Page 55]{BratelliRobinson.1997}  of boundary conditions that 
\begin{align} 
\Delta_{S_L,N}  :=   \sum_{ n \in  I_L  } \Delta_{ Q_{m,L}^\circ , N}    \leq  \Delta_L    . 
\end{align}
Furthermore for a function $f$ we  define the function giving the supremum on each individual block   
\begin{align}
M_{S_L}(f):=\sum_{m \in I_L} 1_{Q_{m,L} }  \sup_{\xi \in Q_{m,L}}  f(\xi)  .
\end{align} 
Furthermore we define the negative part of a function by $f_- := \max\{ 0 , - f \}$. 
Thus we obtain 
 \begin{align} \label{mainboundonboxes} 
- c \Delta_{S_L,N}  -   M_{S_L}( \lambda V_-)    \leq  -c \Delta_L   +  \lambda V  . 
\end{align}

We define the function  $\chi  : \R \to \R$ by 
\begin{align}
\chi(x) = \left\{ \begin{array}{ll} 1  & , \ x \leq  0 \\ \langle x \rangle^{-\alpha}  & ,  \ x > 0  . \end{array} \right. 
\end{align} 
Then  $\chi$ is continuous,  monotone decreasing, and for any $m \geq 0$, we find for $x \geq 0$ 
\begin{align} \label{auxfuncest} 
\chi(x - m ) \leq 2^{\alpha+1} \frac{1 + m^\alpha}{1+x^\alpha} .
\end{align} 
To see \eqref{auxfuncest}, we observe that for $0 \leq x \leq 2m$ 
\begin{align*}
\chi(x - m ) \leq 1  = \frac{1 + 2^\alpha m^\alpha}{1+(2m)^\alpha} \leq  2^\alpha \frac{1+  m^\alpha}{1+x^\alpha} ,
\end{align*} 
and for $x \geq 2m$ 
\begin{align*}
\chi(x - m ) =\frac{1 }{\langle x-m \rangle^\alpha} =  
\frac{1 }{(1+(x/2+x/2-m)^2)^{\alpha/2}} \leq \frac{1 }{(1+(x/2)^2)^{\alpha/2}}
 \leq \frac{2^{\alpha+1}}{1+x^\alpha} .
\end{align*} 
This shows \eqref{auxfuncest}. Now using $ \langle x \rangle^{-\alpha}  \leq \chi(x)$ we find by operator monotonicity 
\begin{align} \label{tracefirst:ghj-001} 
{\rm tr}  \langle  - c \Delta_L + \lambda V \rangle^{-\alpha }  \leq {\rm tr} \chi( - c \Delta_L + \lambda V) .
\end{align} 
Using trace monotonicity, or more precisely that $\chi$ is monotone decreasing and the min max principle we find from \eqref{mainboundonboxes} 
\begin{align} \label{tracefirst:ghj-002} 
 {\rm tr} \chi( - c \Delta_L + \lambda V)   \leq {\rm tr} \chi( - c \Delta_{S_L,N}  - M_{S_L}(  \lambda V_-) ) 
\end{align} 
Let $P_{L^2(Q_{L,n})}$ denote the projection  in $L^2(\R^d)$ or $L^2(\Lambda_L)$ onto the subspace $L^2(Q_{L,n})$.
Then 
\begin{align} \label{partitionhilbertsub} 
 \sum_{  n \in I_L } P_{L^2(Q_{L,n})} = {\rm id}_{L^2(\Lambda_L) }  . 
\end{align}  
We find 
with \eqref{partitionhilbertsub}  and \eqref{auxfuncest}
\begin{align} \label{tracefirst:ghj} 
 & {\rm tr} \chi( - c \Delta_{S_L,N}  - M_{S_L}( \lambda V_-) )  = 
  \sum_{  n \in I_L } {\rm tr}  P_{L^2(Q_{L,n} ) } \chi( - c \Delta_{S_L,N}  - M_{S_L}( \lambda V_-) ) \\
&\leq 2^{\alpha+1}  \sum_{  n \in I_L } {\rm tr}P_{L^2(Q_{L,n} ) }\frac{1+  M_{S_L}( \lambda V_-)^\alpha}{ 1 + ( - c \Delta_{S_L,N})^\alpha } 
 \leq  2^{\alpha+1} \sum_{  n \in I_L }  ( \lambda^\alpha \| V  1_{Q_{L,n}} \|_{\infty}^\alpha + 1) {\rm tr}_{L^2(Q_{L,n} ) } \frac{  1  }{  (- c \Delta_{Q_{L,n} , N})^\alpha + 1 }  
 \nn , 
\end{align}
where in the last equality we used that  $0 \leq  M_{S_L}( \lambda V_-) 1_{Q_{L,n}}  \leq \| \lambda V 1_{Q_{L,n}} \|_\infty$. 
Concatenating Ineq.  \eqref{tracefirst:ghj-001}, \eqref{tracefirst:ghj-002}, \eqref{tracefirst:ghj}  we arrive at  
\begin{align} \label{tracefirst:ghj-33} 
{\rm tr} \langle - c \Delta_L + \lambda V    \rangle^{- \alpha   } 
& \leq  2^{\alpha+2} \sum_{  n \in I_L }  ( \lambda^\alpha \| V  1_{Q_{L,n}} \|_{\infty}^\alpha + 1) {\rm tr}_{L^2(Q_{L,n} ) } \frac{  1  }{  (- c \Delta_{Q_{L,n} , N})^\alpha + 1 }  
\end{align}

Now we calculate  the normalized trace. We 
divide into  the interior cubes  and boundary cubes. 
For this we write 
$I_L^{\rm int} = \{ m \in \Z^d : \Lambda_1 + m \subset \Lambda_L \}$ and
$I_L^{\rm bdry} =  I_L \setminus I_L^{\rm int} $.

\begin{figure}[H]
\begin{center}

\begin{tikzpicture}[scale=0.8]
  \draw[->] (-2.5, 0.25) -- (3, 0.25) node[right] {$\R$};
  \draw[->] (0.25 ,  - 2.5) -- (0.25,3) node[above] {$\R$};
     \draw[-] (-2.5, 0) -- (3, 0) ;
     \draw[-] (-2.5, 0.5) -- (3, 0.5) ;
       \draw[-] (-2.5, 1) -- (3, 1) ;
   \draw[-] (-2.5, 1.5) -- (3, 1.5) ;
      \draw[-] (-2.5, 2) -- (3, 2) ;
      \draw[-] (-2.5, 2.5) -- (3, 2.5) ;
           \draw[-] (-2.5, -0.5) -- (3, -0.5) ;
       \draw[-] (-2.5, -1) -- (3, -1) ;
   \draw[-] (-2.5, -1.5) -- (3, -1.5) ;
      \draw[-] (-2.5, -2) -- (3, -2) ;
             \draw[-] (-2.5, 1) -- (3, 1) ;
   \draw[-] (-2.5, 1.5) -- (3, 1.5) ;
      \draw[-] (-2.5, 2) -- (3, 2) ;
      \draw[-] (-2.5, 2.5) -- (3, 2.5) ;
               \draw[-] ( 0,-2.5) -- (0,3) ;
           \draw[-] ( -0.5,-2.5) -- (-0.5,3) ;
       \draw[-] (-1,-2.5) -- (-1,3) ;
   \draw[-] (-1.5,-2.5) -- (-1.5,3) ;
      \draw[-] (-2,-2.5) -- ( -2,3) ;
                 \draw[-] ( 0.5,-2.5) -- (0.5,3) ;
       \draw[-] (1,-2.5) -- (1,3) ;
   \draw[-] (1.5,-2.5) -- (1.5,3) ;
      \draw[-] (2,-2.5) -- ( 2,3) ;
      \draw[-] ( 2.5,-2.5) -- ( 2.5,3) ;
      \draw[red,-] ( 2.25,-1.75) -- (2.25,2.25) node[right]{$\Lambda_L$}; 
            \draw[red,-] ( -1.75,-1.75) -- (-1.75, 2.25) ; 
            \draw[red,-] ( -1.75,2.25) -- (2.25, 2.25) ; 
            \draw[red,-] ( -1.75, -1.75) -- (2.25, - 1.75) ; 
            \draw[fill=magenta, opacity=0.1]  (-1.5,-1.5) rectangle (2,2);
                   \draw[fill=cyan, opacity=0.1]  (-1.75,2) rectangle (2.25,2.25);
                 \draw[fill=cyan, opacity=0.1]  (-1.75,-1.75) rectangle (2.25,-1.5);
                       \draw[fill=cyan, opacity=0.1]  (-1.75,-1.5) rectangle (-1.5,2);
                  \draw[fill=cyan, opacity=0.1]  (2,-1.5) rectangle (2.25,2);
         \foreach \x in {-3,-2,-1,0,1,2,3}
         {  \foreach \y in {-3,-2,-1,0,1,2,3}
         {
                \filldraw[red] (0.25+\x*0.5,0.25+\y*0.5) circle (0.5pt);
                }}
         \foreach \x in {-4,4}
         {  \foreach \y in {-4,-3,-2,-1,0,1,2,3,4}
         {
                \filldraw[blue] (0.25+\x*0.5,0.25+\y*0.5) circle (0.5pt);
                }}
          \foreach \y in {-4,4}
         {  \foreach \x in {-3,-2,-1,0,1,2,3}
         {
                \filldraw[blue] (0.25+\x*0.5,0.25+\y*0.5) circle (0.5pt);
                }}               
\end{tikzpicture}

\caption{\label{Decomposition} \small   The set   $I_L$ is indicated by the points.
The  points which are  coloured with red are the points in  
$I_L^{\rm int}$,  and the  points which are  coloured with blue are the points in  
$I_L^{\rm bdry} $.
The  rectangles   $Q_{m,L}$ with $m \in I_L^{\rm int}$ are coloured 
with red, and   rectangles   $Q_{m,L}$ with $m \in I_L^{\rm bdry}$ are coloured 
with blue.  }
\end{center} 
\end{figure}
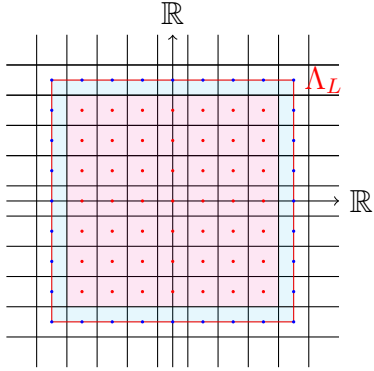

We have   $$ | I_L^{\rm int} | |\Lambda_1|   \leq  |\Lambda_L| $$ and 
there exists a constant $C_d$ such that 
 for $L \geq 1$  
$$
| I_L^{\rm bdry}| |\Lambda_1|  \leq  C_d  |\Lambda_L|  . 
$$
 Thus for large $L$ we find from  \eqref{tracefirst:ghj-33} 
 \begin{align} \label{tracefirst:ghj-1fin} 
& \frac{1}{|\Lambda_L|} {\rm tr}  \langle - c \Delta_L + \lambda V    \rangle^{- \alpha   }   \nn  \\
&  \leq  
 \frac{2^{\alpha+2}}{|I_L^{\rm int} |}  
 \sum_{   n \in I_L^{\rm int}  }   ( \lambda^\alpha\| V  1_{Q_{L,n}} \|_{\infty}^\alpha + 1)   
\frac{1}{ |\Lambda_1|}    {\rm tr}_{L^2(Q_{L,n} ) } \frac{  1  }{  (- c  \Delta_{Q_{L,n} , N})^\alpha + 1 } 
 \nn \\
& +  \frac{ 2^{\alpha+2 }C_d}{|I_L^{\rm bdry} |}  
 \sum_{   n \in I_L^{\rm bdry}  }  ( \lambda^\alpha\| V  1_{Q_{L,n}} \|_{\infty}^\alpha + 1)   
\frac{1}{ |\Lambda_1|}    {\rm tr}_{L^2(Q_{L,n} ) } \frac{  1  }{  (- c  \Delta_{Q_{L,n} , N})^\alpha + 1 } 
\end{align} 
Now the proof follows using  Lemmas  \ref{boundintracfreeneumann}  and  
 \ref{boundintraceexppot} to estimate the right hand side of \eqref{tracefirst:ghj-1fin}. 
\end{proof} 

\begin{lemma}\label{boundintracfreeneumann} Let $a \in \R^d$ with $a_j > 0$ for $j=1,...,d$. 
Let  $S_a = \prod_{j=1}^d (-a_j/2,a_j/2)$. Then
\begin{align} \label{neumannboxestimate} 
 {\rm tr} \left( \frac{1}{ \left(- c \Delta_{S_a,N}\right)^\beta + 1 }   \right)
& =    \sum_{n \in \N_0^d} \frac{1}{  1 + 
\left(  c  \pi^2 \sum_{j=1}^d \left(n_j/a_j\right)^2 \right)^\beta } . 
\end{align} 
In particular if $a_j \leq 1$  and  $\beta > d/2$  we find 
 \begin{align} \label{neumannboxestimate-w2} 
 &     \sum_{n \in \N_0^d} \frac{1}{  1 + 
\left(  c  \pi^2 \sum_{j=1}^d \left(n_j/a_j\right)^2 \right)^\beta } \leq   \sum_{n \in \N_0^d} \frac{1}{  1 + 
\left(  c  \pi^2 \sum_{j=1}^d \left(n_j\right)^2 \right)^\beta } \nn \\
&  \leq  \int_{\Lambda_{1},*} \frac{ dp   }{1+ \left( c \left( \pi  p\right)^2\right)^\beta }   < \infty .
\end{align} 
\end{lemma}
\begin{proof} 

An  eigenbasis for $-\Delta_{N,S_a}$ is given by  
\begin{align}
\Psi_{n,a}(x)  =\prod_{j=1}^d \{ (2/a_j)^{1/2} \psi_{n_j}(x_j/a_j) \}  , \quad n \in \N_0^d,
\end{align} 
where 
\begin{align}
\psi_{k} = \sin(k \pi x) ,  & \quad k=1,3,5,... \\
\psi_{k} = \cos(k \pi x) ,  & \quad k=2,4,6,... \\
\psi_{k} = \frac{1}{2} \sqrt{2} ,  & \quad k= 0  ,
\end{align} 
cf. \cite[Page 266]{SR4}.
The eigenvalues are readily calculated giving
\begin{align}
-\Delta_{N,(-a,a)^d} \Psi_{n,a}  = E_n(a) \Psi_{n,a} 
\end{align} 
\begin{align}
E_n(a) =  \pi^2 \sum_{j=1}^d (n_j/a_j)^2 .
\end{align} 
Thus using this basis to calculate the trace we find 
by first  observing  that  translation is a unitary transformation, we find    for $d \leq 3$ 
\begin{align}
&   {\rm tr}_{L^2(\Lambda_{2a}) } \frac{  1  }{  (-  c  \Delta_{S_a , N})^\beta + 1 } \\
& =   \sum_{n \in \N_0^d}  \SP{ \Psi_{n,a},  \left(1+\left( c E_n(a_L)  \right)^\beta \right)^{-1}\Psi_{n,a}}  \nonumber 
\\
& =    \sum_{n \in \N_0^d}  \left(1+\left( c \pi^2
 \sum_{j=1}^d \left(\frac{ n_j}{a }\right)^2 \right)^\beta  \right)^{-1}
\label{ineqddd-0} 
\end{align} 
This shows \eqref{neumannboxestimate}. Now the first inequality of \eqref{neumannboxestimate-w2}
follows by monotonicity  and finiteness   by  Lemma \ref{lem:discest}. 
\end{proof}

\begin{lemma}\label{boundintraceexppot}  Let  $\beta \in [0, 2]$ Then there  exists a constant $C_{\rm P}(B)$  such that 
for all $S \subset \Lambda_L$ with  $S \subset \Lambda_1 + p$ for some  $p \in \Z^d$ 
\begin{align} 
{\bf E}_L \| V 1_{S} \|_\infty^\beta  \leq   C_{\rm P}(B) .
\end{align} 
\end{lemma} 
\begin{proof} 
Using the ${\bf E}_L$ is the expectation of a probability measure, we find by H\"older inequality
 ${\bf E}_L \| V 1_{S} \|_\infty^\beta   \leq ( {\bf E}_L \| V 1_{S} \|_\infty^2)^{\frac{\beta}{2}}$. Thus  
it suffices to consider without loss $\beta =2$. 
Let $p \in \Z^d$. Let   $S \subset \Lambda_1 + p$. 
Since $B$ is a Schwartz function there exists for any $a > 0$ a constant $C_{B,a}$ 
such that 
$$
|B(x) | \leq  C_{B,a} \langle x \rangle^{-a} 
$$
If $x , y \in \Lambda_L$, then there exists an $n \in \Z^d$ with $\| n \|_\infty \leq 1$ such that 
$$
x - y + n L \in \Lambda_L 
$$
Thus it follows that for $I = \{ n \in \Z^d : \| n \|_\infty \leq  1 \}$  
$$
|B_\#(x-y) | \leq \sum_{n \in I  }  C_{B,a} \langle x - y + n L  \rangle^{-a} .
$$

To take the expectation we use the restriction property of the Poisson process and introduce 
the following Processes. 
Let $(\hat{y}_{m,\gamma})_{\gamma=1,...,\hat{l}_m}$  denote Poisson points distributed in $(\Lambda_1+m) \cap \Lambda_L$ with uniform density. Let $v_{i,j}$ with $(i,j) \in \N^2$ denote 
i.i.d. random variables with common distribution  ${\bf P}_v$.   We denote the expectation
 with respect to these random variables by ${\bf E}^{(1)}$. 
By the restriction property of  the Poisson process  \cite{Kingman.1993} we find 
for a subset $S \subset \Lambda_L$ 
\begin{align} {\bf E}_L    \| V 1_{S} \|_\infty^2  
 & \leq {\bf E}_L  \left(  \sup_{x \in S } | \int_{\Lambda_L} B_{\#}(x-y) d\mu_{L,\omega}(y)  |
   \right)^2 
  \label{probestsupbox-1}  \\
    & \leq {\bf E}_L  \sup_{x \in  S}  \left(  \sum_{m \in \Z^d} | \int_{(\Lambda_1+m) \cap \Lambda_L} B_{\#}(x-y) d\mu_{L,\omega}(y)  |
   \right)^2   \nn \\
       & =  {\bf E}^{(1)}  \sup_{x \in  S}  \left(  \sum_{m \in \Z^d} | \sum_{\gamma_m=1}^{\hat{l}_m} v_{m,\gamma_m}   B_{\#}(x-\hat{y}_{m,\gamma_m})   |
   \right)^2   \nn \\
             & \leq   {\bf E}^{(1)}  \sup_{x \in  S}  \left(  \sum_{m \in \Z^d}  \sum_{\gamma_m=1}^{\hat{l}_m} |v_{m,\gamma_m}
             | \sum_{n \in I}   C_{B,a} \langle x   - \hat{y}_{m,\gamma_m} - n L ) \rangle^{-a}    
   \right)^2  \nn   
\end{align} 
Now we use that there 
exists a positive  constant $c$  depending on  the dimension  such that 
\begin{align} \label{geomesttriv-34}
\inf _{x \in S   , y \in  \Lambda_1 + m }  \langle x - y \rangle  \geq \inf _{x \in \Lambda_1   , y \in  \Lambda_1  }  \langle x +y + p  - m - n L  \rangle   \geq c   \langle  p - m - n L  \rangle  
\end{align} 
for all $p,m \in \Z^d$ and $L \geq 1$. 
Inserting  \eqref{geomesttriv-34} into  \eqref{probestsupbox-1} 
we find 
\begin{align} {\bf E}_L    \| V 1_{S} \|_\infty^2  
                & \leq   {\bf E}^{(1)}  \left(  \sum_{m \in \Z^d}  \sum_{n \in I}  \sum_{\gamma_m=1}^{l_m} |v_{m,\gamma_m} | c^{-a} C_{B,a} \langle m - p + n L  \rangle^{-a}    
   \right)^2  \nn
\end{align} 

Now we calculate  the square and find
\begin{align} & {\bf E}_L    \| V 1_{S} \|_\infty^2  \\
    & \leq  c^{-2a}C_{B,a}^2    {\bf E}_L    \sum_{m, m' \in \Z^d}   \sum_{n , n' \in I }   \sum_{\gamma_m=1}^{l_m} 
     \sum_{\gamma_{m'}'=1}^{l_{m'}} |v_{m,\gamma_m} | |v_{m',\gamma_{m'}'} | \langle m - p + n L  \rangle^{-a}  \langle m' - p + n' L  \rangle^{-a}  \nn \\
         & \leq c^{-2a} C_{B,a}^2    {\bf E}_L    \sum_{m, m' \in \Z^d} \sum_{n , n' \in I }   l_m l_{m'} m_2 \langle m - p + n L  \rangle^{-a}  \langle m'  - p +n' L \rangle^{-a}  , \nn 
\end{align} 
 where we used that $\E_v |v_{\gamma} v_{\delta}| \leq  \frac{1}{2} \E_v ( |v_\gamma|^2 + |v_\delta|^2) = m_2$ 
 for any $\gamma,\delta \in \{1,... , l \}$.
Observing that $l_m$ is  Poisson random variable with mean less or equal to   $|\Lambda_1|=1$ we  find ${\bf E} l_m \leq  1$ and ${\bf E} l_m^2 \leq  2$.
Thus using independence of $l_m$ and $l_{m'}$ if $m \neq m'$ we arrive at 
\begin{align} & {\bf E}_L    \| V 1_{S} \|_\infty^2  \label{estonexp}\\
                  & \leq  c^{-2a}C_{B,a}^2      \sum_{m, m' \in \Z^d} \sum_{n , n' \in I }   (1+ \delta_{m,m'})  m_2 \langle m - p + n L  \rangle^{-a}  \langle m' -p + n' L  \rangle^{-a}   =: C_{\rm P}(B,p,L)  , \nn 
\end{align} 
which is finite provided $a $ is sufficiently large. 
Now by translation invariance we see that 
$$
C_{\rm P}(B)   := \sup_{L\geq 1, p \in \Z^d} C_{\rm P}(B,p,L)  < \infty 
$$
This shows the lemma. 
\end{proof}

\section{Estimates on Integrals} 


In this Appendix, we present a collection of estimates on integrals as well as on discrete integrals, i.e.,  infinite sums.
Let us first mention a basic inequality.  For any $z \in \C \setminus \R$ we have the elementary bound 
\begin{equation}  \label{ln-1}
   | z^{-1}|   \leq   \left|   {\rm Im} z \right|^{-1}  .
\end{equation} 
Moreover, for $\lambda \geq 0$ and $z \in \C \setminus [0,\infty)$ the following bound is straightforward to see  
\begin{equation}  \label{ln-2}
   | (z - \lambda)^{-1}|   \leq  {\rm dist}(z,\R_\geq )^{-1}  .
\end{equation} 
The following lemma is a simple estimate,  which is applied in the proof of Lemma \ref{lemestfourdisc}.
\begin{lemma}
\label{ChSymbEst}
For  $a,b \in \R$ with $a \leq b$ and  $ \langle \cdot \rangle^2 f \in L^\infty$ with $\langle \cdot \rangle$ defined in \eqref{Klammern}.  Then  $f \in L^1(\R)$ and
 \begin{align*}
\int_{a}^b  \left|   f(x) \right|   dx \leq \pi \sup_{x \in \R} | \langle x \rangle^2 f(x) | .
\end{align*}
\end{lemma}
The lemma follows from a straightforward estimate, for an explicit calculation  
we refer the reader to \cite{HaslerKoberstein.2024-2}[Lemma A.2].
Moreover we shall need to estimate tracial expressions of resolvents. 
Let us now  estimate a  finite volume expression. For this we shall make use of the following elementary lemma about Riemann integrals.

\begin{lemma}\label{lem:discrriem}  Let $I  = [-c,c]^d$ and $I_j$, $j=1,...,N$, be   a partition (up to boundaries) of $I$  into translates of a square  $Q$ which is centered at the 
origin. Let $\xi_j  \in  I_j$
and $ \Delta x_j = | I_j|$. Suppose $g \geq 0$ is a Riemann integrable function on $I$  and we have 
\begin{equation} \label{eq:estonint} 
\sup_{ \xi \in I_j} f(\xi) \leq \inf_{ \xi \in I_j} g(\xi)  . 
\end{equation} 
Then 
\begin{align*} 
  \left|  \sum_{i=1}^N f(\xi_j) \Delta x_j  \right| \leq  \int_I g(x) dx  .
\end{align*} 
\end{lemma}
\begin{proof}  The statement follows directly from the theory of Riemann integration. 
\end{proof}

We now want to show Proposition \ref{C007} which is used in Section  
\ref{sec:proooferror} to estimate integrals over resolvent-type-functions. To do so, we are first discussing the relation between the discrete sum and the Riemann integral in Lemma \ref{ln0}, to then give a bound on the Riemann integral in Lemma \ref{ln00},  with which we can then show Proposition \ref{C007}.

\begin{lemma} Let $f : \Lambda_L^* \to \C$ and $E \geq  0$. Then 
\label{ln0}
$$
\int_{\Lambda_L^*}   \left|  \frac{  f(q) }{  q^2 - E -  i \eta } \right|    dq 
  \leq C_0(E,d, L , \eta,f)  ,
$$
where 
\begin{align} \label{C0}
C_0(E,d, L, \eta,f) := \| f \|_{*,\infty} \int\limits_{ \| q \|_\infty \leq  \sqrt{ 2  E  + 1  }  }   |q^2 - E - i 2^{-1/2}  \eta  |^{-1}   dq  + (E +1 )^{-1} \| f \|_{*,1} .
\end{align}
where we introduced the notation $\| q \|_\infty = \max_{j=1}^d |q_j|$. 
\end{lemma} 
 To prove Lemma \ref{ln0}  one can use  Lemma \ref{lem:discrriem} and  verify  \eqref{eq:estonint}
by a straightforward estimate. For details see \cite{HaslerKoberstein.2025-2}. 

\begin{lemma} 
\label{ln00}
For $f \in L^1 (\R^d) \cap L^\infty (\R^d)$,  $E >0 $,  $\eta >0$ and $d \in \N$ we have
\begin{align*} 
 \int_{\R^d}     \left|   \frac{ f(q) }{ q^2 - E \pm  i \eta } \right|  dq 
      \leq  C_1(E,d)  \| f \|_\infty    \ln \left(  \frac{1}{ \eta }  +1    \right) + \sqrt{2} \| f \|_1,
\end{align*}
where 
\begin{align} \label{C1}
C_1(E,d):=  \sqrt{2}  \left(  E^{-1/2} ( E+1)^{\frac{d-1}{2}}  + E^{\frac{d-2}{2}} \right)  |S_{d-1}|, 
\end{align}
 with  $|S_{d-1}|$  being the volume of the $d-1$-dimensional sphere with radius $1$.

\end{lemma}

The lemma follows from a straightforward estimate, for an explicit calculation  
we refer the reader to \cite{HaslerKoberstein.2025-2}, where we note that the absolute value can be taken into the integral by absorbing the phase into the function $f$.  
The following proposition will be needed in Section \ref{sec:ausintegrieren}.

\begin{proposition}  
\label{ln000}
Let $E > 0$ and $\eta > 0$. Then for all $L \geq 1$ and $f : \Lambda_L^* \to \C$ we have  

$$
  \int_{\Lambda_L^*}   \left|  f(q)  \left( \frac{1}{2} q^2 - E \pm  i \eta \right)^{-1} \right|   dq   
  \leq C(E,d, L , \eta,f) 
$$
where 
\begin{align} \label{C007}
C(E,d, L , \eta,f) :=  2 \| f \|_{*,\infty}  \left[  C_{1}(2 E,d)  \ln \left(  \eta^{-1}  +1    \right)+  2^d \sqrt{2} (4 E + 1)^{d/2} \right] + 2  \| f \|_{*,1} 
\end{align}
with $C_1(E,d)$ defined in  \eqref{C1}. 
\end{proposition} 
 Proposition \ref{ln000} can be proven  by  combining Lemmas  \ref{ln0} and  \ref{ln00}.
For details, see again  \cite{HaslerKoberstein.2025-2}.
The following Lemma will be  used  for the proofs of   
 Lemmas  \ref{boundintracfreeneumann} and  \ref{finitetracesumm}.

\begin{lemma} \label{lem:discest} 
Let $d\leq 3$ and $\tau\in[0,4-d)$.
 For all compact subsets $I \subset \R$ there exists a constant $C_{I,\tau}$ such that 
    \[
        \int_{\Lambda_L^*}
        \frac{\langle a\rangle^\tau}
             {|a^2- z |^2}\,da
        \leq C_{I,\tau}\bigl(1+ ({\rm dist}(z,[0,\infty)))^{-2}\bigr)
    \]
for all $L \geq 1$ and $z \in \C \setminus [0,\infty)$ with ${\rm Re} z \in I$. 
\end{lemma}
\begin{proof} Let $z   = E + i \eta$.
$$
$$
Suppose first that $E \geq 0$. 
We split the integral into the regions
\[
    a^2\leq 4  E +4
    \qquad\text{and}\qquad
    a^2>4 E +4.
\]
On the first region,
\[
    |a^2-E-i\eta|^2
    =(a^2-E)^2+\eta^2
    \geq \eta^2,
\]
and on  the complementary region we have
\[
 |a^2-E-i\eta| = \sqrt{(a^2 - E)^2 + \eta^2 }  \geq    |a^2-E|
    \geq a^2- E
    \geq \frac12a^2+1 . 
\]
Hence
\begin{align} 
 &    \int_{\Lambda_L^*}
    \frac{\langle a\rangle^\tau}
         {|a^2-E-i\eta|^2}\,da \nn  \\
 & = 
    \int_{\Lambda_L^*}
    \frac{\langle a\rangle^\tau
          \mathbf{1}_{\{a^2\leq4 E +4\}}}
         {|a^2-E-i\eta|^2}\,da 
+     \int_{\Lambda_L^*}
    \frac{
          \mathbf{1}_{\{a^2 > 4 E +4\}}\langle a\rangle^\tau  }
         {|a^2-E-i\eta|^2}\,da \nn  \\
& \leq   |\eta|^{-2} \int_{\Lambda_L^*} \mathbf{1}_{\{a^2\leq4 E +4\}} 
    \langle a\rangle^\tau
         \,da 
+  \int_{\Lambda_L^*}  \mathbf{1}_{\{a^2 > 4 E+4\}} 4 \langle a \rangle^{\tau - 4} 
         \,da  \label{firstmainbound67} 
\end{align} 
We define  for $E \geq 0$ 
\begin{align} \label{boundonIntegralsmall-3} 
C_{+,E,\tau,L}  :=  \int_{\Lambda_L^*} \mathbf{1}_{\{a^2\leq4 E+4\}} 
    \langle a\rangle^\tau
         \,da  
+   \int_{\Lambda_L^*}  \mathbf{1}_{\{a^2 > 4 E +4\}} 4 \langle a \rangle^{\tau - 4} 
         \,da  
\end{align} 
Thus for  ${\rm Re} z \geq 0$, we have ${\rm dist}(z,[0,\infty)) = |{\rm Im} z |$ and so   \eqref{firstmainbound67} 
implies that 
\begin{align} \label{boundonIntegralsmall-0} 
 &    \int_{\Lambda_L^*}
    \frac{\langle a\rangle^\tau}
         {|a^2-E-i\eta|^2}\,da  \leq  
C_{+,E,\tau,L}( ({\rm dist}(z,[0,\infty)))^{-2} + 1 )  .
\end{align} 
for all $L \geq 1$. 

$$
$$
Now assume that  $E \leq  0$ and $|z| \neq 0$. Then  
\[
    |a^2-E-i\eta|^2
    =(a^2+|E|)^2+\eta^2
\]
Consequently,
\begin{align}  \label{estEsmallerthanzero} 
  &   \int_{\Lambda_L^*}
    \frac{\langle a\rangle^\tau}
         {|a^2-E-i\eta|^2}\,da
   \nn  \\ &  =
    \int_{\Lambda_L^*}
    \frac{\langle a\rangle^\tau}
         {   (a^2+|E|)^2+\eta^2     }\,da \nn  \\
&  =  \int_{\Lambda_L^*} 1_{|a| \leq 1}   
  \frac{\langle a\rangle^\tau}
         {   (a^2+|E|)^2+\eta^2     } +
 \int_{\Lambda_L^*} 1_{|a|  >  1}  
  \frac{\langle a\rangle^\tau}
         {   (a^2+|E|)^2+\eta^2     } \nn \\
& \leq  \frac{1}{|E|^2+ \eta^2} \int_{\Lambda_L^*}  2^{\tau/2} 1_{|a|  \leq   1}  \,da  +
 \int_{\Lambda_L^*} 1_{|a|  >  1}   \frac{4 \langle a\rangle^\tau}{(a^2+1)^2}\,da  .
 \end{align} 
Now we use that ${\rm dist}(z,[0,\infty)) =   | z|$, if ${\rm Re} z  \leq  0$
and define
\begin{align} \label{defofC-tau} 
C_{-,\tau,L} :=   \int_{\Lambda_L^*} 2^{\tau/2} 1_{|a| \leq 1}  \,da  +
 \int_{\Lambda_L^*} 1_{|a|  >  1}   \frac{4 \langle a\rangle^\tau}{(a^2+1)^2}\,da  
\end{align} 
Thus we see from  \eqref{estEsmallerthanzero}  that  for all $z \in \C \setminus [ 0,\infty)$ with 
  ${\rm Re} z \leq  0$  we have for all $L \geq 1$ 
     \begin{align} \label{boundonIntegralsmall} 
        \int_{\Lambda_L^*}
        \frac{\langle a\rangle^\tau}
             {|a^2- z |^2}\,da
        \leq C_{-,\tau,L}\bigl(1+ ({\rm dist}(z,[0,\infty)))^{-2}\bigr)
    \end{align} 
The desired estimate now follows using  \eqref{boundonIntegralsmall-0} and \eqref{boundonIntegralsmall} 
on the respective regions. To this end, we note that 
using Lemma     \ref{lem:discrriem} it is straightforward to  show  that 
 $C_{-,\tau,L}$,  cf.  \eqref{defofC-tau},  can be bounded uniformly in $L \geq 1$
and 
 $C_{+,E,\tau,L}$,  cf. \eqref{boundonIntegralsmall-3},  can be  bounded uniformly in $L \geq 1$  for $E$ in compact subsets of $[0,\infty)$. 
 For details we refer the reader to    \cite{HaslerKoberstein.2025-2}, where analogous estimates have been explicitly carried out. 
\end{proof}

\section{Statements}

The authors have no relevant financial or non-financial interests to disclose.
No funding was received to assist with the preparation of this manuscript.
Data sharing not applicable to this article as no datasets were generated or analysed during the current study.

\printbibliography


\end{document}